\newif\ifpublic
\publictrue
\ifpublic
\documentclass[11pt,a4]{article}
\else
\documentclass[11pt,USenglish,letterpaper]{article}
\fi
\PassOptionsToPackage{hyphens}{url}

\usepackage[margin=2.5cm]{geometry}
\usepackage[final,pagebackref,colorlinks=true,linkcolor=blue,urlcolor=blue,citecolor=blue,pdfstartview=FitH,breaklinks=true]{hyperref}
\usepackage[T1]{fontenc}    
\usepackage{arydshln}
\usepackage{amsmath}       
\usepackage{amssymb}
\usepackage{amsthm}
\usepackage{amsbsy}
\usepackage{nicefrac}       
\usepackage{lipsum}
\usepackage{fullpage}
\usepackage{mathpazo}
\usepackage{mathtools}
\usepackage{enumitem}
\usepackage{subfiles}
\usepackage{xcolor}
\usepackage{setspace}
\usepackage{algorithm}%
\usepackage{algorithmic}%
\usepackage{graphicx}
\graphicspath{ {./img/} }
\usepackage{textcomp}
\usepackage{pifont}
\usepackage{comment}
\usepackage{subfig}
\usepackage[numbers,sort&compress]{natbib}

\usepackage{alltt}
\usepackage{boxedminipage}

\usepackage[capitalize,nameinlink]{cleveref}
\crefname{algocf}{alg.}{algs.}
\Crefname{algocf}{Algorithm}{Algorithms}
\renewcommand{\eqref}[1]{\hyperref[#1]{(\ref*{#1})}}

\usepackage{pgfplots}
\usepackage{appendix}
\pgfplotsset{compat=1.14}
\theoremstyle{plain}
\newtheorem{theorem}{Theorem}[section]
\newtheorem{corollary}[theorem]{Corollary}
\newtheorem{lemma}[theorem]{Lemma}

\newtheorem{definition}[theorem]{Definition}
\newtheorem{claim}[theorem]{Claim}

\theoremstyle{definition}
\newtheorem{remark}[theorem]{Remark}

\renewcommand{\epsilon}{\varepsilon}

\renewcommand{\phi}{\varphi}

\newcommand{\cmark}{\ding{51}}%
\newcommand{\xmark}{\ding{55}}%
\newcommand{\aradhya}[1]{{\color{purple} #1}}

\newcommand{\pg}[1]{\noindent\textbf{#1}~~}



\ExplSyntaxOn
\NewDocumentCommand \subtitle { m }{
  \tl_put_right:cn { @title } 
    { \\[.5ex] \LARGE #1 }
}
\ExplSyntaxOff

\begin{document}
\title{Towards Learning-Augmented Peer-to-Peer Networks: Self-Stabilizing Graph Linearization with Untrusted Advice}
\author{
  Vijeth Aradhya\thanks{National University of Singapore, Singapore.}
   \and
   Christian Scheideler\thanks{Paderborn University, Germany.} 
}
\date{}

\maketitle

\begin{abstract}
Distributed peer-to-peer systems are widely popular due to their decentralized nature, which ensures that no peer is critical for the functionality of the system. However, fully decentralized solutions are usually much harder to design, and tend to have a much higher overhead compared to centralized approaches, where the peers are connected to a powerful server. On the other hand, centralized approaches have a single point of failure. Thus, is there some way to combine their advantages without inheriting their disadvantages? To that end, we consider a supervised peer-to-peer approach where the peers can ask a \emph{potentially unreliable} supervisor for \emph{advice}. This is in line with the increasingly popular algorithmic paradigm called \emph{algorithms with predictions} or \emph{learning-augmented algorithms}, but we are the first to consider it in the context of peer-to-peer networks.

Specifically, we design \emph{self-stabilizing} algorithms for the fundamental problem of distributed \emph{graph linearization}, where peers are supposed to recover the ``sorted line'' network from any initial network after a transient fault. With the help of the supervisor, peers can recover the sorted line network in $O(\log n)$ time, if the advice is correct; otherwise, the algorithm retains its original recovery time (i.e., without any supervisor). A crucial challenge that we overcome is to correctly compose multiple self-stabilizing algorithms, that is, one that processes and exploits the advice, and another  that does not rely on the advice at all. Our key technical contributions combine ideas from the fields of overlay networks and proof-labeling schemes. Finally, we give a matching lower bound of $\Omega(\log n)$ for the recovery time of any algorithm if the advice can be corrupted, where $n$ is the network size.
\vspace{5mm}

{\footnotesize \textbf{Keywords}: distributed algorithms, overlay networks, self-stabilization, proof-labeling schemes}

\end{abstract}

\newpage
\tableofcontents
\newpage

\section{Introduction}

Traditionally, online algorithms augmented with advice is an area of research where one attempts to measure how
much knowledge of the future is necessary to achieve a given competitive ratio \cite{EmekFKR09, BoyarFKLM17}. However, the advice might not always be helpful since it is usually hard to predict the future sufficiently well. Ideally, we would like that bad advice does not result in a competitive ratio that is worse than what is achieved without any advice. Motivated by recent advances in machine learning, a new line of research has thus looked at \emph{algorithms with predictions} \cite{MitzenmacherV22}. In this framework \cite{MitzenmacherV20}, the algorithm takes advantage of an extra piece of information, called ``prediction'' (also called ``advice'' or ``hint''), to achieve (near) optimal performance if the prediction is good, while also not performing (much) worse than an algorithm without prediction, if the prediction is bad. 

While competitive analysis in the sequential online setting has been predominant for modeling \emph{temporal} uncertainty (i.e., the input is arriving piece by piece, and decisions need to be made on the fly), problems in distributed networks typically arise from \emph{spatial} uncertainty where the input itself is divided among several compute nodes, and all nodes collectively perform global computation \cite{PapadimitriouY93, IraniR96, Suomela13, AkbariELMSS23}. To surmount the uncertainty, many researchers have investigated trade-offs between the size of advice and the network's performance for many problems (e.g., \cite{GavoillePPR04, FraigniaudGIP09, FraigniaudIP10, FraigniaudKL10, KormanKP10, EllenGMP21}), seeking to \emph{minimize} the information needed to efficiently solve them. Recently, however, there has been a growing trend towards designing distributed algorithms that can not only exploit the advice for better performance, but are also \emph{robust} (i.e., retain their native performance without any advice) when the advice is unreliable \cite{GilbertNVW21, ImMINFOCOM23, AddankiP024, boyarRL25}.

We initiate the study of algorithms with \emph{untrusted} advice for peer-to-peer systems, where the peers continuously execute \emph{local} algorithms to maintain a desirable network with certain connectivity properties. As these systems are highly dynamic (for e.g., due to churn, link failures, memory corruption, etc), such algorithms should quickly \emph{recover} the network after any failure. To that end, we consider the \emph{self-stabilization} framework \cite{dijkstra74} that requires the system to converge to a desirable state after any transient fault which can arbitrarily corrupt the system state. Following the pioneering work of Dijkstra \cite{dijkstra74}, self-stabilization has been promising for modeling recovery in distributed systems (e.g., \cite{SS-AwerbuchKMPV93, SS-DolevW04, SS-BurmanK07, SS-ManneMPT09, SS-JacobRSST14, SS-BlinT18, SS-BurmanCCDNSX21}); for books, see \cite{Dolev2000, SS-tixeuil2009self}.

Specifically, we consider a network model where there is a \emph{supervisor} that can be contacted by the peers if they need advice. Supervised networks have been considered in the past (e.g., \cite{KothapalliS05, FriedmanOPODIS13, FeldmannKSS18, Sup-AfekGP24}), but there only the case has been studied that the advice is reliable. In that scenario, the peers can trust all the connections suggested by the supervisor, which would make the self-stabilization problem for overlay networks very easy: simply suggest the set of peers some peer should ideally be connected to. But if the advice can be corrupted, for example, an advice to connect to non-existing peers, then a malicious supervisor could easily start so-called \emph{Sybil attacks} \cite{Douceur02}, i.e., attacks integrating many potentially malicious peers into the peer-to-peer (P2P) system, that were previously not part of it.

To prevent such attacks, the advice of the supervisor must be limited to peers that a peer is already connected to. Indeed, this scenario has been recently considered in \cite{scheideler2019complexity}, but that work only considered how quickly a well-behaved supervisor can compute a near-optimal schedule for transforming any initial overlay network into the desired overlay network, given that it can only suggest network transformations by pointing to connections that the peers currently have. Bad advice can be hard for peers to identify since they are only aware of their local neighborhood. It is crucial that the advice does not throw the stabilization process backwards so that the peers never converge to the desired network. Further, a crucial challenge is that self-stabilizing algorithms must converge from any initial state, including arbitrarily bad advice from the supervisor.

\begin{figure*}[t]
    \centering \includegraphics[width=\textwidth]{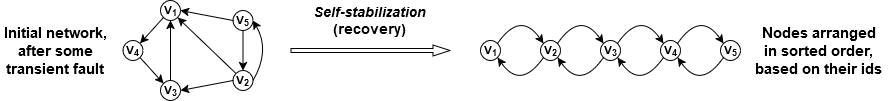}
    \caption{Illustration for \emph{Graph Linearization}. In peer-to-peer networks, there is a directed edge $(u, v)$ if a node $u$ has the \emph{id} of a node $v$ (i.e., node $u$ can send messages to node $v$). Since the id of a node can also be sent in a message, the network can \emph{change} over time. Here, the nodes recover the sorted path: $(v_1, v_2, \dots, v_5)$.}
    \label{fig:illustration-recovery}
\end{figure*}   

Despite the obstacles, we develop a generic self-stabilizing mechanism for solving the fundamental problem of distributed \emph{graph linearization} (topological sorting) \cite{onus2007linearization, Lin-JacobRSS12, Lin-NorNT13, Gall-ThCompSys14, ScheidelerSS15}, where each peer connects to its adjacent peer(s) in the topological sorted order to form the so-called sorted line network, no matter what kind of advice is given. This is achieved by a modular design of separately executing lightweight self-stabilizing algorithms that \emph{process} the advice, and any self-stabilizing algorithm that converges to the sorted line network \emph{without} any advice.

Our advice messages are of \emph{small} size (per peer), and enable the peers to quickly build a \emph{hypercubic} structure that facilitates efficient information exchange. Using that structure, a peer is able to connect to any other peer in the network. Crucially, the advice is also \emph{verifiable} in the sense that if the advice is bad, the peers can quickly detect it and stop making any additional connections due to that advice. Indeed, an added advantage of being able to detect bad advice in the self-stabilization setting is that the peers can quickly return to a state where they are again ready to receive (correct) advice from the supervisor.

\pg{Organization.} First, we describe our contributions, related work, and model. In Section \ref{sec:prelims}, we state a few preliminaries. In Section \ref{sec:advice}, we describe all the different aspects of the advice. In Section \ref{sec:flyover}, we provide the key details of the self-stabilizing processes. Finally, we give the analysis sketch in Section \ref{sec:analysis-sketch}. All the algorithms and their full analysis are provided in the Appendix. 
\subsection{Our Contributions}
We take the first steps towards efficiently and robustly augmenting advice for efficient recovery of P2P networks. Specifically, we provide self-stabilizing algorithms for \emph{graph linearization}, which has been a foundation (i.e., an intermediate step) for many self-stabilizing networks (see, a recent survey of \cite{feldmann2020survey}). Our solution achieves an optimal self-stabilization time if the supervisor is honest; otherwise, it retains the self-stabilization time of any solution without supervision.

To that end, we develop a novel and generic approach in which the supervisor needs to send only one $O(\log n)$-bit sized advice message to any peer, where $n$ denotes the network size. Surprisingly, despite this \emph{small} advice, peers can quickly build a \emph{hypercubic} structure that can be used to \emph{securely} create new connections among peers. Specifically, every peer connects to its neighbor(s) in the sorted line network in $O(\log n)$ rounds, if the advice messages are correct; otherwise, each peer rejects its advice in $O(\log n)$ rounds, and stops making any additional edge manipulations due to that advice. This is done in part by an interesting connection to the field of proof-labeling schemes \cite{KormanDC10}, and the following key technical ingredient (which may be of independent interest) for overlay networks.


\begin{theorem}[Informal, Tree-to-Path] \label{conf-theorem:tree-to-path}
There is a distributed algorithm for forming a path overlay from a (labelled) tree overlay in $O(1)$ rounds, whilst adding $O(1)$ edges per peer.
\end{theorem}


For robustness against a (potentially) malicious supervisor, the peers (simultaneously) run a ``base algorithm $\mathcal{A}_0$,'' which when run independently (i.e., without any supervisor) is known to self-stabilize in $\mathrm{R}(\mathcal{A}_0)$ rounds. When a peer receives an advice message, it executes a set of algorithms (cf. Section \ref{sec:flyover}), designed to \emph{expedite} the base algorithm, so that the network can quickly self-stabilize. Combined with our design of advice, we get the following upper bounds (based on the supervisor) for converging to the sorted line network.

\begin{theorem}[Upper Bound] \label{conf-theorem2:upper-bound}
The network self-stabilizes in $O(\log n)$ rounds if the supervisor is honest; otherwise, the network self-stabilizes in $O(\mathrm{R}(\mathcal{A}_0) + \log n)$ rounds.
\end{theorem}

Finally, we prove that this logarithmic convergence time barrier is unavoidable for any algorithm if the advice can be arbitrarily corrupted. In particular, since a malicious supervisor could potentially inject adversarial ids into the system \cite{Douceur02}, a peer must not make connections to any arbitrary peer id given in the advice. Thus, this constraint, that new edges can only be added to any peer's \emph{current} neighborhood, leads to the following lower bound (for converging to any network).

\begin{theorem}[Lower Bound] \label{conf-theorem:lower-bound}
There exists some initial configuration such that the convergence time for any pair of supervisor and overlay algorithms is $\Omega(\log n)$ rounds.
\end{theorem}

\subsection{Related Work}
Peer-to-peer (P2P) systems have been in the spotlight for several decades (see, e.g., \cite{RowstronD01, StoicaMKKB01, RatnasamyFHKS01, MalkhiNR02, MaymounkovM02}), due to their multifaceted architectures that give rise to different algorithmic challenges in many distributed systems, including but not limited to load-balancing (e.g., \cite{KargerR06, Pa-FelberKSS14}), overlay networks (e.g., \cite{PanduranganFOCS01, MaoDVKS20}), agreement (e.g., \cite{AugustinePODC13, AugustineJCSS15}), publish/subscribe systems (e.g., \cite{Pa-ChandF05, ChocklerMTV07, Pa-BianchiFG07}), Sybil defense (e.g., \cite{GuptaJCSS23, DaniITCS24}), etc.

In many P2P systems, a cornerstone problem is building and maintaining a desirable overlay network. Since these systems need to be highly resilient to failures, researchers have extensively worked on self-stabilizing algorithms for several networks; for e.g., variants of Skip and Chord graphs \cite{KniesburgesSPAA11, jacob2014skip+, Berns21}, hypertrees and radix trees \cite{DolevK08, AspnesW07}, small-world graph \cite{kniesburges2012self}, metric graph \cite{Gmyr-ThCompSys19}, etc; for a comprehensive survey, see \cite{feldmann2020survey}.

Graph linearization is the main \emph{building block} for many of the above self-stabilizing networks. Specifically, a linearization algorithm (e.g. \cite{onus2007linearization}) forms an \emph{ordering} of peers (e.g. based on their ids), whose structure is then exploited to build other components of the desired network. To the best of our knowledge, the state-of-the-art algorithm \cite{jacob2014skip+} that can recover the sorted line network runs in $O(\log^2 n)$ rounds. Notably, the Transitive Closure Framework (TCF) \cite{BernsGP13} is a generic approach for the peers to recover any (locally-checkable) network in $O(\log n)$ rounds, but it is not scalable as it requires the peers to first form a clique network (i.e., each peer's degree necessarily increases to $\Omega(n)$), and then converge to the desired network.

Recently, the algorithmic framework of using predictions \cite{MitzenmacherV22} has gained immense popularity in several fields of algorithms (e.g., \cite{KraskaSIGMOD18, Mitzenmacher18, BhaskaraC0P20, JiangICLR20, LattanziSODA20, LykourisJACM21}) including distributed and networked algorithms, for e.g., contention resolution \cite{GilbertNVW21}, TCP acknowledgment \cite{ImMINFOCOM23}, buffer management \cite{AddankiP024} and graph algorithms \cite{boyarRL25}. These algorithms are given an extra string of bits called ``prediction'' to solve the problem. Typically, the prediction would give information about some ``structure'' of the given input. However, the main challenge is to design \emph{robust} algorithms (incl. the prediction) that can retain their original performance \emph{regardless} of what the prediction is, whilst performing much better if the prediction is helpful.

\subsection{Network Model \& Problem Definition} \label{sec:model}
\pg{Entities.} The \emph{system} has two entities: nodes and a supervisor.

\emph{Nodes} are computing entities that aim to distributively form a \emph{target topology}, denoted by $G^*$, through their connections, where we say that two nodes are ``connected'' if they can communicate (i.e., send/receive messages) with each other. We colloquially use the word ``network'' (formally defined soon) to capture the set of nodes and their connections. The number of nodes is denoted by $n$. Furthermore, the nodes do not have any global knowledge (e.g., network size, shared clock, etc). The nodes always execute the algorithms prescribed to them.

A \emph{supervisor} is a single computing entity, designed to help the nodes form the target topology $G^*$ through communication with the nodes to learn the (current) network, compute over it, and propose connections between nodes. If the supervisor follows its algorithm, we say that the supervisor is \emph{honest}; otherwise, it is said to be \emph{malicious}. There is no restriction on the messages sent by a malicious supervisor, which may even include suggestions for connecting to non-existing nodes.

\pg{Node ids.} Each node has a unique \emph{identifier} (\emph{id}). Each node's memory consists of two \emph{types} of variables: address and non-address variables. The address variables are exclusively used to store the ids of other nodes (referred to as \emph{neighbors}), whereas the non-address variables cannot be used to store node ids.

\pg{Communication Rounds.} The system proceeds in \emph{synchronous} rounds, i.e., a message sent by a node in any round is received by its recipient at the beginning of the next round. In every round, each node runs the given algorithms, given its variables' assignments (and messages in the channel), resulting in changes to its memory and new messages being sent. 

\pg{Comm. among Nodes.} Each node $u$ can send a message to another node $v$ if the node $u$ has node $v$'s id in its memory. If a message is sent to any node $u$, then it is added to node $u$'s \emph{channel} $u.\mathit{Ch}$. Once a message is added to any node's channel, the node is able to read the contents of the message, after which the message is removed from the channel. In general, a node is not informed about the (id of the) sender of a message, unless the sender explicitly adds its own id to the message.

\pg{States and Configurations.} Consider the standard \emph{state} model for self-stabilization \cite{dijkstra74}. The variables and channel of a node are stored in its \emph{register} (aka, mutable memory) and each node has read/write access to its register. Each node also has a non-mutable memory that contains its id and the code of our algorithms. The set of values assigned to the variables and channel of a node is called the \emph{state} of a node. Moreover, the tuple of states of all nodes is called the \emph{configuration}, and the (simultaneous) execution of the code by all the nodes, over consecutive rounds, forms a \emph{sequence} of configurations.

As we consider a distributed message-passing system where the communication links can themselves change over time, each node needs to \emph{explicitly} communicate some information about its own state to (some of) node ids stored in its register, for the network to reach the target topology (see, \cite{jacob2014skip+, feldmann2020survey}), unlike in the standard model \cite{dijkstra74}, where every node reads its own register and the registers of all its neighbors, executes the algorithms and updates its register, in any given step.

\pg{Communication Graph and Network.} Let $G_C = (V, \mathcal{E}_C = E_C \cup I_C)$ be the \emph{communication graph} over the set of nodes $V$ in configuration $C$, where there exists an \emph{explicit} edge $(u, v) \in E_C$ if a node $v$'s id is stored in any address variable of node $u$'s memory, or an \emph{implicit} edge $(u, v) \in I_C$ if a node $v$'s id is in node $u$'s channel. The \emph{network} is simply the communication graph restricted to explicit edges.


\pg{Comm. between Supervisor and Nodes.} The supervisor is aware of the node membership in the system, i.e., ids of all the nodes. It can send a request to any node to learn the (current) communication graph. Moreover, the supervisor can send \emph{advice} messages to help nodes quickly converge to the desired target topology. As is typical, we assume that the supervisor cannot forge messages sent by nodes. 

\pg{Functions, Rules and Actions.} Each node runs a sequence of \emph{algorithms}, and each algorithm consists of a sequence of \emph{functions}. A function has the form, $\langle \text{Label} \rangle\text{: } \langle \text{Guard} \rangle \rightarrow \langle \text{Rule 1}; \text{Rule 2}; \dots  \rangle $, where \emph{label} is the function name, \emph{guard} is a Boolean predicate over the node's variables, and \emph{rules} define the \emph{actions} by the node, where an action (typically) refers to either updating a local variable or sending a new message. Each rule is itself of the form $\langle \text{Guard} \rangle \rightarrow \langle \text{Action 1}; \text{Action 2}; \dots  \rangle $.

\emph{Pseudocodes.} Each line in a function is typically dedicated to a single rule, but if there are multiple rules in a line, then the last set of actions in that line belongs to the last rule. Moreover, a function may \emph{call} another function in any of its rules. (In the standard state model \cite{dijkstra74}, the algorithms are directly defined by a set of rules; we extend it for better readability.)

\emph{Local Computation.} In each round, each node checks the guard of every function in the specified order; if the guard is satisfied, the node executes the rules in the function. Similarly, a node checks the guard of a rule; if it is satisfied, the node proceeds to execute the actions in it. Whereas, if no guard is specified, the function or rule is always executed.

\subsection*{Problem Definition}

Given any (weakly) connected communication graph, the nodes must form (and stay in) the target topology. For graph linearization, the target topology is the network where each node is connected to its neighbor(s) in the \emph{sorted} order, based on their ids. As the nodes also need to process the advice given by the supervisor, we allow them to store a \emph{small} number of extra edges, which help in quickly linearizing the network.

\pg{Initial, Target and Legal Configurations.} Let $\mathcal{C}$ and $\mathcal{T} \subseteq \mathcal{C}$ be the set of \emph{all} possible configurations and the set of \emph{target} configurations, respectively. Let $\mathcal{I}$ be the set of \emph{initial} configurations where $\forall C \in \mathcal{C}, G_C \text{ is weakly connected} \iff C \in \mathcal{I}$. (We need only weak connectivity for the initial configuration, so the non-address variables can be arbitrarily corrupted.)

The ``sorted path'' network \cite{onus2007linearization} is the target topology, $G^* = (V, E^*)$, where $\forall u, v \in V, (u = \mathit{succ}(v) \lor v = \mathit{succ}(u)) \iff (u, v) \in E^*$, where $\mathit{succ}(u) = \min\left(\{ v \mid v.\mathit{id} > u.\mathit{id} \}\right)$. Thus, $\mathcal{T}$ is the set of configurations where the network is $G^*$.

Let $N_C(u) = \{ v \mid (u, v) \in E_C \}$ denote the set of all edges stored by a node $u$ in any configuration $C$. Let $\mathcal{L}$ be the set of \emph{legal} configurations such that for each node $u$, $\forall C' \in \mathcal{L}, \exists C \in \mathcal{T}, N_{C}(u) \subseteq N_{C'}(u)$ and $\forall C' \in \mathcal{L}, \forall C \in \mathcal{T}, |N_{C'}(u)| = O(|N_{C}(u)| + \log n)$. In other words, we allow every node to store at most \emph{logarithmic} number of extra edges (i.e., besides the edges in the sorted path) in a legal configuration.

\pg{Self-Stabilization.} We say that an algorithm is \emph{self-stabilizing} if the network, starting from any initial configuration, reaches a legal configuration (\emph{convergence}), and stays in one (\emph{closure}).
\section{Preliminaries} \label{sec:prelims}
\subsection{Dual-State Approach} \label{subsec:dual-state}
We rely on a \emph{dual-state} algorithm design, where the nodes always run a ``base algorithm'' $\mathcal{A}_0$ (in the background), and an algorithm $\mathcal{A}$ is executed for processing the advice given by the supervisor. Moreover, the base algorithm is known to self-stabilize in $\mathrm{R}(\mathcal{A}_0)$ rounds, without any interaction with the supervisor. The objective of algorithm $\mathcal{A}$ is to use the advice for introducing the neighbors in $G^*$ to each other in $O(\log n)$ rounds, whilst adding as few additional edges as possible.

We distinguish the memory (including channel) used for the algorithms $\mathcal{A}$ and $\mathcal{A}_0$ as $\mathrm{M}(\mathcal{A})$ and $\mathrm{M}(\mathcal{A}_0)$, respectively. Since these algorithms are run in tandem, to ensure connectivity for the base algorithm, we define a ``Flush'' operation, denoted by $\mathbf{Flush}(\cdot)$, that moves a given set of node ids to any (arbitrary) address variable(s) of $\mathrm{M}(\mathcal{A}_0)$; this definition is as general as possible so that our approach works for any base algorithm.
\subsection{Universal Overlay Primitives} \label{subsec:univ-prim}

Let us recall the four \emph{universal} communication primitives (for edge manipulations) in the design of distributed self-stabilizing overlay algorithms (see, e.g., \cite{scheideler2019complexity, feldmann2020survey}); see Figure \ref{fig:univ-prim}.

While preserving connectivity, they are used to transform any weakly connected graph to any other weakly connected graph.

\begin{figure*}[t]
    \centering \includegraphics[width=\textwidth]{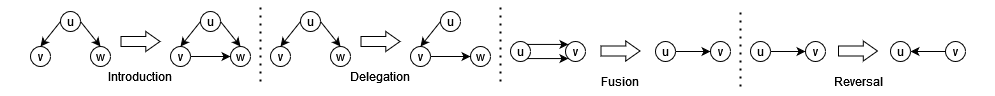}
    \caption{The universal communication primitives, used by our algorithms, incl. base algorithm.}
    \label{fig:univ-prim}
\end{figure*}

\subsection{Local Certification of Spanning Tree} \label{subsec:local-cert}

Our algorithm for the supervised setting exploits a classic result in distributed computing: local certification of spanning tree \cite{AfekWDAG90}. In general, in local certification, there exist a prover and a distributed verifier (\emph{fixed} network of nodes). The prover assigns a small \emph{certificate} to each node. Every node inspects its certificate and the certificates of its neighbors to \emph{verify} a given network property (e.g., spanning tree). If that property is specified correctly at each node, there exists a certificate assignment so that each node ``accepts'' its certificate; otherwise, regardless of what certificates are assigned by the prover, at least one node ``rejects'' its certificate. If the verification algorithm has access to the node id, that type of certification is called \emph{proof-labeling scheme} \cite{KormanDC10}. For more details, we refer to a recent survey \cite{Feuilloley-DMTCS21}.


Specifically, for each node, a certificate for a \emph{spanning tree}, with any node chosen as the root, consists of the information: id of the root, id of the parent (in that tree), and distance of the node from the root. To locally verify the encoded spanning tree, each node carries out \emph{two} checks: (1) every neighboring node has the \emph{same} root id, and (2) its distance is \emph{one} more than the distance of its parent (called ``distance-check''); and the root node checks that its distance is 0.

\section{Sybil Resistant, Compact and Robust Advice} \label{sec:advice}
Here, we provide the key intuition and details of the advice.

\subsection{Network Snapshot and Properties of Advice}

First, if the supervisor does not know anything about the \emph{current} communication graph, there is no hope of giving useful advice because the supervisor should be restricted to a node's neighborhood whenever an edge is proposed. However, as P2P systems are dynamic, the communication graph could have changed. Thus, there should be some mechanism where the nodes can send their neighborhoods when requested by the supervisor. Once the neighborhood is sent, a node can maintain it for $O(1)$ rounds, while continuing to run the base algorithm. (The details on this interaction with the supervisor are given in Section \ref{subsec:distrib-TtP}.) For simplicity, we assume that the supervisor has access to a \emph{snapshot} of the network, on which the advice can be computed, and that the snapshot is maintained by the nodes. 

Given any network snapshot, the supervisor can compute a (close to) minimum sequence of the universal primitives (cf. Section \ref{subsec:univ-prim}), so that the network forms the target topology (see \cite{scheideler2019complexity} for details). However, since the nodes do not have any global information (e.g., round number, etc), they would need to fully trust the supervisor for proposing appropriate edges to eventually form the target topology. This can be problematic for the convergence time as well as the number of edges stored by the nodes over time. Thus, there should be some mechanism by which a node can \emph{quickly verify} whether the proposed edges are indeed helpful, whilst also limiting the number of additional edges stored due to the supervisor.

\pg{Desirable Properties.} Based on these insights, we propose that the advice message at any given node has the following properties to achieve both efficiency and robustness.
\begin{enumerate}
    \item \emph{Sybil-Resistant.} For each node $u$, and for any node id $v$ present in node $u$'s advice message but not present in node $u$'s memory or channel, node $u$ neither assigns $v$ to any variable (i.e. store in its memory), nor sends $v$ to any other node (i.e., circulate it in the network).
    \item \emph{Compact.} Size of the advice is small, i.e., a message of size at most $O(\log n)$ bits is given to each node.
    \item \emph{Robust.} If any node received an incorrect advice, every node rejects its advice (if that node had received advice), and restores the dedicated memory for handling advice (i.e., memory $\mathrm{M}(\mathcal{A})$), in $O(\log n)$ rounds; whereas, if every node received correct advice, the nodes will be in a legal configuration in $O(\log n)$ rounds.
\end{enumerate}

\subsection{Building a Global Structure from any Snapshot} \label{subsec:global-struct}

Given the properties of advice, we describe two immediate challenges towards designing the advice. First, we should be able to connect every pair of nodes that are adjacent in the sorted-path topology but may be situated \emph{far away} from each other in any given snapshot, in $O(\log n)$ rounds. Second, if any node received an incorrect advice, the other nodes should be informed about it in $O(\log n)$ rounds by a \emph{scalable} mechanism so that they can also quickly reject their advice.

One can observe that the above challenges arise from a \emph{lack} of a \emph{global structure} in the network. Thus, our first attempt is to build a lightweight structure, specifically a \emph{hypercube}, from the advice given to the nodes. Building a hypercubic overlay can be used to ``kill two birds with one stone'': first, a node can rely on a \emph{local} routing algorithm, where a message from a source is forwarded to the next node and so on until it reaches the destination, to initiate a \emph{new} connection (e.g., for forming the sorted path); second, if a node detects any fault with the advice, it can \emph{quickly} inform the other nodes in the hypercube, after which the nodes can leave the hypercubic structure.

\begin{figure*}[t]
    \centering \includegraphics[width=1\textwidth]{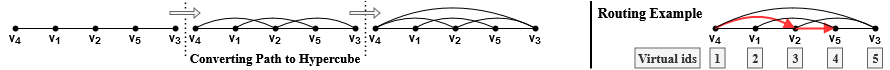}
    \caption{Five nodes hierarchically introduce their neighbors to each other to form exponentially spaced-out \emph{shortcuts}. If $v_4$ needs to send a message $m$ to the node with virtual id $4$ (i.e., $v_5$), $m$ is first forwarded to the ``closest'' shortcut (wrt. \emph{virtual ids}), i.e., $v_2$; this is repeatedly done until $m$ eventually reaches $v_5$.}
    \label{fig:pointer-doubling}
\end{figure*}

Now, we enlist the steps for building a hypercubic structure.
\begin{enumerate}
    \item \emph{Snapshot to Tree.} Supervisor chooses a rooted spanning tree in the snapshot. As part of the advice, each node is given the id of its parent and the distance to the root.
    \item \emph{Tree to Path.} Each node performs a \emph{local transformation} (see Section \ref{subsec:tree-to-path}, for the Tree-to-Path algorithm), so that all the nodes are part of a \emph{single path}. As part of the advice, nodes are given unique ``virtual ids,'' going from 1 to $n$, from one end of that path to the other.
    \item \emph{Path to Hypercube.} By a pointer-doubling approach \cite{JaJa92} (over that path), each node hierarchically introduces its neighbors to each other. Thus, in $O(\log n)$ rounds, every node connects to ``shortcuts,'' at distances of $1, 2, 4, \dots$ from it, forming what we call a ``flyover'' which is a hypercubic structure\footnote{Note that a hypercube is technically defined for a network size that is a power of 2. In general, flyover belongs to a family of hypercubic networks, and that path and collection of shortcut edges contain the hypercube edges.}; for e.g., see Figure \ref{fig:pointer-doubling}.
    \item \emph{Routing.} Once the nodes form a flyover, every node can route messages to any other node \cite{leighton2014introduction} (c.f. Figure \ref{fig:pointer-doubling}), using the \emph{virtual ids} and \emph{shortcuts}, in $O(\log n)$ rounds. 
\end{enumerate}

At this point, an astute reader may note that multiple flyovers could form in different parts of the network, whereas some nodes may not even have received any advice. Such \emph{concurrency} issues are exactly what makes this problem of handling a malicious supervisor difficult. In Section \ref{subsec:c-and-v-flyover}, we outline how this flyover is carefully constructed and verified, and how small information (i.e., a unique identifier for flyover) is exchanged between the nodes (within and outside the flyover), to quickly detect faulty advice.

\subsection{Locally Verifying the Advised sorted Path} \label{subsec:advised-sorted-path}


After we have a desriable global structure (i.e., all the nodes forming a flyover over a sequentially numbered path), one natural approach is that for each node, as part of its advice, send the virtual id (i.e., position on that numbered path) of its neighbor(s) in the target topology. Each node can then quickly route its \emph{own id} over the hypercubic structure to the node with that virtual id to create a new connection, thereby forming the ``advised'' target topology. But as the nodes have a \emph{local} view, they cannot easily \emph{verify} if the actual target topology was indeed formed. 

For example, let the sorted path be $(u_1, u_2, \dots, u_n)$, where $u_1$ is the node with the smallest node id. If the advice for node $u_i$ consists of the virtual ids of nodes $u_{i-1}$ and $u_{i+1}$, the node $u_i$ can send its own node id to them via ``hypercubic routing'' (by virtual ids and shortcuts; e.g., see Figure \ref{fig:pointer-doubling}) \cite{leighton2014introduction}, and make a connection to them. Once these edges are added, each node can \emph{locally verify} that it has at most 2 edges, and if it has 2 edges, it is the ``middle'' node id (in sorted order). But a bad supervisor can give advice so that the paths $(u_1, \dots, u_{n/2})$ and $(u_{n/2}, \dots, u_{n})$ are formed, without any node detecting any fault with the advice. In other words, the ``advised'' sorted path (i.e., connections suggested by the supervisor via virtual ids) do not form a \emph{single connected} component. Being able to (distributively) detect if any two components of the overlay are connected, is a non-trivial task with the best-known technique (i.e., linear-probing), requiring $O(n)$ rounds \cite{feldmann2020survey}.

\begin{algorithm}
\caption{Supervisor Algorithm}
\label{alg:cloud-alg}
\begin{algorithmic}[1]
\REQUIRE Snapshot (connected and undirected graph) $G_\mathcal{S} = (V, E_\mathcal{S})$.
\STATE Let $\mathit{Adv}(u)$ denote the advice message sent to a node $u$. Moreover, if $\mathit{Adv}(u).x = y$, then ``$x = y$'' appears in the advice message for the node $u$; Let $\mathrm{dist}_T(u, v)$ denote the distance between any nodes $u$ and $v$ in a tree $T$;
\STATE \COMMENT{Relevant graphs (definitions)}
\STATE Let $T_{\mathcal{S}}$ be a labelled rooted spanning tree in $G_\mathcal{S}$ with any node $r$ as root, and $l(\cdot)\text{: } V \rightarrow \{0,1\}$ as the labeling where $l(r) = 0 \land (\forall v \in V, l(v) = l(r) + \mathrm{dist}_{T_{\mathcal{S}}}(v, r)\mod 2)$; 
\STATE Let $P_{\mathcal{S}}$ be (undirected version of) the path returned by Tree-to-Path algorithm (see, Section \ref{subsec:tree-to-path}) over input $T_{\mathcal{S}}$; 
\STATE Let $G^* = (V, E^*)$ be the (undirected) sorted path where  $\forall u, v \in V, u = \mathit{succ}(v) \lor v = \mathit{succ}(u) \iff (u, v) \in E^*$;
\STATE Let $T^*$ be the rooted spanning tree, called the \emph{sorted-path tree}, in the sorted path $G^*$ with node $r$ as the root;
\FOR{each node $u \in V$}
\STATE \COMMENT{Virtual id (i.e., position on the path; cf. Section \ref{subsec:global-struct})}
\STATE $\mathit{Adv}(u).\mathit{vID} \coloneqq$ $\mathrm{dist}_{P_{\mathcal{S}}}(u, r) + 1$;
\STATE \COMMENT{To verify the ``advised'' sorted path (cf. Section \ref{subsec:advised-sorted-path})}
\STATE $u \neq r \rightarrow \mathit{Adv}(u).\text{\emph{c-par}} \coloneqq$ $\mathrm{dist}_{P_{\mathcal{S}}}(u', r) + 1$, where $u'$ is the parent of $u$ in $T^*$; \COMMENT{Virtual id of $u$'s parent in $T^*$}
\STATE $\mathit{Adv}(u).\text{\emph{c-dist}} \coloneqq \mathrm{dist}_{T^*}(u, r)$; 
\STATE \COMMENT{For the Tree-to-Path Algorithm (cf. Section \ref{subsec:tree-to-path})}
\STATE $u \neq r \rightarrow \mathit{Adv}(u).\mathit{par} \coloneqq$ node id of $u$'s parent in $T_{\mathcal{S}}$;
\STATE $\mathit{Adv}(u).\mathit{dist} \coloneqq$ $\mathrm{dist}_{T_\mathcal{S}}(u, r)$; 
\ENDFOR
\end{algorithmic}
\end{algorithm}

\begin{figure*}[t]
    \centering \includegraphics[width=\textwidth]{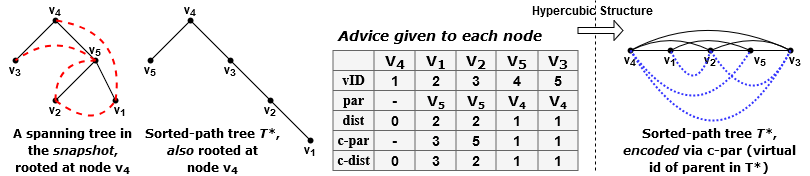}
    \caption{An example of the advice and its outcome (cf. Algorithm \ref{alg:cloud-alg}), where the sorted path $G^*$ is $(v_1, \dots, v_5)$. Here, the (red) dashed edges form the output of the Tree-to-Path Algorithm, and the (blue) dotted edges depict the encoded tree $T^*$. Note that those dotted edges are formed via routing over the flyover (hypercubic structure).}
    \label{fig:example5node}
\end{figure*}

To fix such \emph{connectivity-related} issues in the advised sorted path, we make use of \emph{local certification} of a spanning tree (cf. Section \ref{subsec:local-cert}). Specifically, the sorted path $G^*$ is \emph{encoded} as a rooted spanning tree called ``sorted-path tree $T^*$,'' where the root is the node whose virtual id is 1. As part of advice, each node is given the virtual id of its parent, and the distance to the root, in the tree $T^*$. See Algorithm \ref{alg:cloud-alg}, and Figure \ref{fig:example5node} for a pictorial example. In local certification, the verification of the certificates is carried out over a \emph{fixed} network, but in our case, the nodes in $T^*$ may not even be connected. Thus, each node uses the hypercubic structure for routing and connecting to its parent in $T^*$. Moreover, that structure is also used to quickly disseminate the node id of the root in $T^*$ to all the nodes. In Section \ref{subsec:conn-cert}, we explain how the formation and verification of the tree $T^*$ is accomplished via routing over the hypercubic structure.

\subsection{Tree-to-Path Algorithm} \label{subsec:tree-to-path}

Let $T = (l(\cdot), r, V, E)$ be a labelled rooted undirected tree, where $V$ is the set of vertices, $r \in V$ is the root, $l(\cdot)\text{: }V \rightarrow \{ 0, 1\}$ is a \emph{label} function on vertices where $l(v) = (l(r) + \mathrm{dist}(r, v)) \mod 2$, and $E$ is the set of edges. Let $C_T(v) = \{ u \mid \mathrm{dist}(u, v) = 1 \wedge \mathrm{dist}(v, r) < \mathrm{dist}(u, r) \}$ be the \emph{children} of $v \in V$ in the tree $T$, and there is a \emph{total order} (for e.g., based on ids) on $C_T(v) = \{ u_1, \dots, u_q\}$ such that $u_1 < \dots < u_q$ for any $q \geq 1$. Let $\mathrm{Beg}(P)$ and $\mathrm{End}(P)$ of a directed path $P$ of size at least 2, denote the vertex with no incoming edge and the vertex with no outgoing edge, respectively.


We briefly describe the \emph{Tree-to-Path} algorithm that takes a labelled rooted tree as input, and outputs a directed path. See Figure \ref{fig:TtP} for the pictorial description. Specifically, the algorithm takes in a tree $T = (l(\cdot), r, V, E)$ as input, and outputs a directed path $P$ where $\mathrm{Beg}(P) = r$ and $\mathrm{End}(P) = \min(C_T(r))$ if $l(r) = 0$, whereas it outputs a directed path $P$ where $\mathrm{Beg}(P) = \max(C_T(r))$ and $\mathrm{End}(P) = r$ if $l(r) = 1$. We incorporate this \emph{alternation} based on labels for maintaining the invariant that $\mathrm{Beg}(P)$ and $\mathrm{End}(P)$ are at a distance of 1. Such a simple, fast and local algorithm is amenable to be made self-stabilizing. In Section \ref{subsec:distrib-TtP}, we give the self-stabilizing version of the algorithm that converts any (connected) overlay to a path overlay, using the given advice, in $O(1)$ rounds.

\begin{figure*}[t]
    \centering \includegraphics[width=0.9\textwidth]{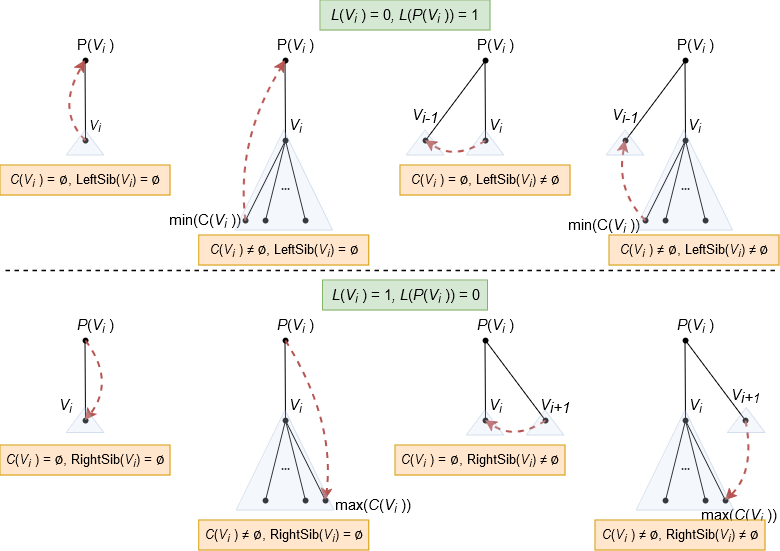}
    \caption{Cases in Tree-to-Path Algorithm for any node $v_{i}$ where $i$ is its order (e.g., via total order on ids) in set $C(P(v_i))$, i.e., $\mathrm{LeftSib}(v_i) = \{ v_j \mid j < i \}$ and $\mathrm{RightSib}(v_i) = \{ v_j \mid j > i\}$. $L(v),P(v)$ and $C(v)$ are label, parent and children of a node $v$ (resp.). Dashed arrow depicts an edge in the final path.}
    \label{fig:TtP}
\end{figure*}
\section{Flyovers: Fast Stabilization via Routing} \label{sec:flyover}

We rely on the dual-state approach (cf. Section \ref{subsec:dual-state}), where a node executes a set of algorithms (for handling the advice) for helping the base algorithm to quickly form the sorted path.

In this section, we give the intuition and pseudocodes for \emph{three} self-stabilizing processes, that are concurrently executed to help the nodes form the sorted path in $O(\log n)$ rounds, if the advice is correct; otherwise, every node ignores the advice in $O(\log n)$ rounds, whilst still executing the base algorithm.  Before delving into the details, we outline the design of our solution, and also describe the roles of all the variables for it.

\subsection{Variables, High-Level Design, and Initialization}

We refer to the hypercubic structure that a node belongs to as a \emph{flyover}, as the nodes ``bypass'' the current network to find ``short paths'' to other nodes. Recall that a flyover is built on some path (over all the nodes) formed due to the advice (cf. Section \ref{subsec:tree-to-path}).

\pg{Variables.} First, we explain all the variables of any node $u$ that make up a flyover.
\begin{enumerate}
    \item $u.L$ and $u.R$ denote the sets of left and right \emph{shortcuts} (resp.); recall, shortcuts are node ids that are at distance $1, 2, 4, \dots$ on the path. $u.S = u.L \cup u.R$ is the set of all shortcuts. Moreover, we refer to $u.S_l(i)$ and $u.S_r(i)$ as the ``$i$-level'' left and right shortcut (resp.).
    \item $u.\mathit{vID}$ is used for storing the \emph{virtual id}; recall, the virtual id is provided in the advice, which is used to determine its \emph{position} on the path. (This is not to be confused with the unique id of a node, see Section \ref{sec:model}; when we just refer to ``id,'' we mean the unique id of the node.) By design, the node with \emph{no} left shortcut has virtual id 1. 
    \item $u.\mathit{flyID}$ is used for storing a ``unique identifier'' for the flyover, called the \emph{flyover id}. By design, the node id of the node with \emph{no} left shortcut, determines the flyover id.
    \item $u.\mathit{exit}$ is a binary variable indicating if the node should \emph{exit} the flyover. If any node exits the flyover, it runs only the base algorithm (until it receives a new advice).
\end{enumerate}

We also use three variables for storing and verifying what we refer to as \emph{connectivity certificate}, based on the local certification of spanning tree; these variables are used for certifying the sorted-path tree $T^*$. (Algorithm \ref{alg:cloud-alg} has the definition of $T^*$.)
\begin{enumerate}
    \item $u.\textit{c-par}$ is used for storing the \emph{virtual id} of the parent of node $u$ in the sorted-path tree $T^*$.
    \item $u.\textit{c-dist}$ is used for storing the \emph{distance} from the root to node $u$ in the sorted-path tree $T^*$.
    \item $u.\textit{c-ids}$ is used for storing the \emph{node ids} of node $u$'s neighbors in the sorted-path topology $G^*$.
\end{enumerate}

\pg{Our Design.} A node $u$ updates the variables $u.\textit{vID}, u.\textit{c-par}$ and $u.\textit{c-dist}$ (if it's not part of a flyover) when it receives the advice. (See Algorithm \ref{alg:cloud-alg} for how those values are determined.) The Tree-to-Path algorithm is then executed, and based on its position, the \emph{first} left and right shortcuts, i.e., $u.S_l(1)$ and $u.S_r(1)$, are updated to its left and right neighbors on the path (resp.). Subsequently, the nodes build the flyover over the path via pointer-doubling, where the $(i+1)$-level shortcuts are formed from the $i$-level shortcuts; see Figure \ref{fig:pointer-doubling}. By design, the node id of the ``leftmost'' node (in the path), say, node $v$ with $v.L = \emptyset \land v.R \neq \emptyset$, is the \emph{flyover id}, and is propagated to the rest of the nodes in flyover, i.e., for any node $u$, $u.\mathit{flyID}$ is initially assigned its own node id (as a default value), but is eventually assigned the flyover id. Moreover, node $v$ is the \emph{only} node with $v.\mathit{vID} = 1$ and $v.\textit{c-dist} = 0$, as it is the root in the sorted-path tree $T^*$. See Figure \ref{fig:example5node} for an illustration.

To verify the connectivity certificate, i.e., execute distance-checks (cf. Section \ref{subsec:advised-sorted-path}), each node $u \neq v$ must send the distance $u.\textit{c-dist}$ to the node with virtual id $u.\textit{c-par}$, say node $w$, on the flyover. However, since node $u$ may not have a direct connection to node $w$, it sends $u.\textit{c-dist}$ to the \emph{closest} left or right shortcut (depending on the virtual ids, i.e., $u.\textit{vID}$ and $u.\textit{c-par}$), and so on, until this message eventually reaches node $w$ whose virtual id is $u.\textit{c-par}$. This is a standard method for routing a message between any pair of nodes in a hypercubic structure \cite{leighton2014introduction}, on which we further provide details below. Finally, once the node $w$ receives this message, it adds node id $u$ to $w.\textit{c-ids}$ if the distances match, i.e., $w.\textit{c-dist} = u.\textit{c-dist} - 1$ (to form the sorted path). Moreover, every node $u$ checks that it is sorted with respect to node ids in $u.\textit{c-ids}$ (to verify the sorted path).

\begin{algorithm}[ht]
\caption{Init}
\label{alg:conf-init}
\begin{algorithmic}[1]
\STATE $\mathbf{BasicChecks}$
\begin{ALC@g}
\STATE $(\exists \langle \text{RejFlyover} \rangle \in u.\mathit{Ch}) \rightarrow u.\mathit{exit} \coloneqq 1$;
\STATE \COMMENT{Flyover structure}
\STATE $u.S = \emptyset \land (u.\textit{c-ids} \neq \emptyset \lor u.\textit{flyID} \neq u.\mathit{id}) \rightarrow u.\mathit{exit} \coloneqq 1$;
\STATE $(u.L \neq \emptyset \land u.\mathit{vID} \leq 1) \rightarrow u.\mathit{exit} \coloneqq 1$;
\STATE $(u.L = \emptyset \land u.R \neq \emptyset) \land (u.\mathit{vID} \neq 1 \lor u.\mathit{flyID} \neq u.id) \rightarrow  u.\mathit{exit} \coloneqq 1$;
\STATE \COMMENT{Connectivity certificate} 
\STATE $(u.\mathit{vID} = 1 \land u.\text{\emph{c-dist}} \neq 0) \rightarrow u.\mathit{exit} \coloneqq 1$; 
\STATE $(u.\mathit{vID} > 1 \land u.\text{\emph{c-dist}} \leq 0) \rightarrow u.\mathit{exit} \coloneqq 1$;
\STATE $(u.S \neq \emptyset \land u.\mathit{vID} > 1) \land (\mathbf{NextStop}(u.\text{\emph{c-par}}) = \varnothing) \rightarrow u.\mathit{exit} \coloneqq 1$; 
\STATE \COMMENT{Sorted-path topology}
\STATE $|u.\text{\emph{c-ids}}| > 2 \rightarrow u.\mathit{exit} \coloneqq 1$;
\STATE $|u.\text{\emph{c-ids}}| = 2 \land (u.\mathit{id} > \max(u.\text{\emph{c-ids}})) \rightarrow u.\mathit{exit} \coloneqq 1$;
\STATE $|u.\text{\emph{c-ids}}| = 2 \land (u.\mathit{id} < \min(u.\text{\emph{c-ids}})) \rightarrow u.\mathit{exit} \coloneqq 1$;
\end{ALC@g}
\STATE
\STATE $\mathbf{RejectFlyover} \text{: } u.\mathit{exit} = 1 \rightarrow$
\begin{ALC@g}
\STATE \COMMENT{Inform and flush the node ids part of node $u$'s flyover}
\STATE $\mathbf{SendRejFly}(( u.S \cup \{u.\mathit{flyID}\} \cup u.\text{\emph{c-ids}} ) \setminus u.\mathit{id})$;
\STATE \COMMENT{Set the variables to their default values}
\STATE $u.L \coloneqq \emptyset; u.R \coloneqq \emptyset; u.\mathit{vID} \coloneqq 0; u.\mathit{flyID} \coloneqq u.\mathit{id};$
\STATE $u.\mathit{exit} \coloneqq 0; u.\text{\emph{c-par}} \coloneqq 0; u.\text{\emph{c-dist}} \coloneqq -1; u.\text{\emph{c-ids}} \coloneqq \emptyset;$
\end{ALC@g}
\end{algorithmic}
\end{algorithm}

\begin{algorithm}
\caption{Helper Functions \textit{\small (Executed Only When Called)}}
\begin{algorithmic}[1]
\label{alg:conf-helper}
\STATE $\mathbf{NextStop(\mathit{val})}$
\begin{ALC@g}
\STATE $(\mathit{val} < 1 \lor \mathit{val} = u.\mathit{vID})\rightarrow $ Return ``$\varnothing$'';
\STATE $(u.S = \emptyset \lor u.\mathit{vID} < 1) \rightarrow $ Return ``$\varnothing$'';
\STATE $(\mathit{val} > u.\mathit{vID} \land u.R = \emptyset) \rightarrow $ Return ``$\varnothing$'';
\STATE $(\mathit{val} < u.\mathit{vID} \land u.L = \emptyset) \rightarrow $ Return ``$\varnothing$'';
\IF{$(\mathit{val} > u.\mathit{vID})$}
\STATE Return $u.S_r(\mathrm{argmin}_{i \in |u.R|}(\lvert u.\mathit{vID} + 2^{(i-1)} - \mathit{vID}\rvert))$;
\ELSE
\STATE Return $u.S_l(\mathrm{argmin}_{i \in |u.L|}(\lvert u.\mathit{vID} - 2^{(i-1)} - \mathit{vID}\rvert))$;
\ENDIF
\end{ALC@g}
\STATE 
\STATE $\mathbf{SendRejFly(\mathit{nodes})}$
\begin{ALC@g}
\STATE SEND($ \langle \text{RejFlyover} \rangle$) to ids in $\mathit{nodes}$; $\mathbf{Flush}(\mathit{nodes})$;
\end{ALC@g}
\end{algorithmic}
\end{algorithm}

\pg{Init Function.} Our design results in a few simple tests that are executed in every round; see Algorithm \ref{alg:conf-init}. Here, the $\mathbf{NextStop}(x)$ function returns the $i$-level left or right shortcut, whose virtual id, $u.\mathit{vID} \pm 2^{i}$, is closest to $x$; if it is not possible to do so (e.g., $u.\mathit{vID} < x$ and a right shortcut does not exist), then the function returns ``$\varnothing$''. see Algorithm \ref{alg:conf-helper} for its description. Observe that if any of the tests fail, $u.\mathit{exit}$ is updated to 1. In that case, node $u$ rejects the flyover, i.e., it sends $\langle \text{RejFlyover} \rangle$ message to all the node ids, \emph{flushes} those ids (see Section \ref{subsec:dual-state} for the ``flush'' operation), and sets all the variables to the default values.

\subsection{Construction and Verification of Flyover} \label{subsec:c-and-v-flyover}

Here, we explain the first two self-stabilizing processes concerned with the flyover structure.  As the flyover is \emph{symmetrically} built over some path, for readability, we provide the explanations with respect to right shortcuts (if they exist), which similarly apply for left shortcuts. Recall that if a node $v$ is a $i$-level (right) shortcut of node $u$, then node $v$ should be of distance $2^i$ to (the right of) node $u$ on the path.

Our algorithms consist of a precise set of conditions for both, sending a particular message, and for specific actions to be taken (for e.g., update the set of shortcuts, or set the exit variable to 1) after receiving a message. Moreover, if any node $u$ is not part of any flyover (i.e., $u.S = \emptyset$), and receives a message regarding a flyover, then node $u$ sends $\langle \text{RejFlyover} \rangle$ message to the node ids part of that message, and flushes them.

Algorithm \ref{alg:conf-hypercube} provides the sequence of rules that govern the \emph{construction} of the flyover.

A node $u$ sends its own node id to its first right shortcut, $u.S_r(1)$, using the $\langle \text{TestLine-R} \rangle$ message, to verify whether that information is correct. Consequently, the node that receives the message, sends $\langle \text{RejFlyover} \rangle$ message to node $u$, and flushes it, if node $u$ is not its left neighbor.

A node $u$ sends its $i$-level left shortcut $u.S_l(i)$, including its own id, to its $i$-level right shortcut $u.S_r(i)$, using $\langle \text{FlyConst-R} \rangle$ message, so that $u.S_r(i)$ can add or verify its $(i+1)$-level left shortcut. This set of rules facilitates a pointer-doubling style, bottom-up flyover construction over the path (regardless of any shortcuts that may already exist). Consequently, the node that receives the message, say node $w$, sets $w.S_l(i+1) = u.S_l(i)$ only if $|w.L| = i$ and $w.S_l(i) = u$; otherwise, it sends $\langle \text{RejFlyover} \rangle$ to nodes $u$ and $u.S_l(i)$, and flushes them.

\begin{algorithm}[ht]
\caption{Flyover Construction}
\label{alg:conf-hypercube}
\begin{algorithmic}[1]
\STATE $\mathbf{TestFlyoverConstruction}$
\begin{ALC@g}
\STATE $u.R \neq \emptyset \rightarrow$ SEND($\langle \text{TestLine-R}, u.\mathit{id} \rangle$) to $u.S_r(1)$;
\STATE $u.L \neq \emptyset \rightarrow$ SEND($\langle \text{TestLine-L}, u.\mathit{id} \rangle$) to $u.S_l(1)$;
\IF{$(u.R \neq \emptyset \land u.L \neq \emptyset)$}
\FOR{each $i$ in $\{1, \dots,  \min(|u.R|, |u.L|)\}$}
\STATE SEND($\langle \text{FlyConst-R}, u.S_l(i), i, u.\mathit{id} \rangle$) to $u.S_r(i)$;
\STATE SEND($\langle \text{FlyConst-L}, u.S_r(i), i, u.\mathit{id} \rangle$) to $u.S_l(i)$;
\ENDFOR
\ENDIF
\end{ALC@g}
\STATE
\STATE \COMMENT{Here, we omit code for a ``\dots-L'' msg due to similarity in handling a ``\dots-R'' msg}
\STATE $\mathbf{R\_TestFlyoverConstruction}$
\begin{ALC@g}
\FOR{each received message $\langle \text{TestLine-R},\mathit{sen} \rangle$}
\STATE $(u.L \neq \emptyset \land u.S_l(1) \neq \mathit{sen} \lor u.L = \emptyset) \rightarrow u.\mathit{exit} \coloneqq 1$;
\STATE $u.S = \emptyset \lor u.\mathit{exit} = 1  \rightarrow \mathbf{SendRejFly}(\{\mathit{sen}\});$ 
\ENDFOR
\FOR{each received message $\langle \text{FlyConst-R}, w, i, \mathit{sen} \rangle$}
\STATE $u.L = \emptyset \lor (|u.L| \geq i \land u.S_l(i) \neq \mathit{sen}) \rightarrow u.\mathit{exit} \coloneqq 1$;
\STATE $|u.L| \geq (i+1) \land u.S_l(i+1) \neq w \rightarrow u.\mathit{exit} \coloneqq 1$;
\STATE $(u.\mathit{exit} = 0) \land (1 < |u.L| < i) \rightarrow \mathbf{Flush}(\{\mathit{sen}, w\})$; 
\STATE $(u.\mathit{exit} = 0) \land (|u.L| = i \land u.S_l(i) = \mathit{sen}) \rightarrow u.S_l(|u.L| + 1) \coloneqq w$; \COMMENT{Update shortcut!}
\STATE $u.S = \emptyset \lor u.\mathit{exit} = 1  \rightarrow \mathbf{SendRejFly}(\{\mathit{sen, w}\});$
\ENDFOR
\end{ALC@g}
\end{algorithmic}
\end{algorithm}

Algorithm \ref{alg:conf-flymetadata} gives the sequence of rules for the \emph{verification} of virtual id and flyover id.

A node $u$ sends the value $2^{(i-1)} + u.\mathit{vID}$ to node $u.S_r(i)$, using $\langle \text{TestvID} \rangle$ message, to verify the distance of its $i$-level right shortcut (over the path that the flyover is constructed on). Consequently, the node that receives the message, say node $w$, sets $w.\mathit{exit} = 1$ if its virtual id does not match the value.

A node $u$ sends its $u.\mathit{flyID}$ (flyover id), using $\langle \text{TestFlyID} \rangle$ message, to all node ids, except $u.\mathit{flyID}$, in memory (incl. $\mathrm{M}_u(\mathcal{A}_0)$, memory of base algorithm), if either node $u$ is the leftmost node, or $u.\mathit{flyID}$ has been updated (i.e., $u.\mathit{flyID} \neq u.\mathit{id}$). Furthermore, a node $u$, when it is ready to receive advice from the supervisor (i.e., $u.S = \emptyset \land u.\mathit{vID} = 0$), sends ``$\bot$'' to all the node ids using $\langle \text{TestFlyID} \rangle$ message. Crucially, this \emph{exchange} of information (i.e., flyover id) acts as a simple mechanism for handling a supervisor that concurrently sends bad advice to nodes in different parts of the network, in that it helps the nodes quickly detect if there is a flyover constructed over any \emph{subset} of nodes. Consequently, the node that receives the message, say node $w$, sets $w.\mathit{flyID}$ to $u.\mathit{flyID}$, only if it is not the leftmost node, and $w.\mathit{flyID}$ has not been updated yet (i.e., $w.\mathit{flyID}$ is assigned to $w.\mathit{id}$). Finally, if $w.\mathit{flyID}$ does not match $u.\mathit{flyID}$, node $w$ sets $w.\mathit{exit} = 1$, and sends $\langle \text{RejFlyover} \rangle$ message to node $u.\mathit{flyID}$, and flushes it.

\begin{algorithm}[ht]
\caption{Flyover Metadata}
\label{alg:conf-flymetadata}
\begin{algorithmic}[1]
\STATE $\mathbf{TestFlyoverMetadata}$
\begin{ALC@g}
\STATE \COMMENT{Propagate the flyover id if it is updated}
\STATE $\textit{prop-flyID} \coloneqq (u.\mathit{vID} = 1) \lor (u.\mathit{vID} > 1 \land u.\mathit{flyID} \neq u.\mathit{id})$;
\STATE $\textit{prop-flyID} \rightarrow$ SEND($\langle \text{TestFlyID}, u.\mathit{flyID} \rangle$) to all ids, except $u.\mathit{flyID}$, in memory;
\STATE \COMMENT{Ready to receive advice; inform other nodes about it}
\STATE $(u.S = \emptyset \land u.\mathit{vID} = 0) \rightarrow$ SEND($\langle \text{TestFlyID}, \bot \rangle$) to all ids in memory;
\STATE \COMMENT{Check that shortcuts are exponentially spaced-out}
\STATE $|u.R| \geq 1 \rightarrow$ SEND($\langle \text{TestvID},(u.\mathit{vID} + 2^{(i-1)})\rangle$) to $u.S_r(i)$ for all $i \in [|u.R|]$;
\STATE $|u.L| \geq 1 \rightarrow$ SEND($\langle \text{TestvID},(u.\mathit{vID} - 2^{(i-1)})\rangle$) to $u.S_l(i)$ for all $i \in [|u.L|]$;
\end{ALC@g}
\STATE 
\STATE $\mathbf{R\_TestFlyoverMetadata}$
\begin{ALC@g}
\FOR{each received message $\langle \text{TestvID}, \mathit{vID} \rangle$}
\STATE $(u.S = \emptyset) \lor (u.S \neq \emptyset \land  u.\mathit{vID} \neq \mathit{vID}) \rightarrow u.\mathit{exit} = 1$;
\ENDFOR
\FOR{each received message $\langle \text{TestFlyID}, \mathit{flyID} \rangle$}
\IF{$(u.L \neq \emptyset \land u.\mathit{exit} = 0)$}
\STATE $u.\mathit{flyID} = u.\mathit{id} \land \mathit{flyID} \neq \bot \rightarrow u.\mathit{flyID} \coloneqq \mathit{flyID}$; \COMMENT{Update flyID!}
\ENDIF
\STATE $(u.S \neq \emptyset \land u.\mathit{flyID} \neq \mathit{flyID}) \rightarrow u.\mathit{exit} = 1$; 
\STATE $(u.S = \emptyset \land \mathit{flyID} \neq \bot) \rightarrow u.\mathit{exit} = 1$;
\STATE $u.\mathit{exit} = 1 \rightarrow  \mathbf{SendRejFly}(\{\mathit{flyID}\});$
\ENDFOR
\end{ALC@g}
\end{algorithmic}
\end{algorithm}

\subsection{Verification of Connectivity Certificate} \label{subsec:conn-cert}

Here, we explain the self-stabilizing process for verifying and forming the sorted-path topology.

Algorithm \ref{alg:conf-cert} provides the sequence of rules for \emph{routing} and \emph{verifying} connectivity certificates.

First, each node $u$ begins this process after its flyover id $u.\mathit{flyID}$ has been updated (i.e., $u.\mathit{flyID} \neq u.\mathit{id}$). Recall that each node $u$ sends $u.\mathit{flyID}$, via $\langle \text{TestFlyID} \rangle$ message, to all node ids in memory, to ensure that a \emph{single} flyover exists. Thus, the node id of the root in tree $T^*$ (i.e., node with virtual id 1; see Algorithm \ref{alg:cloud-alg}) gets stored and verified by each node.

A node $u$ needs to send $u.\textit{c-dist}$ (i.e., distance to the root in $T^*$), to the node with virtual id $u.\textit{c-par}$, say node $v$. For doing so, the node $u$ \emph{routes} both $u.\textit{c-dist}$ and $u.\mathit{id}$, using $\langle \text{TestCert} \rangle$ message, via an appropriate left or right shortcut (based on the virtual ids), and so on, until the message reaches node $v$. Recall, $\mathbf{NextStop}(x)$ returns the $i$-level left or right shortcut, whose virtual id, $u.\mathit{vID} \pm 2^{i}$, is closest to $x$. In a hypercubic structure, a message can be routed this way in $O(\log n)$ rounds, as the distance to the intended destination is halved each time.

Consequently, the node that receives the message, say node $w \neq v$, assigns $w.\mathit{exit} = 1$, and sends $\langle \text{RejFlyover} \rangle$ message to node $u$, if the message cannot be further forwarded (i.e., $\mathbf{NextStop}(u.\textit{c-par})$ function returns ``$\varnothing$''). Finally, when the node $v$ receives it, if the virtual ids and distances match (i.e., $v.\mathit{vID} = u.\textit{c-par}$ and $v.\textit{c-dist} = u.\textit{c-dist} - 1$), node $v$ adds the node id $u$ to $v.\textit{c-ids}$, and sends its own node id to node $u$, so that node $u$ stores it in $u.\textit{c-ids}$; otherwise, the node $v$ sets $v.\mathit{exit} = 1$, and sends $\langle \text{RejFlyover} \rangle$ to node $u$, and flushes it.

\begin{algorithm}[ht]
\caption{Connectivity Certificate}
\label{alg:conf-cert}
\begin{algorithmic}[1]
\STATE $\mathbf{TestConnCert}\text{: } (u.\mathit{vID} > 1 \land u.\mathit{flyID} \neq u.\mathit{id}) \rightarrow$
\begin{ALC@g}
\STATE SEND($\mathit{msg}$) to $\mathbf{NextStop}(u.\text{\emph{c-par}})$ where $\mathit{msg}$ is $\langle \text{TestCert}, u.\mathit{id}, u.\text{\emph{c-par}}, u.\text{\emph{c-dist}}\rangle$; \COMMENT{Route}
\end{ALC@g}
\STATE 
\STATE $\mathbf{R\_TestConnCert}$
\begin{ALC@g}
\STATE $\textit{prop-flyID} \coloneqq (u.\mathit{vID} = 1) \lor (u.\mathit{vID} > 1 \land u.\mathit{flyID} \neq u.\mathit{id});$
\FOR{each received message $\langle \text{TestCert}, w, \mathit{vID}, \mathit{dist}\rangle$}
\STATE \COMMENT{Distances don't match / Unable to route}
\STATE $(u.\mathit{vID} = \mathit{vID}) \land (\mathit{dist}-1 \neq u.\textit{c-dist}) \rightarrow u.\mathit{exit \coloneqq} 1$;
\STATE $(u.\mathit{vID} \neq \mathit{vID}) \land (\mathbf{NextStop}(\mathit{vID}) = \varnothing) \rightarrow u.\mathit{exit \coloneqq} 1$;
\IF{$(u.S = \emptyset \lor u.\mathit{exit} = 1)$}
\STATE $\mathbf{SendRejFly}(\{\mathit{w}\})$;
\ELSE
\IF{$u.\mathit{vID} = \mathit{vID}$}
\STATE $u.\textit{c-ids} \coloneqq u.\textit{c-ids} \cup \{ w \}$; \COMMENT{Accept id!}
\STATE SEND($\langle \text{IntroCert}, u.\mathit{id} \rangle$) to $w$; \COMMENT{Bidirected edge}
\ELSE
\STATE $\neg \textit{prop-flyID} \rightarrow \mathbf{Flush}(\{w\})$; 
\STATE $\textit{prop-flyID} \rightarrow$ SEND($\langle \text{TestCert}, w, \mathit{vID}, \mathit{dist}\rangle$) to $\mathbf{NextStop}(\mathit{vID})$; \COMMENT{Route, if flyID is updated}
\ENDIF
\ENDIF
\ENDFOR 
\FOR{each received message $\langle \text{IntroCert}, w \rangle$}
\STATE $u.\textit{c-ids} \coloneqq u.\textit{c-ids} \cup \{ w \}$; \COMMENT{Bidirected edge}
\ENDFOR 
\end{ALC@g}
\end{algorithmic}
\end{algorithm}




\section{Analysis Sketch} \label{sec:analysis-sketch}

In this section, we describe the analysis sketch for all the critical aspects of our algorithms. The pseudocodes are given in Section \ref{sec:pseudocodes}, and the full analyses are given in Section \ref{sec:TtP-analysis} and \ref{sec:ss-analysis}. 

Firstly, we show that the edges added by the Tree-to-Path algorithm (c.f. Figure \ref{fig:TtP}; see Algorithm \ref{alg:tree-to-line} for the pseudocode.) form a single (directed) path over the nodes. 
We use a proof by induction over the height of a labbelled rooted tree. We consider two base cases for the trees of height equal to 1, with the label of the root equal to either 0 or 1, where the height of a tree $T = (l(\cdot), r, V, E)$ is equal to $\max_{v\in V}(\mathrm{dist}_T(r, v))$. For the inductive case, we consider any tree $T = (l(\cdot), r, V, E)$ of height greater than 1, and then show that the algorithm outputs a directed path $P$ where $\mathrm{Beg}(P) = r$ and $\mathrm{End}(P)$ is $\min(C_T(r))$ if $l(r) = 0$; whereas it outputs a directed path $P$ where $\mathrm{Beg}(P)$ is $\max(C_T(r))$ and $\mathrm{End}(P) = r$ if $l(r) = 1$. 
\begin{theorem}
The Tree-to-Path Algorithm takes a labelled rooted tree $T = (l(\cdot), r, V, E)$ where $l(r) = 0$ as input, and outputs a directed path $P = (V, E')$ where $\mathrm{Beg}(P) = r$ and $\mathrm{End}(P) = \min(C_T(r))$.
\end{theorem}

Next, we show that the construction of shortcuts over that path happens in $O(\log n)$ rounds, implying that information (e.g., $\langle \text{RejFlyover} \rangle$ messages) can be quickly disseminated in any flyover. Thus, if the flyover is not correctly built, i.e., virtual ids are not consecutive, or shortcuts are not exponentially spaced-out, or flyover id is not properly assigned, then due to Algorithms \ref{alg:conf-init}, \ref{alg:conf-hypercube} and \ref{alg:conf-flymetadata}, some node detects and propagates the fault in the flyover, cueing other nodes to also exit the flyover.

However, if a flyover is ``internally correct,'' and is formed over a subset of nodes, then some key information needs to be exchanged between the nodes inside and outside the flyover. To that end, we rely on Algorithm \ref{alg:conf-flymetadata} where every node sends its flyover id to all node ids in the memory. As all our algorithms rely on the universal primitives (cf. Section \ref{subsec:univ-prim}) that preserve (weak) connectivity, there must exist an (incoming/outgoing) edge to that flyover in any round. By carefully carrying out a case-by-case analysis of how that edge can show up in a node's memory, we prove the following lemma.

\begin{lemma}
Consider any set of nodes $B = \{ v_1, \dots, v_{|B|} \}$ in any round $r$, where $|B| < n$ and $(v_1.L = \emptyset \land v_{|B|}.R = \emptyset)$ and $(\forall i \in [|B| - 1], v_{i}.S_r(1) = v_{i+1} \land v_{i+1}.S_l(1) = v_{i})$. By round $r + O(\log n)$, at least one node $u \in B$ assigns $u.\mathit{exit} = 1$.
\end{lemma}

Moreover, we prove that if all the nodes indeed form a \emph{single} flyover, but if the connections suggested by the supervisor (via virtual ids) do not form $G^*$ (see Section \ref{subsec:advised-sorted-path} for a discussion), then due to Algorithms \ref{alg:conf-init} and \ref{alg:conf-cert} (i.e., routing the connectivity certificates), at least one node quickly detects the bad advice.

\begin{claim}
Let $V = \{ v_1, \dots, v_n \}$ and $(v_1.L = \emptyset \land v_{n}.R = \emptyset)$ and $(\forall i \in [n - 1], v_{i}.S_r(1) = v_{i+1} \land v_{i+1}.S_l(1) = v_{i})$ and $(\forall i \in [n], v_i.\mathit{vID} = i \land v_i.\mathit{flyID} = v_1)$. Let $G' = (V, E')$ such that $(u, v) \in E' \iff (u.\mathit{vID} = v.\textit{c-par} \land v.\textit{c-dist} = u.\textit{c-dist} + 1) \lor (v.\mathit{vID} = u.\textit{c-par} \land v.\textit{c-dist} = u.\textit{c-dist} - 1)$. If $\forall i \in [n] \setminus \{ 1 \}, \exists u_i \in V, v_i.\textit{c-par} = u_i.\mathit{vID}$ where $(u_i.\textit{c-dist} = v_i.\textit{c-dist} - 1)$, then $G'$ is same as $G^*$.
\end{claim}
\begin{proof}[Proof Sketch]
Due to Algorithm \ref{alg:conf-init}, the following invariants are satisfied for each node $u \in V$.
\begin{enumerate}
    \item (Tree Distance) $(u.L = \emptyset \implies u.\textit{c-dist} = 0)$ and $(u.L \neq \emptyset \implies u.\textit{c-dist} > 0)$.
    \item (Degree Constraint; Locally Sorted) $E''(u) \neq \emptyset \implies (|E''(u)| = 1) \lor (|E''(u)| = 2 \land \min(|E''(u)|) < u.\mathit{id} < \max(|E''(u)|))$, where $E''(u) = \{ v \mid (u, v) \in E' \}$.
\end{enumerate}
The above invariants ensure that there can be only one node $u$ with $u.\textit{c-dist} = 0$, and that node is $v_1$. Moreover, the degree of any node in $G'$ is at most 2; if a node's degree is 2, then its own id is neither greater nor less than both the node ids of its neighbors. Given the premise of the lemma, for any node $u \in V \setminus \{ v_1 \}$, there is a node $p(u)$ such that $u.\textit{c-par} = p(u).\mathit{vID}$ and $p(u).\textit{c-dist} = u.\textit{c-dist} - 1$.

We use proof by contradiction to show that $G'$ is indeed $G^*$. First, there cannot be a cycle in $G'$, as the distance checks (i.e., from a node $u$ to its parent $p(u)$) cannot be satisfied at all the nodes (cf. Section \ref{subsec:local-cert}). Moreover, $G'$ cannot be a spanning forest because all the nodes, via flyover id, can verify that there is only one ``root node,'' which is node $v_1$ with its distance $v_1.\textit{c-dist} = 0$. Thus, $G'$ forms a spanning tree. Combining the second invariant mentioned earlier, i.e., each node has degree at most 2, the tree must be a path. Finally, that path has to be the sorted path $G^*$ due to the second invariant because every node checks that its locally sorted with its neighbors in $G'$.
\end{proof}
Finally, we give a lower bound on the convergence time given that the nodes are restricted to make edge manipulations in their current neighborhood. The key observation is that if the distance between any two nodes $u$ and $v$ is $D$ in round $r$, the distance between them is at least $D/2^t$ in round $r+t$, even if all nodes introduce all the node ids in their memory to each other in every round. Thus, if the nodes $u$ and $v$ should be adjacent in the target topology, it takes $\Omega(\log n)$ rounds for node $u$ to obtain the node id $v$ (or vice versa) in the worst-case.

\begin{theorem}
For any Sybil-resistant network, there exists an initial configuration such that the convergence time for any pair of supervisor and overlay algorithms is $\Omega(\log n)$ rounds.
\end{theorem}







\bibliographystyle{plainurl}
\bibliography{main}

\begin{thebibliography}{10}

\bibitem{AddankiP024}
Vamsi Addanki, Maciej Pacut, and Stefan Schmid.
\newblock Credence: Augmenting datacenter switch buffer sharing with {ML} predictions.
\newblock In {\em Proc. {NSDI}}, pages 613--634, 2024.

\bibitem{Sup-AfekGP24}
Yehuda Afek, Gal Giladi, and Boaz Patt{-}Shamir.
\newblock Distributed computing with the cloud.
\newblock {\em Distributed Comput.}, 37(1):1--18, 2024.
\newblock \href {https://doi.org/10.1007/S00446-024-00460-W} {\path{doi:10.1007/S00446-024-00460-W}}.

\bibitem{AfekWDAG90}
Yehuda Afek, Shay Kutten, and Moti Yung.
\newblock Memory-efficient self stabilizing protocols for general networks.
\newblock In {\em Proc. {WDAG}}, pages 15--28, 1990.
\newblock \href {https://doi.org/10.1007/3-540-54099-7\_2} {\path{doi:10.1007/3-540-54099-7\_2}}.

\bibitem{AkbariELMSS23}
Amirreza Akbari, Navid Eslami, Henrik Lievonen, Darya Melnyk, Joona S{\"{a}}rkij{\"{a}}rvi, and Jukka Suomela.
\newblock Locality in online, dynamic, sequential, and distributed graph algorithms.
\newblock In {\em Proc. {ICALP}}, pages 10:1--10:20, 2023.
\newblock \href {https://doi.org/10.4230/LIPICS.ICALP.2023.10} {\path{doi:10.4230/LIPICS.ICALP.2023.10}}.

\bibitem{AspnesW07}
James Aspnes and Yinghua Wu.
\newblock O(logn)-time overlay network construction from graphs with out-degree 1.
\newblock In {\em Proc. {OPODIS}}, pages 286--300, 2007.
\newblock \href {https://doi.org/10.1007/978-3-540-77096-1\_21} {\path{doi:10.1007/978-3-540-77096-1\_21}}.

\bibitem{AugustinePODC13}
John Augustine, Gopal Pandurangan, and Peter Robinson.
\newblock Fast byzantine agreement in dynamic networks.
\newblock In {\em Proc. {PODC}}, pages 74--83, 2013.
\newblock \href {https://doi.org/10.1145/2484239.2484275} {\path{doi:10.1145/2484239.2484275}}.

\bibitem{AugustineJCSS15}
John Augustine, Gopal Pandurangan, Peter Robinson, and Eli Upfal.
\newblock Distributed agreement in dynamic peer-to-peer networks.
\newblock {\em J. Comput. Syst. Sci.}, 81(7):1088--1109, 2015.
\newblock \href {https://doi.org/10.1016/J.JCSS.2014.10.005} {\path{doi:10.1016/J.JCSS.2014.10.005}}.

\bibitem{SS-AwerbuchKMPV93}
Baruch Awerbuch, Shay Kutten, Yishay Mansour, Boaz Patt{-}Shamir, and George Varghese.
\newblock Time optimal self-stabilizing synchronization.
\newblock In {\em Proc. {STOC}}, pages 652--661, 1993.
\newblock \href {https://doi.org/10.1145/167088.167256} {\path{doi:10.1145/167088.167256}}.

\bibitem{Berns21}
Andrew Berns.
\newblock Network scaffolding for efficient stabilization of the chord overlay network.
\newblock In {\em {Proc. SPAA}}, pages 417--419, 2021.
\newblock \href {https://doi.org/10.1145/3409964.3461827} {\path{doi:10.1145/3409964.3461827}}.

\bibitem{BernsGP13}
Andrew Berns, Sukumar Ghosh, and Sriram~V. Pemmaraju.
\newblock Building self-stabilizing overlay networks with the transitive closure framework.
\newblock {\em Theor. Comput. Sci.}, 512:2--14, 2013.
\newblock \href {https://doi.org/10.1016/J.TCS.2013.02.021} {\path{doi:10.1016/J.TCS.2013.02.021}}.

\bibitem{BhaskaraC0P20}
Aditya Bhaskara, Ashok Cutkosky, Ravi Kumar, and Manish Purohit.
\newblock Online linear optimization with many hints.
\newblock In {\em Proc. {NeurIPS}}, 2020.

\bibitem{Pa-BianchiFG07}
Silvia Bianchi, Pascal Felber, and Maria Gradinariu.
\newblock Content-based publish/subscribe using distributed r-trees.
\newblock In {\em Proc. {E}uro-{P}ar}, pages 537--548, 2007.
\newblock \href {https://doi.org/10.1007/978-3-540-74466-5\_57} {\path{doi:10.1007/978-3-540-74466-5\_57}}.

\bibitem{SS-BlinT18}
L{\'{e}}lia Blin and S{\'{e}}bastien Tixeuil.
\newblock Compact deterministic self-stabilizing leader election on a ring: the exponential advantage of being talkative.
\newblock {\em Distributed Comput.}, 31(2):139--166, 2018.
\newblock \href {https://doi.org/10.1007/S00446-017-0294-2} {\path{doi:10.1007/S00446-017-0294-2}}.

\bibitem{boyarRL25}
Joan Boyar, Faith Ellen, and Kim~S. Larsen.
\newblock Distributed graph algorithms with predictions.
\newblock {\em CoRR}, abs/2501.05267, 2025.
\newblock \href {https://arxiv.org/abs/2501.05267} {\path{arXiv:2501.05267}}, \href {https://doi.org/10.48550/ARXIV.2501.05267} {\path{doi:10.48550/ARXIV.2501.05267}}.

\bibitem{BoyarFKLM17}
Joan Boyar, Lene~M. Favrholdt, Christian Kudahl, Kim~S. Larsen, and Jesper~W. Mikkelsen.
\newblock Online algorithms with advice: {A} survey.
\newblock {\em {ACM} Comput. Surv.}, 50(2):19:1--19:34, 2017.
\newblock \href {https://doi.org/10.1145/3056461} {\path{doi:10.1145/3056461}}.

\bibitem{SS-BurmanCCDNSX21}
Janna Burman, Ho{-}Lin Chen, Hsueh{-}Ping Chen, David Doty, Thomas Nowak, Eric~E. Severson, and Chuan Xu.
\newblock Time-optimal self-stabilizing leader election in population protocols.
\newblock In {\em Proc. {PODC}}, pages 33--44, 2021.
\newblock \href {https://doi.org/10.1145/3465084.3467898} {\path{doi:10.1145/3465084.3467898}}.

\bibitem{SS-BurmanK07}
Janna Burman and Shay Kutten.
\newblock Time optimal asynchronous self-stabilizing spanning tree.
\newblock In {\em Proc. {DISC}}, pages 92--107, 2007.
\newblock \href {https://doi.org/10.1007/978-3-540-75142-7\_10} {\path{doi:10.1007/978-3-540-75142-7\_10}}.

\bibitem{Pa-ChandF05}
Rapha{\"{e}}l Chand and Pascal Felber.
\newblock Semantic peer-to-peer overlays for publish/subscribe networks.
\newblock In {\em Proc. {E}uro-{P}ar}, pages 1194--1204, 2005.
\newblock \href {https://doi.org/10.1007/11549468\_130} {\path{doi:10.1007/11549468\_130}}.

\bibitem{ChocklerMTV07}
Gregory~V. Chockler, Roie Melamed, Yoav Tock, and Roman Vitenberg.
\newblock Constructing scalable overlays for pub-sub with many topics.
\newblock In {\em Proc. {PODC}}, pages 109--118, 2007.
\newblock \href {https://doi.org/10.1145/1281100.1281118} {\path{doi:10.1145/1281100.1281118}}.

\bibitem{DaniITCS24}
Varsha Dani, Thomas~P. Hayes, Seth Pettie, and Jared Saia.
\newblock Fraud detection for random walks.
\newblock In {\em Proc. {ITCS}}, pages 36:1--36:22, 2024.
\newblock \href {https://doi.org/10.4230/LIPICS.ITCS.2024.36} {\path{doi:10.4230/LIPICS.ITCS.2024.36}}.

\bibitem{dijkstra74}
Edsger~W. Dijkstra.
\newblock Self-stabilizing systems in spite of distributed control.
\newblock {\em Commun. {ACM}}, 17(11):643--644, 1974.
\newblock \href {https://doi.org/10.1145/361179.361202} {\path{doi:10.1145/361179.361202}}.

\bibitem{Dolev2000}
Shlomi Dolev.
\newblock {\em Self-Stabilization}.
\newblock {MIT} Press, 2000.

\bibitem{DolevK08}
Shlomi Dolev and Ronen~I. Kat.
\newblock Hypertree for self-stabilizing peer-to-peer systems.
\newblock {\em Distributed Comput.}, 20(5):375--388, 2008.
\newblock \href {https://doi.org/10.1007/S00446-007-0038-9} {\path{doi:10.1007/S00446-007-0038-9}}.

\bibitem{SS-DolevW04}
Shlomi Dolev and Jennifer~L. Welch.
\newblock Self-stabilizing clock synchronization in the presence of byzantine faults.
\newblock {\em J. {ACM}}, 51(5):780--799, 2004.
\newblock \href {https://doi.org/10.1145/1017460.1017463} {\path{doi:10.1145/1017460.1017463}}.

\bibitem{Douceur02}
John~R. Douceur.
\newblock The sybil attack.
\newblock In {\em First International Workshop on Peer-to-Peer Systems, {IPTPS}, Revised Papers}, 2002.
\newblock \href {https://doi.org/10.1007/3-540-45748-8\_24} {\path{doi:10.1007/3-540-45748-8\_24}}.

\bibitem{EllenGMP21}
Faith Ellen, Barun Gorain, Avery Miller, and Andrzej Pelc.
\newblock Constant-length labeling schemes for deterministic radio broadcast.
\newblock {\em {ACM} Trans. Parallel Comput.}, 8(3):14:1--14:17, 2021.
\newblock \href {https://doi.org/10.1145/3470633} {\path{doi:10.1145/3470633}}.

\bibitem{EmekFKR09}
Yuval Emek, Pierre Fraigniaud, Amos Korman, and Adi Ros{\'{e}}n.
\newblock Online computation with advice.
\newblock In {\em Proc. {ICALP}}, pages 427--438, 2009.

\bibitem{Pa-FelberKSS14}
Pascal Felber, Peter~G. Kropf, Eryk Schiller, and Sabina Serbu.
\newblock Survey on load balancing in peer-to-peer distributed hash tables.
\newblock {\em {IEEE} Commun. Surv. Tutorials}, 16(1):473--492, 2014.
\newblock \href {https://doi.org/10.1109/SURV.2013.060313.00157} {\path{doi:10.1109/SURV.2013.060313.00157}}.

\bibitem{FeldmannKSS18}
Michael Feldmann, Christina Kolb, Christian Scheideler, and Thim Strothmann.
\newblock Self-stabilizing supervised publish-subscribe systems.
\newblock In {\em Proc. IPDPS}, pages 1050--1059, 2018.
\newblock \href {https://doi.org/10.1109/IPDPS.2018.00114} {\path{doi:10.1109/IPDPS.2018.00114}}.

\bibitem{feldmann2020survey}
Michael Feldmann, Christian Scheideler, and Stefan Schmid.
\newblock Survey on algorithms for self-stabilizing overlay networks.
\newblock {\em {ACM} Comput. Surv.}, 53(4):74:1--74:24, 2021.
\newblock \href {https://doi.org/10.1145/3397190} {\path{doi:10.1145/3397190}}.

\bibitem{Feuilloley-DMTCS21}
Laurent Feuilloley.
\newblock Introduction to local certification.
\newblock {\em Discret. Math. Theor. Comput. Sci.}, 23(3), 2021.
\newblock \href {https://doi.org/10.46298/DMTCS.6280} {\path{doi:10.46298/DMTCS.6280}}.

\bibitem{FraigniaudGIP09}
Pierre Fraigniaud, Cyril Gavoille, David Ilcinkas, and Andrzej Pelc.
\newblock Distributed computing with advice: information sensitivity of graph coloring.
\newblock {\em Distributed Comput.}, 21(6):395--403, 2009.
\newblock \href {https://doi.org/10.1007/S00446-008-0076-Y} {\path{doi:10.1007/S00446-008-0076-Y}}.

\bibitem{FraigniaudIP10}
Pierre Fraigniaud, David Ilcinkas, and Andrzej Pelc.
\newblock Communication algorithms with advice.
\newblock {\em J. Comput. Syst. Sci.}, 76(3-4):222--232, 2010.
\newblock \href {https://doi.org/10.1016/J.JCSS.2009.07.002} {\path{doi:10.1016/J.JCSS.2009.07.002}}.

\bibitem{FraigniaudKL10}
Pierre Fraigniaud, Amos Korman, and Emmanuelle Lebhar.
\newblock Local {MST} computation with short advice.
\newblock {\em Theory Comput. Syst.}, 47(4):920--933, 2010.
\newblock \href {https://doi.org/10.1007/S00224-010-9280-9} {\path{doi:10.1007/S00224-010-9280-9}}.

\bibitem{FriedmanOPODIS13}
Roy Friedman, Gabriel Kliot, and Alex Kogan.
\newblock Hybrid distributed consensus.
\newblock In {\em Proc. {OPODIS}}, pages 145--159, 2013.
\newblock \href {https://doi.org/10.1007/978-3-319-03850-6\_11} {\path{doi:10.1007/978-3-319-03850-6\_11}}.

\bibitem{Gall-ThCompSys14}
Dominik Gall, Riko Jacob, Andr{\'{e}}a~W. Richa, Christian Scheideler, Stefan Schmid, and Hanjo T{\"{a}}ubig.
\newblock A note on the parallel runtime of self-stabilizing graph linearization.
\newblock {\em Theory Comput. Syst.}, 55(1):110--135, 2014.
\newblock \href {https://doi.org/10.1007/S00224-013-9504-X} {\path{doi:10.1007/S00224-013-9504-X}}.

\bibitem{GavoillePPR04}
Cyril Gavoille, David Peleg, St{\'{e}}phane P{\'{e}}rennes, and Ran Raz.
\newblock Distance labeling in graphs.
\newblock {\em J. Algorithms}, 53(1):85--112, 2004.
\newblock \href {https://doi.org/10.1016/J.JALGOR.2004.05.002} {\path{doi:10.1016/J.JALGOR.2004.05.002}}.

\bibitem{GilbertNVW21}
Seth Gilbert, Calvin Newport, Nitin~H. Vaidya, and Alex Weaver.
\newblock Contention resolution with predictions.
\newblock In {\em Proc. {PODC}}, pages 127--137, 2021.
\newblock \href {https://doi.org/10.1145/3465084.3467911} {\path{doi:10.1145/3465084.3467911}}.

\bibitem{Gmyr-ThCompSys19}
Robert Gmyr, Jonas Lef{\`{e}}vre, and Christian Scheideler.
\newblock Self-stabilizing metric graphs.
\newblock {\em Theory Comput. Syst.}, 63(2):177--199, 2019.
\newblock \href {https://doi.org/10.1007/S00224-017-9823-4} {\path{doi:10.1007/S00224-017-9823-4}}.

\bibitem{GuptaJCSS23}
Diksha Gupta, Jared Saia, and Maxwell Young.
\newblock Bankrupting sybil despite churn.
\newblock {\em J. Comput. Syst. Sci.}, 135:89--124, 2023.
\newblock \href {https://doi.org/10.1016/J.JCSS.2023.02.004} {\path{doi:10.1016/J.JCSS.2023.02.004}}.

\bibitem{ImMINFOCOM23}
Sungjin Im, Benjamin Moseley, Chenyang Xu, and Ruilong Zhang.
\newblock Online dynamic acknowledgement with learned predictions.
\newblock In {\em Proc. {INFOCOM}}, pages 1--10, 2023.
\newblock \href {https://doi.org/10.1109/INFOCOM53939.2023.10228882} {\path{doi:10.1109/INFOCOM53939.2023.10228882}}.

\bibitem{IraniR96}
Sandy Irani and Yuval Rabani.
\newblock On the value of coordination in distributed decision making.
\newblock {\em {SIAM} J. Comput.}, 25(3):498--519, 1996.
\newblock \href {https://doi.org/10.1137/S0097539794261428} {\path{doi:10.1137/S0097539794261428}}.

\bibitem{SS-JacobRSST14}
Riko Jacob, Andr{\'{e}}a~W. Richa, Christian Scheideler, Stefan Schmid, and Hanjo T{\"{a}}ubig.
\newblock Skip\({}^{\mbox{+}}\): {A} self-stabilizing skip graph.
\newblock {\em J. {ACM}}, 61(6):36:1--36:26, 2014.
\newblock \href {https://doi.org/10.1145/2629695} {\path{doi:10.1145/2629695}}.

\bibitem{jacob2014skip+}
Riko Jacob, Andr{\'{e}}a~W. Richa, Christian Scheideler, Stefan Schmid, and Hanjo T{\"{a}}ubig.
\newblock Skip\({}^{\mbox{+}}\): {A} self-stabilizing skip graph.
\newblock {\em J. {ACM}}, 61(6):36:1--36:26, 2014.
\newblock \href {https://doi.org/10.1145/2629695} {\path{doi:10.1145/2629695}}.

\bibitem{Lin-JacobRSS12}
Riko Jacob, Stephan Ritscher, Christian Scheideler, and Stefan Schmid.
\newblock Towards higher-dimensional topological self-stabilization: {A} distributed algorithm for delaunay graphs.
\newblock {\em Theor. Comput. Sci.}, 457:137--148, 2012.
\newblock \href {https://doi.org/10.1016/J.TCS.2012.07.029} {\path{doi:10.1016/J.TCS.2012.07.029}}.

\bibitem{JaJa92}
Joseph~F. J{\'{a}}J{\'{a}}.
\newblock {\em An Introduction to Parallel Algorithms}.
\newblock Addison-Wesley, 1992.

\bibitem{JiangICLR20}
Tanqiu Jiang, Yi~Li, Honghao Lin, Yisong Ruan, and David~P. Woodruff.
\newblock Learning-augmented data stream algorithms.
\newblock In {\em Proc. {ICLR}}, 2020.

\bibitem{KargerR06}
David~R. Karger and Matthias Ruhl.
\newblock Simple efficient load-balancing algorithms for peer-to-peer systems.
\newblock {\em Theory Comput. Syst.}, 39(6):787--804, 2006.
\newblock \href {https://doi.org/10.1007/S00224-006-1246-6} {\path{doi:10.1007/S00224-006-1246-6}}.

\bibitem{KniesburgesSPAA11}
Sebastian Kniesburges, Andreas Koutsopoulos, and Christian Scheideler.
\newblock Re-chord: a self-stabilizing chord overlay network.
\newblock In {\em Proc. {SPAA}}, pages 235--244, 2011.
\newblock \href {https://doi.org/10.1145/1989493.1989527} {\path{doi:10.1145/1989493.1989527}}.

\bibitem{kniesburges2012self}
Sebastian Kniesburges, Andreas Koutsopoulos, and Christian Scheideler.
\newblock A self-stabilization process for small-world networks.
\newblock In {\em Proc. {IPDPS}}, pages 1261--1271, 2012.
\newblock \href {https://doi.org/10.1109/IPDPS.2012.115} {\path{doi:10.1109/IPDPS.2012.115}}.

\bibitem{KormanKP10}
Amos Korman, Shay Kutten, and David Peleg.
\newblock Proof labeling schemes.
\newblock {\em Distributed Comput.}, 22(4):215--233, 2010.
\newblock \href {https://doi.org/10.1007/S00446-010-0095-3} {\path{doi:10.1007/S00446-010-0095-3}}.

\bibitem{KormanDC10}
Amos Korman, Shay Kutten, and David Peleg.
\newblock Proof labeling schemes.
\newblock {\em Distributed Comput.}, 22(4):215--233, 2010.
\newblock \href {https://doi.org/10.1007/S00446-010-0095-3} {\path{doi:10.1007/S00446-010-0095-3}}.

\bibitem{KothapalliS05}
Kishore Kothapalli and Christian Scheideler.
\newblock Supervised peer-to-peer systems.
\newblock In {\em Proc. {ISPAN}}, pages 188--193, 2005.
\newblock \href {https://doi.org/10.1109/ISPAN.2005.81} {\path{doi:10.1109/ISPAN.2005.81}}.

\bibitem{KOUTSOPOULOS2017408}
Andreas Koutsopoulos, Christian Scheideler, and Thim Strothmann.
\newblock Towards a universal approach for the finite departure problem in overlay networks.
\newblock {\em Inf. Comput.}, 255:408--424, 2017.
\newblock \href {https://doi.org/10.1016/J.IC.2016.12.006} {\path{doi:10.1016/J.IC.2016.12.006}}.

\bibitem{KraskaSIGMOD18}
Tim Kraska, Alex Beutel, Ed~H. Chi, Jeffrey Dean, and Neoklis Polyzotis.
\newblock The case for learned index structures.
\newblock In {\em Proc. {SIGMOD}}, pages 489--504, 2018.
\newblock \href {https://doi.org/10.1145/3183713.3196909} {\path{doi:10.1145/3183713.3196909}}.

\bibitem{LattanziSODA20}
Silvio Lattanzi, Thomas Lavastida, Benjamin Moseley, and Sergei Vassilvitskii.
\newblock Online scheduling via learned weights.
\newblock In {\em Proc. {SODA}}, pages 1859--1877, 2020.
\newblock \href {https://doi.org/10.1137/1.9781611975994.114} {\path{doi:10.1137/1.9781611975994.114}}.

\bibitem{leighton2014introduction}
Frank~Thomson Leighton.
\newblock {\em Introduction to parallel algorithms and architectures: Arrays{\textperiodcentered} trees{\textperiodcentered} hypercubes}.
\newblock M. Kaufmann Publishers, 1992.

\bibitem{LykourisJACM21}
Thodoris Lykouris and Sergei Vassilvitskii.
\newblock Competitive caching with machine learned advice.
\newblock {\em J. {ACM}}, 68(4):24:1--24:25, 2021.
\newblock \href {https://doi.org/10.1145/3447579} {\path{doi:10.1145/3447579}}.

\bibitem{MalkhiNR02}
Dahlia Malkhi, Moni Naor, and David Ratajczak.
\newblock Viceroy: a scalable and dynamic emulation of the butterfly.
\newblock In {\em Proc. {PODC}}, pages 183--192, 2002.
\newblock \href {https://doi.org/10.1145/571825.571857} {\path{doi:10.1145/571825.571857}}.

\bibitem{SS-ManneMPT09}
Fredrik Manne, Morten Mjelde, Laurence Pilard, and S{\'{e}}bastien Tixeuil.
\newblock A new self-stabilizing maximal matching algorithm.
\newblock {\em Theor. Comput. Sci.}, 410(14):1336--1345, 2009.
\newblock \href {https://doi.org/10.1016/J.TCS.2008.12.022} {\path{doi:10.1016/J.TCS.2008.12.022}}.

\bibitem{MaoDVKS20}
Yifan Mao, Soubhik Deb, Shaileshh~Bojja Venkatakrishnan, Sreeram Kannan, and Kannan Srinivasan.
\newblock Perigee: Efficient peer-to-peer network design for blockchains.
\newblock In {\em Proc. {PODC}}, pages 428--437, 2020.
\newblock \href {https://doi.org/10.1145/3382734.3405704} {\path{doi:10.1145/3382734.3405704}}.

\bibitem{MaymounkovM02}
Petar Maymounkov and David Mazi{\`{e}}res.
\newblock Kademlia: {A} peer-to-peer information system based on the {XOR} metric.
\newblock In {\em Proc. {IPTPS}}, pages 53--65, 2002.
\newblock \href {https://doi.org/10.1007/3-540-45748-8\_5} {\path{doi:10.1007/3-540-45748-8\_5}}.

\bibitem{Mitzenmacher18}
Michael Mitzenmacher.
\newblock A model for learned bloom filters and optimizing by sandwiching.
\newblock In {\em Proc. {NeurIPS}}, 2018.

\bibitem{MitzenmacherV20}
Michael Mitzenmacher and Sergei Vassilvitskii.
\newblock Algorithms with predictions.
\newblock In {\em Beyond the Worst-Case Analysis of Algorithms}, pages 646--662. Cambridge University Press, 2020.
\newblock \href {https://doi.org/10.1017/9781108637435.037} {\path{doi:10.1017/9781108637435.037}}.

\bibitem{MitzenmacherV22}
Michael Mitzenmacher and Sergei Vassilvitskii.
\newblock Algorithms with predictions.
\newblock {\em Commun. {ACM}}, 65(7):33--35, 2022.
\newblock \href {https://doi.org/10.1145/3528087} {\path{doi:10.1145/3528087}}.

\bibitem{Lin-NorNT13}
Rizal~Mohd Nor, Mikhail Nesterenko, and S{\'{e}}bastien Tixeuil.
\newblock Linearizing peer-to-peer systems with oracles.
\newblock In {\em Proc. {SSS}}, pages 221--236, 2013.
\newblock \href {https://doi.org/10.1007/978-3-319-03089-0\_16} {\path{doi:10.1007/978-3-319-03089-0\_16}}.

\bibitem{onus2007linearization}
Melih Onus, Andr{\'{e}}a~W. Richa, and Christian Scheideler.
\newblock Linearization: Locally self-stabilizing sorting in graphs.
\newblock In {\em Proc. {ALENEX}}, 2007.
\newblock \href {https://doi.org/10.1137/1.9781611972870.10} {\path{doi:10.1137/1.9781611972870.10}}.

\bibitem{PanduranganFOCS01}
Gopal Pandurangan, Prabhakar Raghavan, and Eli Upfal.
\newblock Building low-diameter {P2P} networks.
\newblock In {\em Proc. {FOCS}}, pages 492--499, 2001.
\newblock \href {https://doi.org/10.1109/SFCS.2001.959925} {\path{doi:10.1109/SFCS.2001.959925}}.

\bibitem{PapadimitriouY93}
Christos~H. Papadimitriou and Mihalis Yannakakis.
\newblock Linear programming without the matrix.
\newblock In {\em Proc. {STOC}}, pages 121--129, 1993.
\newblock \href {https://doi.org/10.1145/167088.167127} {\path{doi:10.1145/167088.167127}}.

\bibitem{RatnasamyFHKS01}
Sylvia Ratnasamy, Paul Francis, Mark Handley, Richard~M. Karp, and Scott Shenker.
\newblock A scalable content-addressable network.
\newblock In {\em Proc. {SIGCOMM}}, pages 161--172, 2001.
\newblock \href {https://doi.org/10.1145/383059.383072} {\path{doi:10.1145/383059.383072}}.

\bibitem{RowstronD01}
Antony I.~T. Rowstron and Peter Druschel.
\newblock Pastry: Scalable, decentralized object location, and routing for large-scale peer-to-peer systems.
\newblock In {\em Proc. Middleware}, pages 329--350, 2001.
\newblock \href {https://doi.org/10.1007/3-540-45518-3\_18} {\path{doi:10.1007/3-540-45518-3\_18}}.

\bibitem{scheideler2019complexity}
Christian Scheideler and Alexander Setzer.
\newblock On the complexity of local graph transformations.
\newblock In {\em Proc. ICALP}, pages 150:1--150:14, 2019.
\newblock \href {https://doi.org/10.4230/LIPICS.ICALP.2019.150} {\path{doi:10.4230/LIPICS.ICALP.2019.150}}.

\bibitem{ScheidelerSS15}
Christian Scheideler, Alexander Setzer, and Thim Strothmann.
\newblock Towards establishing monotonic searchability in self-stabilizing data structures.
\newblock In {\em Proc. {OPODIS}}, pages 24:1--24:17, 2015.
\newblock \href {https://doi.org/10.4230/LIPICS.OPODIS.2015.24} {\path{doi:10.4230/LIPICS.OPODIS.2015.24}}.

\bibitem{StoicaMKKB01}
Ion Stoica, Robert~Tappan Morris, David~R. Karger, M.~Frans Kaashoek, and Hari Balakrishnan.
\newblock Chord: {A} scalable peer-to-peer lookup service for internet applications.
\newblock In {\em Proc. {SIGCOMM}}, pages 149--160, 2001.
\newblock \href {https://doi.org/10.1145/383059.383071} {\path{doi:10.1145/383059.383071}}.

\bibitem{Suomela13}
Jukka Suomela.
\newblock Survey of local algorithms.
\newblock {\em {ACM} Comput. Surv.}, 45(2):24:1--24:40, 2013.
\newblock \href {https://doi.org/10.1145/2431211.2431223} {\path{doi:10.1145/2431211.2431223}}.

\bibitem{SS-tixeuil2009self}
S\'{e}bastien Tixeuil.
\newblock Self-stabilizing algorithms.
\newblock In {\em Algorithms and Theory of Computation Handbook}, pages 26:1--26:45. Chapman \& Hall/CRC, 2010.

\end{thebibliography}

\appendix 
\newpage

\section{Pseudocodes} \label{sec:pseudocodes}

In every round, each node executes the flyover algorithms (cf. Section \ref{sec:flyover-code}), distributed Tree-to-Path algorithms (cf. Section \ref{subsec:distrib-TtP}), and then the base algorithm (in that order). In Section \ref{sec:helpful-figs}, we outline the list of all variables used by our algorithms.

\pg{Transferring the ``advised'' neighbors to base algorithm.} For any node $u$, the variable $u.\textit{c-id}$ is used to store its neighbors on the ``advised'' topology (i.e., connections made by routing over the hypercubic structure due to the advice). Thus, in every round, the ids stored in that variable (if it is nonempty) are added to the address variable(s) of the base algorithm, used to store the neighbors of the target topology, before the algorithm is executed. If the supervisor is honest, this helps the nodes obtain the explicit edges required for the base algorithm to stabilize. However, if the advice is corrupted in any way, then by design, all nodes reject their advice (i.e. execute $\mathbf{RejectFlyover}$ function) in $O(\log n)$ rounds, and for any node $u$, $u.\textit{c-id}$ becomes empty.

\pg{Choice of base algorithm.} In general, any base algorithm that self-stabilizes to the sorted line using the overlay primitives (cf. Section \ref{subsec:univ-prim}) can be chosen. We only require that the algorithm use the Delegate-after-Reversal or DR primitive (cf. Definition \ref{def:bidirected}), in place of the Delegate primitive, to facilitate information exchange in both directions. (Note that this mechanism is already part of many existing algorithms, including \cite{onus2007linearization, jacob2014skip+}.) The DR primitive first delegates an edge in the reverse direction, and then it is immediately set to the intended direction. This primitive can either be implemented in a single round with a ``preprocessing'' phase (where the edge-reversal happens; see, e.g., Skip+ algorithm \cite{jacob2014skip+} where nodes create bidirected links in such a phase), or simply in two rounds in our model.

As examples, if the base algorithm is the classic Linearization algorithm \cite{onus2007linearization}, the network self-stabilizes in $O(n)$ rounds if the supervisor is malicious; if the base algorithm is the Skip+ algorithm \cite{jacob2014skip+} (that converges to a topology where the sorted-path is part of it, and the degree of each node is $O(\log n)$), the network self-stabilizes in $O(\log^2 n)$ rounds if the supervisor is malicious. Both of those algorithms can be directly used in our setting.

\subsection{Flyover Pseudocode} \label{sec:flyover-code}

\begin{algorithm} \label{alg:init}
\caption{Init}
\begin{algorithmic}[1]
\STATE $\mathbf{BasicChecks}$
\begin{ALC@g}
\STATE $(\exists \langle \text{RejFlyover} \rangle \in u.\mathit{Ch}) \rightarrow u.\mathit{exit} \coloneqq 1$;
\STATE $(u.S = \emptyset) \land (u.\textit{c-ids} \neq \emptyset \lor u.\textit{flyID} \neq u.\mathit{id}) \rightarrow u.\mathit{exit} \coloneqq 1$;
\STATE $(u.L = \emptyset \land u.R \neq \emptyset) \land (u.\mathit{vID} \neq 1 \lor u.\mathit{flyID} \neq u.id) \rightarrow  u.\mathit{exit} \coloneqq 1$;
\STATE $(u.L \neq \emptyset) \land (u.\mathit{vID} \leq 1) \rightarrow u.\mathit{exit} \coloneqq 1$;
\STATE $(u.\mathit{vID} = 1 \land u.\text{\emph{c-dist}} \neq 0) \lor (u.\mathit{vID} > 1 \land u.\text{\emph{c-dist}} = 0) \rightarrow u.\mathit{exit} \coloneqq 1$;
\STATE $(u.S \neq \emptyset \land u.\mathit{vID} > 1) \land (\mathbf{NextStop}(u.\text{\emph{c-par}}) = \varnothing) \rightarrow u.\mathit{exit} \coloneqq 1$;
\STATE $|u.\text{\emph{c-ids}}| > 2 \rightarrow u.\mathit{exit} \coloneqq 1$;
\STATE $(|u.\text{\emph{c-ids}}| = 2) \land (u.\mathit{id} > \max(u.\text{\emph{c-ids}}) \lor u.\mathit{id} < \min(u.\text{\emph{c-ids}})) \rightarrow u.\mathit{exit} \coloneqq 1$;
\end{ALC@g}
\STATE
\STATE $\mathbf{RejectFlyover} \text{: } u.\mathit{exit} = 1 \rightarrow$
\begin{ALC@g}
\STATE $\mathit{flyids} \coloneqq ( u.S \cup \{u.\mathit{flyID}\} \cup u.\text{\emph{c-ids}} ) \setminus u.\mathit{id}$;
\STATE SEND($ \langle \text{RejFlyover} \rangle$) to ids in $\mathit{flyids}$; $\mathbf{Flush}(\mathit{flyids})$; 
\STATE $u.L \coloneqq \emptyset; u.R \coloneqq \emptyset; u.\mathit{vID} \coloneqq 0; u.\mathit{flyID} \coloneqq u.\mathit{id};$
\STATE $u.\mathit{exit} \coloneqq 0; u.\text{\emph{c-par}} \coloneqq 0; u.\text{\emph{c-dist}} \coloneqq -1; u.\text{\emph{c-ids}} \coloneqq \emptyset;$
\end{ALC@g}
\end{algorithmic}
\end{algorithm}

\begin{algorithm} \label{alg:tests}
\caption{Test Flyover and Connectivity Certificate}
\begin{algorithmic}[1]
\STATE $\mathbf{TestFlyoverConstruction}$
\begin{ALC@g}
\STATE $u.R \neq \emptyset \rightarrow$ SEND($\langle \text{TestLine-R}, u.\mathit{id} \rangle$) to $u.S_r(1)$;
\STATE $u.L \neq \emptyset \rightarrow$ SEND($\langle \text{TestLine-L}, u.\mathit{id} \rangle$) to $u.S_l(1)$;
\IF{$(u.R \neq \emptyset \land u.L \neq \emptyset)$}
\STATE $\forall i \in [\min(|u.R|, |u.L|)]$, SEND($\langle \text{FlyConst-R}, u.S_l(i), i, u.\mathit{id} \rangle$) to $u.S_r(i)$;
\STATE $\forall i \in [\min(|u.R|, |u.L|)]$, SEND($\langle \text{FlyConst-L}, u.S_r(i), i, u.\mathit{id} \rangle$) to $u.S_l(i)$;
\ENDIF
\end{ALC@g}
\STATE
\STATE $\mathbf{TestConnCertificate}\text{: } (u.\mathit{vID} > 1 \land u.\mathit{flyID} \neq u.\mathit{id}) \rightarrow$
\begin{ALC@g}
\STATE SEND($\langle \text{TestCert}, u.\mathit{id}, u.\text{\emph{c-par}}, u.\text{\emph{c-dist}}\rangle$) to $\mathbf{NextStop}(u.\text{\emph{c-par}})$;
\end{ALC@g}
\STATE 
\STATE $\mathbf{TestFlyoverMetadata}$
\begin{ALC@g}
\STATE $\textit{prop-flyID} \coloneqq (u.\mathit{vID} = 1) \lor (u.\mathit{vID} > 1 \land u.\mathit{flyID} \neq u.\mathit{id});$
\STATE $|u.R| \geq 1 \rightarrow$ SEND($\langle \text{TestvID},(u.\mathit{vID} + 2^{(i-1)})\rangle$) to $u.S_r(i)$ for all $i \in [|u.R|]$;
\STATE $|u.L| \geq 1 \rightarrow$ SEND($\langle \text{TestvID},(u.\mathit{vID} - 2^{(i-1)})\rangle$) to $u.S_l(i)$ for all $i \in [|u.L|]$;
\STATE $(u.S = \emptyset \land u.\mathit{vID} = 0) \rightarrow$ SEND($\langle \text{TestFlyID}, \bot \rangle$) to all ids in memory;
\STATE $\textit{prop-flyID} \rightarrow$ SEND($\langle \text{TestFlyID}, u.\mathit{flyID} \rangle$) to all ids, except $u.\mathit{flyID}$, in memory;
\end{ALC@g}
\end{algorithmic}
\end{algorithm}

\begin{algorithm} \label{alg:resp}
\caption{Response Functions (Flyover Related)}
\begin{algorithmic}[1]
\STATE $\mathbf{R\_TestFlyoverConstruction}$
\begin{ALC@g}
\STATE \COMMENT{Here, we omit code for a ``\dots-L'' msg due to similarity in handling a ``\dots-R'' msg}
\FOR{each received message $\langle \mathit{msg},\mathit{sen} \rangle$ where $\mathit{msg} \in \{ \text{TestLine-R}, \text{TestLine-L} \}$}
\STATE $(\mathit{msg} = \text{TestLine-R}) \land ((u.L \neq \emptyset \land u.S_l(1) \neq \mathit{sen}) \lor u.L = \emptyset) \rightarrow u.\mathit{exit} \coloneqq 1$;
\STATE $u.S = \emptyset \lor u.\mathit{exit} = 1  \rightarrow$  SEND($\langle \text{RejFlyover} \rangle$) to $\mathit{sen}$; $\mathbf{Flush}(\{\mathit{sen}\})$; 
\ENDFOR
\FOR{each received message $\langle \mathit{msg}, w, i, \mathit{sen} \rangle$ where $\mathit{msg} \in \{ \text{FlyConst-R}, \text{FlyConst-L} \}$}
\IF{$\mathit{msg} = \text{FlyConst-R}$}
\STATE $(|u.L| \geq i \land u.S_l(i) \neq \mathit{sen}) \lor u.L = \emptyset \rightarrow u.\mathit{exit} \coloneqq 1$;
\STATE $|u.L| \geq (i+1) \land u.S_l(i+1) \neq w \rightarrow u.\mathit{exit} \coloneqq 1$;
\STATE $u.\mathit{exit} = 0 \land (1 < |u.L| < i) \rightarrow \mathbf{Flush}(\{\mathit{sen}, w\})$;
\STATE $u.\mathit{exit} = 0 \land (|u.L| = i \land u.S_l(i) = \mathit{sen}) \rightarrow u.S_l(|u.L| + 1) \coloneqq w$;
\ENDIF
\STATE $u.S = \emptyset \lor u.\mathit{exit} = 1   \rightarrow$ SEND($\langle \text{RejFlyover} \rangle$) to $\{\mathit{sen, w}\}$; $\mathbf{Flush}(\{\mathit{sen, w}\})$; 
\ENDFOR
\end{ALC@g}
\STATE
\STATE $\mathbf{R\_TestConnCertificate}$
\begin{ALC@g}
\FOR{each received message $\langle \text{TestCert}, w, \mathit{vID}, \mathit{dist}\rangle$}
\STATE $\textit{prop-flyID} \coloneqq (u.\mathit{vID} = 1) \lor (u.\mathit{vID} > 1 \land u.\mathit{flyID} \neq u.\mathit{id});$
\STATE $(u.\mathit{vID} = \mathit{vID}) \land (\mathit{dist}-1 \neq u.\textit{c-dist}) \rightarrow u.\mathit{exit \coloneqq} 1$;
\STATE $(u.\mathit{vID} \neq \mathit{vID}) \land (\mathbf{NextStop}(\mathit{vID}) = \varnothing) \rightarrow u.\mathit{exit \coloneqq} 1$;
\IF{$(u.S = \emptyset \lor u.\mathit{exit} = 1)$}
\STATE SEND($\langle \text{RejFlyover} \rangle$) to $\mathit{w}$; $\mathbf{Flush}(\{\mathit{w}\})$; 
\ELSE
\STATE $u.\mathit{vID} \neq \mathit{vID} \land \neg \textit{prop-flyID} \rightarrow \mathbf{Flush}(\{w\})$; 
\STATE $u.\mathit{vID} \neq \mathit{vID} \land \textit{prop-flyID} \rightarrow$ SEND($\langle \text{TestCert}, w, \mathit{vID}, \mathit{dist}\rangle$) to $\mathbf{NextStop}(\mathit{vID})$;
\STATE $u.\mathit{vID} = \mathit{vID} \rightarrow u.\textit{c-ids} \coloneqq u.\textit{c-ids} \cup \{ w \}$; SEND($\langle \text{IntroCert}, u.\mathit{id} \rangle$ to $w$;
\ENDIF
\ENDFOR 
\FOR{each received message $\langle \text{IntroCert}, w \rangle$}
\STATE $u.\textit{c-ids} \coloneqq u.\textit{c-ids} \cup \{ w \}$; 
\ENDFOR
\end{ALC@g}
\STATE 
\STATE $\mathbf{R\_TestFlyoverMetadata}$
\begin{ALC@g}
\FOR{each received message $\langle \text{TestvID}, \mathit{vID} \rangle$}
\STATE $(u.S = \emptyset) \lor (u.S \neq \emptyset \land  u.\mathit{vID} \neq \mathit{vID}) \rightarrow u.\mathit{exit} = 1$;
\ENDFOR
\FOR{each received message $\langle \text{TestFlyID}, \mathit{flyID} \rangle$}
\IF{$(u.L = \emptyset \land u.R \neq \emptyset) \land (u.\mathit{exit} = 0)$}
\STATE $(u.\mathit{flyID} = u.\mathit{id} \land \mathit{flyID} \neq \bot) \rightarrow u.\mathit{flyID} \coloneqq \mathit{flyID}$; 
\ENDIF
\STATE $(u.S \neq \emptyset \land u.\mathit{flyID} \neq \mathit{flyID}) \lor (u.S = \emptyset \land \mathit{flyID} \neq \bot) \rightarrow u.\mathit{exit} = 1$; 
\STATE $u.\mathit{exit} = 1 \rightarrow$  SEND($\langle \text{RejFlyover} \rangle$) to $\mathit{flyID}$; $\mathbf{Flush}(\{\mathit{flyID}\})$; 
\ENDFOR
\end{ALC@g}
\end{algorithmic}
\end{algorithm}
\newpage

\subsection{Tree-to-Path Pseudocode} \label{sec:TtP-codes}

In this subsection, we first provide the pseudocode of the (centralized) Tree-to-Path algorithm described in Section \ref{subsec:tree-to-path}. Then, we give the pseudocode for the distributed implementation of the same algorithm, which also includes the interaction with the supervisor.

\subsubsection{Sequential Tree-to-Path} \label{subsec:seq-TtP}

\begin{algorithm}
\caption{Tree-to-Path}
\label{alg:tree-to-line}
\begin{algorithmic}[1]
\REQUIRE Let the input to the algorithm be a (labelled) rooted tree $T = (l(\cdot), r, V, E)$ where $l(r) = 0$ and $|V| \geq 2$. For any $v \in V$, let $C(v) = \{ u \mid \mathrm{dist}(u, v) = 1 \wedge \mathrm{dist}(v, r) < \mathrm{dist}(u, r) \}$.
\ENSURE Output of the algorithm is a directed path graph $P$ where $\mathrm{Beg}(P) = r$.
\FOR{$v \in V$ and $C(v) \neq \emptyset$}
\STATE Assign a total ordering on $C(v) = \{ u_1, \dots, u_q \}$ such that $u_1 < \dots < u_q$ for $q \geq 1$;
\ENDFOR
\STATE Let $E' \leftarrow \emptyset$;
\FOR{$v \in V$ and $v \neq r$ and $l(v) = 1$}
\STATE Let $p_v \in V$ such that $v \in C(p_v)$; \COMMENT{Parent of $v$}
\STATE Let $\mathit{RSib} = \{ w \mid w > v \wedge w \in C(p_v)\}$; \COMMENT{Right siblings (via the total ordering)}
\IF{$\mathit{RSib} = \emptyset$ and $C(v)  = \emptyset$}
\STATE Add (directed) edge $(p_v, v)$ to $E'$;
\ELSIF{$\mathit{RSib} = \emptyset$ and $C(v)  \neq \emptyset$}
\STATE Add (directed) edge $(p_v, \max(C(v)))$ to $E'$;
\ELSIF{$\mathit{RSib} \neq \emptyset$ and $C(v)  = \emptyset$}
\STATE Add (directed) edge $(\min(\mathit{RSib}), v)$ to $E'$;
\ELSE
\STATE Add (directed) edge $(\min(\mathit{RSib}), \max(C(v)))$ to $E'$;
\ENDIF
\ENDFOR
\FOR{$v \in V$ and $v \neq r$ and $l(v) = 0$}
\STATE Let $p_v \in V$ such that $v \in C(p_v)$; \COMMENT{Parent of $v$}
\STATE Let $\mathit{LSib} = \{ w \mid w < v \wedge w \in C(u)\}$; \COMMENT{Left siblings (via the total ordering)}
\IF{$\mathit{LSib} = \emptyset$ and $C(v)  = \emptyset$}
\STATE Add (directed) edge  $(u, p_v)$ to $E'$;
\ELSIF{$\mathit{LSib} = \emptyset$ and $C(v)  \neq \emptyset$}
\STATE Add (directed) edge $(\min(C(v)), p_v
)$ to $E'$;
\ELSIF{$\mathit{LSib} \neq \emptyset$ and $C(v)  = \emptyset$}
\STATE Add (directed) edge $(v, \max(\mathit{LSib}))$ to $E'$;
\ELSE
\STATE Add (directed) edge $(\min(C(v)), \max(\mathit{LSib}))$ to $E'$;
\ENDIF
\ENDFOR
\STATE Output $P = (V, E')$;
\end{algorithmic}
\end{algorithm}


\subsubsection{Distributed Tree-to-Path} \label{subsec:distrib-TtP}

Here, we describe the self-stabilizing algorithms for receiving and processing the advice message. We explicitly address the interaction between the supervisor and the nodes, and also show that it can be combined with the rest of the algorithms.

From the (honest) supervisor's perceptive, if there is any node that is not part of any flyover, and has its exit variable set to 1, that node would request for advice. Once the supervisor receives the first request, it \emph{waits} for $O(\log n)$ rounds until \emph{all} the nodes are not part of any flyover, and have their exit variables set to 1. (Note that by design, if the advice or flyover is not correct, every node executes $\mathbf{RejectFlyover}$ function of Algorithm \ref{alg:conf-init} in $O(\log n)$ rounds; see Section \ref{anal:det-bad-advice} in the analysis.) At this point, the supervisor learns the ``network snapshot,'' by requesting each node to send its neighborhood, that is, the ids stored in its memory or channel.

From the perspective of any node $u$, to avoid being unnecessarily contacted by the supervisor, it should not be part of any flyover (that is, $u.S \neq \emptyset$), nor should its $u.\mathit{exit}$ variable be set to 1, while processing an advice message (i.e., during the interaction with the supervisor). After being requested for the neighborhood by the supervisor, every node can store and maintain its current neighborhood for $O(1)$ rounds to process the advice given by the supervisor. This is achieved by setting and tracking a ``timer'' variable that indicates that the advice is being processed.

We explain the additional variables of a node $u$ that are used in the following algorithms.
\begin{enumerate}
    \item $u.t$ is used as a timer variable which is set to $5$ when the supervisor requests for its current neighborhood; $u.t$ is decremented by 1 in every round, until it becomes 0 again.
    \item $u.\mathit{dist}$ is used for storing the distance from the root, in a rooted spanning tree of the current snapshot, given by the supervisor; see Algorithm \ref{alg:cloud-alg} for the details. This distance is primarily used by nodes for calculating the label in the Tree-to-Path algorithm (cf. Section \ref{subsec:tree-to-path}), and is also used for quickly checking that any subset of ``parent'' ids (in the advice) do not form a cycle, using the distance-checks in local certification (cf. Section \ref{sec:advice}).
\end{enumerate}

\begin{algorithm}
\caption{Init2}
\label{alg:init2}
\begin{algorithmic}[1]
\STATE $\mathbf{BasicChecks2}$
\begin{ALC@g}
\STATE $u.t > 5 \rightarrow u.t \coloneqq 5$; $u.t < 0 \rightarrow u.t \coloneqq 0$;
\STATE $u.t > 0 \rightarrow u.t \coloneqq u.t - 1$; \COMMENT{Decrement $u.r$ (until 0)}
\STATE $u.t \leq 1 \land (u.S = \emptyset \lor u.\mathit{exit} = 1) \rightarrow u.\mathit{vID} \coloneqq 0$; \COMMENT{Reset \emph{vID} if flyover is not formed} 
\end{ALC@g}
\STATE
\STATE $\mathbf{SnapshotReq}\text{: } u.S = \emptyset \land u.\mathit{exit} = 0 \land u.t = 0 \rightarrow$
\begin{ALC@g}
\IF{there exists a message $\langle \text{RequestSnapshot}\rangle$}
\STATE $u.t \coloneqq 5$; $\mathit{snap} \coloneqq $ Set of all ids in $\mathrm{M}_u(\mathcal{A}_0)$;  \COMMENT{Set up timer and neighborhood}
\STATE SEND($\langle \text{Intro}, u.\mathit{id} \rangle$) to all nodes in $\mathit{snap}$; $\forall v \in \mathit{snap}, $ SEND($\langle \text{Intro}, v \rangle$) to thyself;
\STATE SEND($\langle \text{Neighborhood}, \mathit{snap} \rangle$) to supervisor;
\ENDIF
\end{ALC@g}
\end{algorithmic}
\end{algorithm}

\begin{algorithm}
\caption{Local Advice Verification}
\label{alg:local-cert}
\begin{algorithmic}[1]
\STATE $\mathbf{GetAdvice}$
\begin{ALC@g}
\STATE $busy \coloneqq \neg (u.S = \emptyset \land u.\mathit{exit} = 0 \land u.t > 1)$;
\STATE $\mathit{snap} \coloneqq \emptyset; (\forall \langle \text{Intro}, v \rangle \in u.\mathit{Ch}) \rightarrow \mathit{snap} \coloneqq \mathit{snap} \cup \{ v \}$;
\IF{$(\langle \text{Advice}, \cdots \rangle \in u.\mathit{Ch}) \land (\neg \mathit{busy}) \land (u.t = 4)$}
\IF{$\langle \text{Advice}, \cdots \rangle$ is well-formed (e.g., $\mathit{Adv}(u).\mathit{par} \in \mathit{snap}$, etc)} 
\STATE Update relevant variables (e.g., $u.\mathit{vID}, u.\textit{c-dist}, u.\mathit{dist}$);
\STATE SEND($\langle \text{TestAdvice}, u.\mathit{dist}, u.\mathit{id} \rangle$) to $\mathit{Adv}(u).\mathit{par}$;
\ENDIF
\ENDIF
\STATE $\mathbf{Flush}(\mathit{snap})$;
\end{ALC@g}
\STATE 
\STATE $\mathbf{CertifyTree}$
\begin{ALC@g}
\STATE $\mathit{ignore\_msg} \coloneqq \neg (u.S = \emptyset \land u.\mathit{exit} = 0 \land u.t > 1); \mathit{children} \coloneqq \emptyset$;
\FOR{each received message $\langle \text{TestAdvice}, \mathit{dist}, v\rangle$}
\STATE $\mathit{children} \coloneqq v \cup \mathit{children}$; $(u.\mathit{dist} \neq \mathit{dist} - 1) \rightarrow \mathit{ignore\_msg} \coloneqq 1$;
\ENDFOR
\STATE $(\neg \mathit{ignore\_msg} \land  \mathit{children} \neq \emptyset) \rightarrow \mathbf{SetupLocalTransform}(\mathit{children})$;
\STATE $\mathbf{Flush}(\mathit{children})$; 
\end{ALC@g}
\end{algorithmic}
\end{algorithm}

\newpage

\begin{algorithm}
\caption{Distributed Tree-to-Path}
\label{alg:distrib-tree-to-path}
\begin{algorithmic}[1]
\STATE $\mathbf{LocalTransform}$
\begin{ALC@g}
\STATE $\mathit{ignore\_msg} \coloneqq \neg (u.S = \emptyset \land u.\mathit{exit} = 0 \land u.t > 1);$
\STATE $\mathit{children} \coloneqq \emptyset; \mathit{parent} \coloneqq \bot; \textit{r-sib} \coloneqq \bot; \textit{l-sib} \coloneqq \bot$;
\FOR{each received message $\langle \text{Verified}, \mathit{msg}, v \rangle$ where $\mathit{msg} \in \{\text{parent},\text{sib+},\text{sib-},\text{child}\}$}
\STATE $(\mathit{msg} = \text{parent} \land \mathit{parent} \neq \bot) \rightarrow  \mathit{ignore\_msg} \coloneqq 1$;
\STATE $(\mathit{msg} = \text{sib-} \land \textit{l-sib} \neq \bot) \rightarrow \mathit{ignore\_msg} \coloneqq 1$;
\STATE $(\mathit{msg} = \text{sib+} \land \textit{r-sib} \neq \bot) \rightarrow \mathit{ignore\_msg} \coloneqq 1$;
\STATE $(\mathit{msg} = \text{parent} \land \neg \mathit{ignore\_msg}) \rightarrow \mathit{parent} \coloneqq v$;
\STATE $(\mathit{msg} = \text{sib-} \land \neg \mathit{ignore\_msg}) \rightarrow \textit{l-sib} \coloneqq v$; 
\STATE $(\mathit{msg} = \text{sib+} \land \neg \mathit{ignore\_msg}) \rightarrow \textit{r-sib} \coloneqq v$;
\STATE $(\mathit{msg} = \text{child} \lor \mathit{ignore\_msg} = 1) \rightarrow \mathit{children} \coloneqq v \cup \mathit{children};$
\ENDFOR
\STATE $(\mathit{parent} = \bot \land u.\mathit{dist} \neq 0) \lor (\mathit{parent} \neq \bot \land u.\mathit{dist} < 1) \rightarrow \mathit{ignore\_msg} \coloneqq 1;$
\STATE $(\mathit{ignore\_msg} = 0 \land \mathit{parent} \neq \bot) \rightarrow \mathbf{ExecuteTransform}(\mathit{parent}, \textit{r-sib}, \textit{l-sib}, \mathit{children})$;
\STATE $\mathbf{Flush}(\{\mathit{parent}, \textit{r-sib}, \textit{l-sib}\} \cup \mathit{children})$;
\end{ALC@g}
\STATE
\STATE $\mathbf{JoinPath}$
\begin{ALC@g}
\STATE $\mathit{ignore\_msg} \coloneqq \neg (u.S = \emptyset \land u.\mathit{exit} = 0 \land u.t > 1); \mathit{flyL} \coloneqq \bot; \mathit{flyR} \coloneqq \bot;$
\FOR{each received message $\langle \mathit{msg}, v \rangle$ where $\mathit{msg} \in \{\text{Path+},\text{Path-}\}$}
\STATE $(\mathit{msg} = \text{Path-} \land (u.\mathit{vID} = 1 \lor \mathit{flyL} \neq \bot)) \rightarrow \mathit{ignore\_msg} \coloneqq 1$;
\STATE $(\mathit{msg} = \text{Path+} \land \mathit{flyR} \neq \bot) \rightarrow \mathit{ignore\_msg} \coloneqq 1$;
\STATE $(\mathit{msg} = \text{Path-} \land \neg \mathit{ignore\_msg}) \rightarrow \mathit{flyL} \coloneqq v$;
\STATE $(\mathit{msg} = \text{Path+} \land \neg \mathit{ignore\_msg}) \rightarrow \mathit{flyR} \coloneqq v$;
\STATE $\mathit{ignore\_msg} = 1 \rightarrow \mathbf{Flush}(\{v\})$;
\ENDFOR
\STATE \COMMENT{First left and right shortcuts in the flyover}
\STATE $\mathit{ignore\_msg} = 0 \rightarrow u.S_l(1) \coloneqq \mathit{flyL}; u.S_r(1) \coloneqq \mathit{flyR};$
\STATE $\mathbf{Flush}(\{\mathit{flyL}, \mathit{flyR}\})$;
\end{ALC@g}
\end{algorithmic}
\end{algorithm}
\subsection{Helper Functions} \label{subsec:helpers}

\begin{algorithm}
\caption{Helper Functions (Flyover Related)}
\begin{algorithmic}[1]
\label{alg:helper}
\STATE $\mathbf{NextStop(\mathit{val})}$
\begin{ALC@g}
\STATE $(\mathit{val} < 1 \lor \mathit{val} = u.\mathit{vID}) \lor (u.S = \emptyset \lor u.\mathit{vID} < 1) \rightarrow $ Return ``$\varnothing$'';
\STATE $(\mathit{val} > u.\mathit{vID} \land (u.R = \emptyset)) \lor (\mathit{val} < u.\mathit{vID} \land (u.L = \emptyset)) \rightarrow $ Return ``$\varnothing$'';
\IF{$(\mathit{val} > u.\mathit{vID})$}
\STATE Return $u.S_r(\mathrm{argmin}_{i \in |u.R|}(\lvert u.\mathit{vID} + 2^{(i-1)} - \mathit{vID}\rvert))$;
\ELSE
\STATE Return $u.S_l(\mathrm{argmin}_{i \in |u.L|}(\lvert u.\mathit{vID} - 2^{(i-1)} - \mathit{vID}\rvert))$;
\ENDIF
\end{ALC@g}
\end{algorithmic}
\end{algorithm}

\begin{algorithm}
\caption{Helper Functions (Advice Related)}
\begin{algorithmic}[1]
\STATE $\mathbf{SetupLocalTransform(\mathit{children})}$
\begin{ALC@g}
\STATE Consider a total order on $\mathit{children} = \{ u_1, \dots, u_q \}$ such that $u_1 < \dots < u_q$ for $q > 1$;
\STATE SEND($\langle \text{Verified}, \text{parent}, u.\mathit{id} \rangle$) to all ids in $\mathit{children}$;
\STATE SEND($\langle \text{Verified}, \text{sib+}, u_j \rangle $) to $u_{j-1}$ for $q \geq j > 1$;
\STATE SEND($\langle \text{Verified}, \text{sib-}, u_j \rangle$) to $u_{j+1}$ for $q > j \geq 1$;
\STATE SEND($\langle \text{Verified}, \text{child}, w \rangle$) to self, for each $w \in \mathit{children}$;
\STATE
\end{ALC@g}
\STATE $\mathbf{ExecuteTransform(\mathit{parent}, \mathit{right\_sib}, \mathit{left\_sib}, \mathit{children})}$
\begin{ALC@g}
\STATE Let the total ordering on $\mathit{children} = \{ u_1, \dots, u_q \}$ such that $u_1 < \dots < u_q$ for $q > 1$; 
\STATE \COMMENT{Path+ messages (i.e., right shortcut) follow the orientation as in Algorithm \ref{alg:tree-to-line}.}
\FOR{$(u.\mathit{dist} \neq 0) \land ((u.\mathit{dist}\mod 2) = 1)$}
\IF{$\mathit{right\_sib} = \bot \land \mathit{children}  = \emptyset$}
\STATE SEND($\langle \text{Path+}, u.\mathit{id} \rangle$) to $\mathit{parent}$; SEND($\langle \text{Path-}, \mathit{parent}\rangle$) to self;
\ELSIF{$\mathit{right\_sib} = \bot \land \mathit{children}  \neq \emptyset$}
\STATE SEND($\langle \text{Path+}, \max(\mathit{children}) \rangle$) to $\mathit{parent}$; SEND($\langle \text{Path-}, \mathit{parent} \rangle $) to $\max(\mathit{children})$;
\ELSIF{$\mathit{right\_sib} \neq \bot \land \mathit{children}  = \emptyset$}
\STATE SEND($\langle \text{Path+}, u.\mathit{id} \rangle$) to $\mathit{right\_sib}$; SEND($\langle \text{Path-}, \mathit{right\_sib} \rangle$) to self;
\ELSE
\STATE SEND($\langle \text{Path+}, \max(\mathit{children}) \rangle$) to $\mathit{right\_sib}$; SEND($\langle \text{Path-}, \mathit{right\_sib} \rangle$) to $\max(\mathit{children})$;
\ENDIF
\ENDFOR
\FOR{$(u.\mathit{dist} \neq 0) \land ((u.\mathit{dist}\mod 2) = 0)$}
\IF{$\mathit{left\_sib} = \bot \land \mathit{children}  = \emptyset$}
\STATE SEND($\langle \text{Path-}, u.\mathit{id} \rangle$) to $\mathit{parent}$; SEND($\langle \text{Path+}, \mathit{parent} \rangle$) to self;
\ELSIF{$\mathit{left\_sib} = \bot \land \mathit{children} \neq \emptyset$}
\STATE SEND($\langle \text{Path-}, \min(\mathit{children}) \rangle$) to $\mathit{parent}$; SEND($\langle \text{Path+}, \mathit{parent} \rangle$) to $\min(\mathit{children})$;
\ELSIF{$\mathit{left\_sib} \neq \bot \land \mathit{children}  = \emptyset$}
\STATE SEND($\langle \text{Path-}, u.\mathit{id} \rangle$) to $\mathit{left\_sib}$; SEND($\langle \text{Path+}, \mathit{left\_sib} \rangle$) to self;
\ELSE
\STATE SEND($\langle \text{Path-}, \min(\mathit{children}) \rangle$) to $\mathit{left\_sib}$; SEND($\langle \text{Path+}, \mathit{left\_sib} \rangle$) to $\min(\mathit{children})$;
\ENDIF
\ENDFOR
\end{ALC@g}
\end{algorithmic}
\end{algorithm}
\newpage
\section{Tree-to-Path Analysis} \label{sec:TtP-analysis}

\begin{theorem}
Algorithm \ref{alg:tree-to-line} takes a (labelled) rooted tree $T = (l(\cdot), r, V, E)$ where $l(r) = 0$ as input, and outputs a directed path graph $P = (V, E')$, where $\mathrm{Beg}(P) = r$ and $\mathrm{End}(P) = \min(C_T(r))$.
\end{theorem}
\begin{proof}
The height of a rooted tree $T$, denoted by $H(T)$, is equal to $\max_{v\in V}(\mathrm{dist}(r, v))$. For each node $v$, let $C(v) = \{ u \mid \mathrm{dist}(u, v) = 1 \wedge \mathrm{dist}(v, r) < \mathrm{dist}(u, r) \}$. For each node $v \in V$ where $C(v) \neq \emptyset$, there is a total ordering (e.g., based on ids) on $C(v) = \{ u_1, \dots, u_k \}$ such that $u_1 < \dots < u_k$.

We inductively prove the following statements over the height of a tree, with the base cases as the set of trees of height 1 (and the label of root can either be 0 or 1).
\begin{enumerate}
    \item Algorithm $\ref{alg:tree-to-line}$ takes in any rooted tree $T_b = (l(\cdot), p, V, E)$ where $l(p) = 0$, and outputs a (directed) path graph $P = (V, E')$ where $\mathrm{Beg}(P) = p$ and $\mathrm{End}(P)$ is $\min(C(p))$ for $|V| \geq 2$.
    \item Algorithm $\ref{alg:tree-to-line}$ takes in any rooted tree $T_b = (l(\cdot), p, V, E)$ where $l(p) = 1$, and outputs a (directed) path graph $P = (V, E')$ where $\mathrm{Beg}(P)$ is $\max(C(p))$ and $\mathrm{End}(P) = p$ for $|V| \geq 2$.
\end{enumerate}
Consider the base case of any rooted tree $T_b = (l(\cdot), p, V, E)$ where $l(p) = 0$ and $|V| \geq 2$ and $H(T_b) = 1$. Let $C(p) = \{ p_1, \dots, p_q \}$ such that $p_1 < \dots < p_q$ for any $q \geq 1$. Algorithm $\ref{alg:tree-to-line}$ adds the following edges, resulting in a path $P$, where $\mathrm{Beg}(P) = p$ and $\mathrm{End}(P)$ is $\min(C(p))$.
\begin{enumerate}
    \item The (directed) edge $(p, p_q)$ is added.
    \item For $1 \leq r < q$, the (directed) edge $(p_{r+1}, p_r)$ is added.
\end{enumerate}

Similarly, consider the base case of any rooted tree $T_b = (l(\cdot), p, V, E)$ where $l(p) = 1$ and $|V| \geq 2$ and $H(T_b) = 1$. Let $C(p) = \{ p_1, \dots, p_q \}$ such that $p_1 < \dots < p_q$ for any $q \geq 1$. Algorithm $\ref{alg:tree-to-line}$ adds the following edges, resulting in a path $P$, where $\mathrm{Beg}(P) = \max(C(p))$ and $\mathrm{End}(P)$ is $p$.
\begin{enumerate}
    \item The (directed) edge $(p_1, p)$ is added.
    \item For $1 \leq r < q$, the (directed) edge $(p_{r+1}, p_r)$ is added.
\end{enumerate}


We introduce some notation for the inductive cases, where $T = (l(\cdot), p, V, E)$ is any rooted tree of $H(T) = h+1$ for any $h \geq 1$. For any node $p$, let $C'(p) = \{ u \mid u \in C(p) \wedge C(u) \neq \emptyset\}$. Let $T_w = (l(\cdot), w, V_d(w), E_d(w))$ be the subtree rooted at $w$ and $l(w) = l(p) + \mathrm{dist}(p, w)\mod 2$, where $V_d(w) = \{ u \mid \mathrm{dist}(u, w) < \mathrm{dist}(u, p) \}$ and $E_d(w) = \{ (u, v) \mid (u, v) \in E \wedge \{u, v\} \in V_d(w) \}$. By induction, Algorithm $\ref{alg:tree-to-line}$ outputs a path, denoted by $\mathcal{A}(T_v)$, for each subtree $T_v$ rooted at $v \in C'(p)$.

Consider any rooted tree $T$ where $l(p) = 0$. Algorithm $\ref{alg:tree-to-line}$ adds the following (directed) edges, resulting in a (directed) path $P = (V, E')$ where $\mathrm{Beg}(P) = p$ and $\mathrm{End}(P)$ is $\min(C(p))$. 
\begin{enumerate}
    \item Either $(p, \mathrm{Beg}(\mathcal{A}(T_{p_q}))$ or $(p, p_q)$ is added based on whether $p_q \in C'(p)$ or $p_q \notin C'(p)$. By induction, $\mathrm{Beg}(\mathcal{A}(T_{p_q}))$ is $\max(C(p_q))$ and $\mathrm{End}(\mathcal{A}(T_{p_q}))$ is $p_q$, if $p_q \in C'(p)$, as $l(p_q) = 1$.
    \item Similarly, if there exists an $r$ such that, $1 \leq r < q$, we consider these exhaustive set of cases.
    \begin{enumerate}
        \item If $p_r, p_{r+1} \in C'(p)$, then the edge $(\mathrm{End}(\mathcal{A}(T_{p_{r+1}}), \mathrm{Beg}(\mathcal{A}(T_{p_r}))$ is added. By induction, $\mathrm{Beg}(\mathcal{A}(T_{p_r}))$ is $\max(C(p_r))$ and $\mathrm{End}(\mathcal{A}(T_{p_{r+1}}))$ is $p_{r+1}$, if $p_r, p_{r+1} \in C'(p)$, as $l(p_{r}) = 1$ and $l(p_{r+1}) = 1$.
        \item If $p_{r+1} \in C'(p)$ and $p_{r} \notin C'(p)$, then the edge $(\mathrm{End}(\mathcal{A}(T_{p_{r+1}}), p_r)$ is added. By induction, $\mathrm{End}(\mathcal{A}(T_{p_{r+1}}))$ is $p_{r+1}$, if $p_{r+1} \in C'(p)$, as $l(p_{r+1}) = 1$.
        \item If $p_{r+1} \notin C'(p)$ and $p_{r} \in C'(p)$, then edge $(p_{r+1}, \mathrm{Beg}(\mathcal{A}(T_{p_r}))$ is added. By induction, $\mathrm{Beg}(\mathcal{A}(T_{p_r}))$ is $\max(C(p_r))$, if $p_r \in C'(p)$, as $l(p_{r}) = 1$.
        \item If $p_r, p_{r+1} \notin C'(p)$, then edge $(p_{r+1}, p_r)$ is added. (Same as base case where $l(p) = 0$.)
    \end{enumerate}
\end{enumerate}
Similarly, consider any rooted tree $T$ where $l(p) = 1$. Algorithm $\ref{alg:tree-to-line}$ adds the following (directed) edges, resulting in a (directed) path $P = (V, E')$ where $\mathrm{Beg}(P) = \max(C(p))$ and $\mathrm{End}(P)$ is $p$. 
\begin{enumerate}
    \item Either $(\mathrm{End}(\mathcal{A}(T_{p_1}), p)$ or $(p_1, p)$ is added based on whether $p_1 \in C'(p)$ or $p_1 \notin C'(p)$. By induction, $\mathrm{End}(\mathcal{A}(T_{p_1}))$ is $\min(C(p_1))$ and $\mathrm{Beg}(\mathcal{A}(T_{p_1}))$ is $p_1$, if $p_1 \in C'(p)$, as $l(p_1) = 0$.
    \item Similarly, if there exists an $r$ such that, $1 \leq r < q$, we consider these exhaustive set of cases.
    \begin{enumerate}
        \item If $p_r, p_{r+1} \in C'(p)$, then the edge $(\mathrm{End}(\mathcal{A}(T_{p_{r+1}}), \mathrm{Beg}(\mathcal{A}(T_{p_r}))$ is added. By induction, $\mathrm{Beg}(\mathcal{A}(T_{p_r}))$ is $p_r$ and $\mathrm{End}(\mathcal{A}(T_{p_{r+1}}))$ is $\min(C(p_{r+1}))$, if $p_r, p_{r+1} \in C'(p)$, as $l(p_{r}) = 0$ and $l(p_{r+1}) = 0$.
        \item If $p_{r+1} \in C'(p)$ and $p_{r} \notin C'(p)$, then the edge $(\mathrm{End}(\mathcal{A}(T_{p_{r+1}}), p_r)$ is added. By induction, $\mathrm{End}(\mathcal{A}(T_{p_{r+1}}))$ is $\min(C(p_{r+1}))$, if $p_{r+1} \in C'(p)$, as $l(p_{r+1}) = 0$.
        \item If $p_{r+1} \notin C'(p)$ and $p_{r} \in C'(p)$, then edge $(p_{r+1}, \mathrm{Beg}(\mathcal{A}(T_{p_r}))$ is added. By induction, $\mathrm{Beg}(\mathcal{A}(T_{p_r}))$ is $p_r$, if $p_r \in C'(p)$, as $l(p_{r}) = 0$.
        \item If $p_r, p_{r+1} \notin C'(p)$, then edge $(p_{r+1}, p_r)$ is added. (Same as base case where $l(p) = 1$.)
    \end{enumerate}
\end{enumerate}
\end{proof}
\section{Self-Stabilization Analysis} \label{sec:ss-analysis}

\subsection{Basic Definitions}

\begin{definition} \label{def:orig-alg}
We use $\mathrm{M}_u(\mathcal{A}_0)$ to represent the (internal) memory (including the channel buffer) that a node $u$ uses for executing algorithm $\mathcal{A}_0$ respectively. Moreovoer, we use $\mathrm{R}(\mathcal{A}_0)$ for the (best-known) upper bound for the (self-stabilization) convergence of the algorithm $\mathcal{A}_0$ (when it is independently executed).
\end{definition}

\begin{definition} \label{def:diff-addr-vars}
We use the term \emph{flyover-related variables} to refer to the set of variables in Section \ref{sec:flyover-code}. Similarly, we use the term \emph{advice-related variables} to refer to the set of variables in Section \ref{subsec:distrib-TtP}.
\end{definition}

\begin{definition} \label{def:flush}
We say that a node $u$ executes a \emph{flush operation} on a set of node ids, if it stores them in any set of address variables in $\mathrm{M}_u(\mathcal{A}_0)$ (i.e., memory allocated to algorithm $\mathcal{A}_0$).
\end{definition}

\begin{definition} \label{def:well-formed-adv}
Let $\mathit{Adv}(u)$ be an advice message at node $u$. Moreover, when we refer to $\mathit{Adv}(u).x = y$, then $x = y$ has appeared in the advice. $\mathit{Adv}(u)$ is said to be \emph{well-formed} if the following conditions hold.
\begin{enumerate}
    \item $\mathit{Adv}(u).\mathit{par} \in \{ \bot, \mathrm{id} \}$, where $\mathrm{id}$ is present in its memory.
    \item If $\mathit{Adv}(u).\mathit{par} = \bot$, then $\mathit{Adv}(u).\mathit{dist} = 0, \mathit{Adv}(u).\mathit{vID} = 1, \mathit{Adv}(u).\textit{c-dist} = 0$; else $\mathit{Adv}(u).\mathit{dist} > 0, \mathit{Adv}(u).\mathit{vID} > 1$ and $\mathit{Adv}(u).\textit{c-dist} > 0$.
    \item $\mathit{Adv}(u).\mathit{vID} \in \mathbb{N}$ and $\mathit{Adv}(u).\textit{c-par} \in \mathbb{N}$.
\end{enumerate}
\end{definition}

\begin{definition} \label{def:dual-state}
We say that a node $u$ is in \emph{dual-state} if $u.S \neq \emptyset$. Furthermore, if node $u$ updates a previously non-empty $u.S$ to $\emptyset$, then we say that node $u$ has \emph{exited from the dual-state}.
\end{definition}

\begin{definition} 
A node $u$ is said to \emph{process} an advice message if $u$ enters Line 6 in Algorithm \ref{alg:local-cert}. An advice message, for any node $u$, is said to be \emph{effective} if node $u$ enters dual-state before $u.t \leq 1$. We say that a node $u$ is \emph{attentive}, or \emph{ready to receive advice}, if $u.S = \emptyset \land u.\mathit{exit} = 0 \land u.t = 0$.
\end{definition}

\begin{definition}\label{def:bidirected}
\emph{(Delegate-after-Reversal; DR Primitive)} A node $u$ is said to \emph{use DR primitive} for algorithm $\mathcal{A}_0$ in round $r$, if the following conditions hold during the execution of algorithm $\mathcal{A}_0$. 
\begin{enumerate}
    \item If node $u$ is required to \emph{delegate} an edge $(u, w)$ to node $v$, i.e., send the id of node $w$ to node $v$, for e.g., via some message $\langle \mathrm{msg}, w \rangle$, then (1) node $u$ first sends $\langle \mathrm{Rev}, v, \mathrm{msg}, w \rangle$ to node $w$, and (2) once node $w$ receives $\langle \mathrm{Rev}, v \rangle$, node $w$ sends its own id to node $v$ (for e.g., by message $\langle \mathrm{msg}, w \rangle$).
    \item If node $u$ is required to send some message $\langle \mathrm{msg}, w_1, \dots, w_2, \dots, w_t \rangle$ for some $t \geq 1$ to node $v$, then (1) node $u$ first sends $\langle \mathrm{Rev}, v, \mathrm{msg}, w_1, \dots, w_2, \dots, w_t \rangle$ to node $w_1$, and for all $1 < j \leq t$, node $u$ sends $\langle \mathrm{Rev}, v \rangle$ to node $w_j$, and (2) once node $w_1$ receives $\langle \mathrm{Rev}, v, \mathrm{msg}, w_1, \dots, w_2, \dots, w_t \rangle$, node $w_1$ sends $\langle \mathrm{msg}, w_1, \dots, w_2, \dots, w_t \rangle$ to node $v$, and for all $1 < j \leq t$, once node $w_j$ receives $\langle \mathrm{Rev}, v \rangle$, node $w_j$ sends its own id to node $v$.
\end{enumerate}
The message complexity of DR primitive is $O(1)$ factor to that of ``Delegate'' primitive (see, e.g., \cite{KOUTSOPOULOS2017408}). 
\end{definition}

\begin{definition} \label{def:s-robust}
The network is \emph{robust against a supervisor}, or \emph{s-robust}, if these two conditions hold.
\begin{enumerate}
    \item \emph{(Sybil Resistance)} For each node $u$, and for any node id $v$ present in node $u$'s advice message but \emph{not} present in node $u$'s memory or channel, node $u$ neither assigns $v$ to any of its variables (i.e. store in its memory) nor sends $v$ to any other id present in its memory or channel (i.e., send it to other nodes).
    \item \emph{(Robustness)} If all nodes do not form a single, correctly configured flyover, then all nodes exit the dual-state (thereby rejecting any ``supervisor-induced'' overlay) in $O(\log n)$ rounds.
    %
\end{enumerate}
\end{definition}



\subsection{Flyover-Related Definitions}

\begin{definition} \label{def:backbone}
We say that a set of nodes $B \subseteq V$ is a \emph{backbone}, where $|B| \geq 2$, if the following conditions hold. For ease of exposition, wlog, we denote the set $B$ as $\{ v_1, v_2, \dots \}$.
\begin{enumerate}
    \item \emph{(Wings)} $(v_1.L = \emptyset \lor v_1.S_l(1) \notin B)$ and $(v_{|B|}.R = \emptyset \lor v_{|B|}.S_r(1) \notin B)$.
    \item \emph{(Spine)} $\forall i \in [|B| - 1], v_i.S_r(1) = v_{i+1} \land v_{i+1}.S_l(1) = v_{i}$.
\end{enumerate}
If for all $ u \in V \setminus B, u \cup B$ is not a backbone, then $B$ is said to be a \emph{maximal backbone}. A maximal backbone is said to be \emph{winged maximal backbone} if $v_1.L \neq \emptyset \lor v_{|B|}.R \neq \emptyset$.
\end{definition}

\begin{definition} \label{def:ouroboros}
We say that a set of nodes $B \subseteq V$ is an \emph{ouroboros} if any of the following conditions hold. For ease of exposition, wlog, we denote the set $B$ as $\{ v_1, v_2, \dots \}$.
\begin{enumerate}
    \item \emph{(Perfect Ring)} $(v_1.S_l(1) = v_{|B|} \land v_{|B|}.S_r(1) = v_1)$ and $\forall i, v_i.S_r(1) = v_{i+1} \land v_{i+1}.S_l(1) = v_{i}$.
    \item \emph{(Stylish Ring)} $(v_1.S_l(1) \in B \lor v_{|B|}.S_r(1) \in B)$ and $\forall i, v_i.S_r(1) = v_{i+1} \land v_{i+1}.S_l(1) = v_{i}$.
\end{enumerate}
If for all $ u \in V \setminus B, u \cup B$ is not an  ouroboros, then $B$ is said to be a \emph{maximal ouroboros}.
\end{definition}

\begin{definition} \label{def:lost-nodes}
An node $u$ in a dual-state, is said to be \emph{lost}, if it is not in any backbone or ouroboros.
\end{definition}

\begin{definition} \label{def:flyover}
We say that a backbone $B \subseteq V$ is a \emph{flyover} if the following conditions hold. Let $v_1 \in B$ such that $(v_1.L = \emptyset \lor v_1.S_l(1) \in V \setminus B)$. For all $i \in [|B|-1]$, let $v_{i+1} \in B$ such that $v_i.S_r(1) = v_{i+1}$.
\begin{enumerate}
    \item \emph{(Virtual IDs)} $\forall i \in [|B|-1]$, $v_{i+1}.\mathit{vID} = v_{i}.\mathit{vID} + 1$.
    \item \emph{(Hypercube)} $\forall i \in [|B| - 1], \forall j \in \left[\lfloor \log_2(|B|-i) \rfloor + 1\right], (v_i.S_r(j) = v_{i+2^{j-1}} \land v_{i+2^{j-1}}.S_l(j) = v_i)$.
\end{enumerate}
\end{definition}

\begin{definition} \label{def:corr-config-backbone}
Consider a maximal backbone $B = \{ v_1, v_2, \dots \} \subseteq V$ where $(v_1.L = \emptyset \lor v_1.S_l(1) \in V \setminus B)$ and $\forall i \in [|B|-1], v_i.S_r(1) = v_{i+1}$. Let $G^*_B = (B, E^*_B)$ such that  $(u, v) \in E^*_B \iff (u = \mathit{succ}_B(v) \lor v = \mathit{succ}_B(u))$, where $\mathit{succ}_B(u) = \min(\{ w.\mathit{id} \mid (w \in B) \land (w.\mathit{id} > u.\mathit{id}) \})$. Let $T^*_B$ be the rooted spanning tree in the graph $G^*_B$ with node $v_1$ as the root. Then, $B$ is said to be \emph{correctly configured} if the following conditions hold.
\begin{enumerate}
    \item $\forall i \in [|B|], v_{i}.\mathit{vID} = i$ and $v_{i}.\mathit{flyID} = v_1.\mathit{id}$. 
    \item $\forall i \in [|B| - 1], v_i.\textit{c-par} = p_i.\mathit{vID}$ where $p_i$ denotes the parent of $v_i$ in $T^*_B$.
    \item $\forall i \in [|B|]$, $E^*_B(v_i) \subseteq v_i.\textit{c-ids}$ where $E^*(v_i)$ denotes the edge-set of $v_i$ in $G^*_B$. 
    \item $\forall i \in [|B|], v_i.\textit{c-dist} = \mathit{dist}_{T^*_B}(v_i, v_1)$ where $\mathit{dist}_{T^*_B}(u, v)$ is the distance between $v_1$ and $v_i$ in $T^*_B$.
\end{enumerate}
\end{definition}

\begin{definition} \label{def:prec-dual-state}
A node $u$, with $u.S \neq \emptyset$, in round $r$, is said to be in a \emph{precarious} dual-state if any of the following conditions are satisfied.
\begin{enumerate}
    \item $u$ is a lost node.
    \item $u$ belongs to a maximal ouroboros.
    \item $u$ belongs to either a maximal backbone of size less than $n$, or maximal backbone of size equal to $n$ but some node sets its $\mathtt{exit}$ variable to 1 in any round $r' \in \{r, r+1, \dots r+3\lceil \log n \rceil \}$.
\end{enumerate}
\end{definition}

\begin{definition} \label{def:stable-network}
The network is said to be \emph{stable} in any round $r$ if for each node $u \in V, u.S \neq \emptyset$ but $u$ is not in a precarious dual-state.
\end{definition}

\subsection{Helpful Observations}

\begin{remark} \label{flyover-inv}
For any node $u$, the following invariants are ensured (via $\mathbf{BasicChecks}$ function).
\begin{enumerate}
    \item $u.S = \emptyset \implies u.\mathit{flyID} = u.\mathit{id} \land u.\textit{c-ids} = \emptyset$.
    \item $u.L = \emptyset \land u.R \neq \emptyset \implies u.\mathit{vID} = 1 \land u.\mathit{flyID} = u.\mathit{id} \land u.\textit{c-dist} = 0$.
    \item $u.L \neq \emptyset \implies u.\textit{c-dist} > 0 \land u.\textit{vID} > 1$.
    \item $(u.S \neq \emptyset \land u.\mathit{vID} > 1) \implies u.\textit{c-par}$ can be ``reached,'' i.e., via left or right neighbors.
    \item $u.\text{\emph{c-ids}} \neq \emptyset \implies (|u.\text{\emph{c-ids}}| = 1) \lor (|u.\text{\emph{c-ids}}| = 2 \land \min(u.\text{\emph{c-ids}}) < u.\mathit{id} < \max(u.\text{\emph{c-ids}}))$.
\end{enumerate}
\end{remark}


\begin{remark} \label{rem:flushflyover}
Due to $\mathbf{FlushFlyover}$ function, whenever any node $u$ exits from dual-state, all ids part of node $u$'s flyover variables (including $u.\mathit{flyID}$) receive $\langle \text{RejFlyover} \rangle$ message. For any node $u$, if some node id $v \neq u$ is stored in the set $\{ u.\mathit{flyID} \} \cup u.\mathit{S} \cup u.\textit{c-ids}$, it can only be removed from that set in the $\mathbf{FlushFlyover}$ function.
\end{remark}



\begin{remark} \label{rem:attentive-bot}
If a node $u$ executes $\mathbf{SnapshotReq}$ function (or, node $u$ is \emph{attentive}) in round $r > 1$, then $u.S = \emptyset \land u.\mathit{exit} = 0 \land u.t = 0$ in that round. Since $u.S = \emptyset \land u.\mathit{exit} = 0$ condition is satisfied, and $u.S$ and $u.\mathit{exit}$ can be set to $\emptyset$ and $0$ (resp.) only in $\mathbf{FlushFlyover}$ function, node $u$ sends $\bot$ to all ids in its memory (and channel) in round $r$ due to $\mathbf{TestFlyoverMetadata}$ function. This is because either, $\mathbf{FlushFlyover}$ function was executed in round $r$ (in which case, $u.\mathit{vID}$ is set to 0), or, both $u.S$ and $u.\mathit{exit}$ were already equal to $\emptyset$ and $0$ (resp.) at the beginning of round $r$, in which case, $u.\mathit{vID}$ was set to $0$ in $\mathbf{BasicChecks2}$ function (in Algorithm \ref{alg:init2}) in round $r-1$.
\end{remark}

\begin{claim}
Let $B = \{ v_1, \dots, v_{|B|} \}$ be a maximal backbone where $(v_1.L = \emptyset \land v_{|B|}.R = \emptyset)$ and $(\forall i \in [|B| - 1], v_{i}.S_r(1) = v_{i+1})$ and $(\forall i \in [|B|], v_i.\mathit{vID} = i \land v_i.\mathit{flyID} = v_1)$. Consider the following graphs.
\begin{itemize}
    \item Let $G^*_B = (B, E^*_B)$ such that $(u, v) \in E^*_B \iff (u = \mathit{succ}_B(v) \lor v = \mathit{succ}_B(u))$, where $\mathit{succ}_B(u) = \min(\{ w.\mathit{id} \mid (w \in B) \land (w.\mathit{id} > u.\mathit{id}) \})$.
    \item Let $G'_B = (B, E'_B)$ such that $(u, v) \in E'_B \iff (u.\mathit{vID} = v.\textit{c-par} \land v.\textit{c-dist} = u.\textit{c-dist} + 1) \lor (v.\mathit{vID} = u.\textit{c-par} \land v.\textit{c-dist} = u.\textit{c-dist} - 1)$.
\end{itemize}
If for all $i \in [|B|] \setminus \{ 1 \}, \exists u_i \in B, v_i.\textit{c-par} = u_i.\mathit{vID}$ where $(u_i.\textit{c-dist} = v_i.\textit{c-dist} - 1)$, then $G'_B$ is $G^*_B$.
\end{claim}
\begin{proof}
Due to Remark \ref{flyover-inv}, the following invariants are (locally) maintained. For each node $u \in B$,
\begin{enumerate}
    \item (Tree Distance) $(u.L = \emptyset \implies u.\textit{c-dist} = 0)$ and $(u.L \neq \emptyset \implies u.\textit{c-dist} > 0)$.
    \item (Degree Constraint; Locally Sorted) $E'(u) \neq \emptyset \implies (|E'(u)| = 1) \lor (|E'(u)| = 2 \land \min(|E'(u)|) < u.\mathit{id} < \max(|E'(u)|))$, where $E'(u) = \{ v \mid (u, v) \in E'_B \}$.
\end{enumerate}
These invariants ensure there is only one node $u$ with $u.\textit{c-dist} = 0$, and that node is $v_1$. Moreover, the degree of any node in $G'$ is at most 2; if a node's degree is 2, then its own id is neither greater nor less than both the node ids of its neighbors. Given the premise of the lemma, for any node $u \in V \setminus \{ v_1 \}$, there is a node $p(u)$ such that $u.\textit{c-par} = p(u).\mathit{vID}$ and $p(u).\textit{c-dist} = u.\textit{c-dist} - 1$.

First, we show that $G'_B$ forms a spanning tree (with degree at most 2, i.e., a path) using a key idea behind the classic $O(\log n)$-bits proof-labelling scheme of a \emph{spanning tree} (see, e.g., \cite{AfekWDAG90, KormanDC10, Feuilloley-DMTCS21}). In that scheme, each node is given a ``certificate'' consisting of the root node id, a parent node id (among its neighbors) and the distance from the root (in a rooted spanning tree), where the ``root'' is the (only) node with its distance equal to 0 (and has no parent). Then, each node examines the certificates of its neighbors in the graph, and checks the distance (i.e., from its parent) and the root ids (i.e., whether they are same). If any certificate is incorrectly assigned, then there exists (at least) one node that ``rejects'' its certificate. Otherwise, if each certificate is correctly assigned (with respect to some rooted spanning tree), then each node ``accepts'' its certificate.

We use proof by contradiction to show that $G'$ is indeed $G^*$. First, there cannot be a cycle in $G'$, as the distance checks (i.e., from a node $u$ to its parent $p(u)$) cannot be satisfied at all nodes (cf. Section \ref{subsec:local-cert}). Moreover, $G'$ cannot be a spanning forest because all the nodes, via flyover id, can verify that there is only one ``root node,'' i.e., node $v_1$ with distance $v_1.\textit{c-dist} = 0$. Thus, $G'$ forms a spanning tree. Combining the second invariant mentioned earlier, i.e., each node has degree at most 2, the tree must be a path. Finally, that path has to be the sorted-path $G^*$ due to the second invariant because every node checks that its locally sorted with its neighbors in $G'$.
\end{proof}

\begin{claim} \label{claim:dr-prim}
Let $p \in B$, $q \notin B$ and $p \in \mathrm{M}_q(\mathcal{A}_0)$ in any round $r$, where $B$ is a backbone. If node $q$ uses the DR primitive (as in Definition \ref{def:bidirected}) in round $r$, then at least one of the following statements is true:
\begin{itemize}
    \item either edge $(q, p)$ exists in round $r+1$,
    \item or edge $(p', q')$ where $p' \in B$ and $q' \notin B$ exists, by round $r+2$.
\end{itemize}
\end{claim}
\begin{proof}
Let us first consider the ``normal'' execution of algorithm $\mathcal{A}_0$ (i.e, without DR primitive) in round $r$. As it preserves weak connectivity, there can be two (exhaustive) cases in round $r+1$. 
\begin{enumerate}
    \item $(q, p)$ still exists.
    \item Some path $P = (q, v_0, \dots, v_t)$ exists, where $v_0 \in M_{q}(\mathcal{A}_0)$, $v_t = p$ and $\forall i \in [t], v_i \notin M_{q}(\mathcal{A}_0)$. 
\end{enumerate}
In the path $P$, color a node blue if it is present in $B$; otherwise, color it red. Consider the edges (regardless of their directions): $(q, v_0), \dots, (v_{t-1}, v_{t})$. There is at least one pair of adjacent nodes colored differently. If node $q$ runs algorithm $\mathcal{A}_0$ \emph{with} the DR primitive in round $r$, due to reversal of edges, in either round $r+1$ or $r+2$, there exists an edge $(p', q')$ where $p' \in B$ and $q' \notin B$.
\end{proof}

\subsection{Connectivity Preservation}

\begin{lemma} \label{lemma:conn-global}
If the initial communication graph is weakly connected, then the communication graph in any round $r \geq 1$ is also weakly connected.
\end{lemma}
\begin{proof}
Consider a node id $v$ in some node $u$ (whether in its channel or address variables, i.e., $v$ is present in a message sent to $u$, or $v$ is assigned to some variable in $u$). If $v$ is present in $u$, then $u$ and $v$ are (weakly) connected. On the other hand, if $v$ is ``delegated,'' i.e., $u$ sends the id of $v$ to another node $w$, then we need to show that $u$ and $v$ are (weakly) connected via some path $(u, v_1, \dots, v_t)$ where $v_t = v$ for some $t \geq 1$. Thus, we focus on all the possible cases in our algorithms where a node id could be delegated, and a node id is present in a received message, to show that the network remains weakly connected in the next round.

\begin{itemize}
    \item First, by design, every node runs an algorithm $\mathcal{A}_0$ (in parallel, for converging to target topology) that preserves connectivity (i.e., following the universal overlay primitives for edge manipulations). Moreover, if nodes use the DR primitive (instead of ``Delegate'' primitive) while running algorithm $\mathcal{A}_0$, connectivity is still preserved (as the direction of delegated edge is only opposite to the intended direction, but is subsequently reversed; see Definition \ref{def:bidirected}) after edge manipulations by algorithm $\mathcal{A}_0$ in any round.
    \item In the $\mathbf{TestFlyoverConstruction}, \mathbf{TestFlyoverMetadata}$ and $\mathbf{TestConnCertificate}$ functions, if node $u$ sends $v$ to some other node, then $v$ also remains in node $u$.
    \item We analyze the response functions that process the received messages.
    
    In $\mathbf{R\_TestFlyoverConstruction}$ function, node $u$ processes two types of messages. For any message $\langle \mathit{msg}, \mathit{sen} \rangle$, $\mathit{sen}$ is stored in, either $u.S$ or $\mathrm{M}_u(\mathcal{A}_0)$. For any message $\langle \mathit{msg}, w, i, \mathit{sen} \rangle$, $w$ and $\mathit{sen}$ are stored in either $u.S$ or $\mathrm{M}_u(\mathcal{A}_0)$.
    
    In $\mathbf{R\_TestFlyoverMetadata}$ function, for any message $\langle \text{TestFlyID}, \mathit{flyID} \rangle$, there are two cases: $\mathit{flyID}$ is stored in either $u.\mathit{flyID}$ or $\mathrm{M}_u(\mathcal{A}_0)$.
    
    In $\mathbf{R\_TestConnCertificate}$ function, for any message $\langle \text{TestCert}, w, \mathit{vID}, \mathit{dist} \rangle$, there are three cases: (1) $w$ is sent to node $u' \in u.S$, in which case, there is a path $(u, u', w)$, or (2) $u$ stores $w$ in $u.\textit{c-ids}$, or (3) $u$ stores $w$ in $\mathrm{M}_u(\mathcal{A}_0)$; for any message $\langle \text{IntroCert}, w \rangle$, $u$ stores $w$ in $u.\textit{c-ids}$.
    \item We turn our attention to distributed tree-to-path algorithms (cf. Section \ref{subsec:distrib-TtP}). In $\mathbf{GetAdvice}$ function, for any message $\langle \text{Intro}, v\rangle$, node $u$ stores $v$ in $\mathrm{M}_u(\mathcal{A}_0)$. In $\mathbf{CertifyTree}$ function, for any message $\langle \text{TestAdvice}, \mathit{dist}, v\rangle$, node $u$ stores $v$ in $\mathrm{M}_u(\mathcal{A}_0)$. Similarly, in $\mathbf{LocalTransform}$ function, for any message $\langle \text{Verified}, \mathit{msg}, v \rangle$, node $u$ stores $v$ in $\mathrm{M}_u(\mathcal{A}_0)$. In $\mathbf{JoinPath}$ function, for any message $\langle \mathit{msg}, v \rangle$, there are two cases: node $u$ stores $v$ either in $u.S$ or $\mathrm{M}_u(\mathcal{A}_0)$.
\end{itemize}
\end{proof}

\subsection{Flyover Construction}

\begin{lemma} \label{lemma:fly-const}
Consider any backbone $B = \{ v_1, \dots, v_{|B|} \}$, where $|B| \geq 3$ and $(v_1.L = \emptyset \lor v_1.S_l(1) \notin B)$ and $\forall i \in [|B|-1], v_i.S_r(1) = v_{i+1}$, in any round $r$. For any $i \in [k]$ and $k = \lfloor \log_2(|B| - 1) \rfloor$, if all nodes in $B$ have their $\mathtt{exit}$ variable set to $0$ in all rounds $r, r+1, \dots, r+i'$ where $i' \geq i$, then in any round $r+i'$, the following statements hold true.
\begin{enumerate}
    \item \emph{(Virtual IDs)} $\forall j \in [|B| - 1], v_{j+1}.\mathit{vID} = v_{j}.\mathit{vID} + 1$.
    \item \emph{(Endpoints)} $|v_1.R| \geq (i+1)$ and $|v_{|B|}.L| \geq (i+1)$.
    \item \emph{(Hypercube Edges)} $\forall j \in [i], v_1.S_r(j) = v_{1+2^j}$ and $\forall j \in [i], v_{|B|}.S_l(j) = v_{|B|-2^j}$.
\end{enumerate}
\end{lemma}
\begin{proof} If the first statement is true in round $r$, then it remains true until some node exits dual-state; otherwise, some node sets its \texttt{exit} variable to 1 due to $\mathbf{TestFlyoverMetadata}$ function. Recall that \texttt{vID} is assigned when a node receives an advice (see Algorithm \ref{alg:local-cert}). Moreover, for any node $u$, once $u.\mathit{vID}$ is assigned, then it can only be updated in $\mathbf{FlushFlyover}$ function (which is executed only after $u.\mathit{exit}$ = 1). As no node in $B$ sets its $\mathtt{exit}$ variable to 1 until round $r+i'$ for any $i' \geq i$ and $i \in [k]$, the statement holds true in round $r+i'$ for any $i' \geq 1$.

For latter statements, we use induction over the size of a backbone. Let any backbone of size $s$ be $B_s = \{ u_1, \dots, u_s\}$, where $(u_1.L = \emptyset \lor u_1.S_l(1) \notin B_s)$ and $\forall i \in [s-1], u_i.S_r(1) = u_{i+1}$. For the base case, consider $s = 3$ where $B_s = \{ u_1, u_2, u_3 \}$. In round $r$, node $u_2$ sends $\langle \text{FlyConst-R},u_1,1,u_2 \rangle$ and $\langle \text{FlyConst-L},u_3,1,u_2 \rangle$ messages to node $u_3$ and $u_1$ respectively. Since the nodes do not set their $\mathtt{exit}$ variables to $1$ in both rounds $r$ and $r+1$, node $u_1$ assigns $u_3$ to $u_1.S_r(2)$ and node $u_3$ assigns $u_1$ to $u_3.S_l(2)$ in round $r+1$ (if they have not already done so), which also implies that $|u_1.R| \geq 2$ and $|u_3.L| \geq 2$ in round $r+1$. Recall that for any node $u$, for any $j \geq 1$, once $u.S_l(j)$ or $u.S_r(j)$ is assigned, then it can only removed (or updated) in $\mathbf{FlushFlyover}$ function (which is executed only after $u.\mathit{exit} = 1$). As no node in $B$ sets its $\mathtt{exit}$ variable to 1 until round $r+i'$ for any $i' \geq i$ and $i \in [k]$ where $k = 1$ (for $s = 3$), these statements hold true in round $r+i'$ for any $i' \geq 1$.

For inductive case, we assume that the statements are true for any backbone of size $(s - 1) \geq 3$, and then show that they hold true for any backbone $B_{s} = \{ u_1, \dots, u_s \}$. To that end, we split the proof into two cases, for any $p > 1$ and $p \in \mathbb{N}$, either $s \neq 2^p + 1$ or $s = 2^p + 1$.

In the first case, where $s \neq 2^p + 1, k = \lfloor \log_2(s - 1) \rfloor = \lfloor \log_2(s - 2) \rfloor$. Consider the two backbones $B^1_s = \{ u_1, u_2, \dots, u_{s-1}\}$ and $B^2_s = \{ u_{2}, u_{3}, \dots, u_{s}\}$. For any $i \in [k]$, if all nodes in $B_s$ have their $\mathtt{exit}$ variable set to $0$ in all rounds $r, \dots, r+i'$ where $i' \geq i$, then in any round $r+i'$, the statements hold true, independently, for $B^1_s$ and $B^2_s$. Thus, they also hold true for backbone $B_{s}$.

In the second case, where $s = 2^p + 1, k = \lfloor \log_2(s - 1) \rfloor = \lfloor \log_2(s - 2) \rfloor + 1$. Consider the two backbones $B^1_s = \{ u_1, u_2, \dots, u_{s'}\}$ and $B^2_s = \{ u_{s'}, u_{s'+1}, \dots, u_{s}\}$ where $s = 2s'-1$. For any $i \in [k-1]$, if all nodes in $B_s$ have their $\mathtt{exit}$ variable set to $0$ in rounds $r, \dots, r+i'$ where $i' \geq i$, then in any round $r+i'$, the statements hold true, independently, for $B^1_s$ and $B^2_s$, which implies, in round $r+k-1$, $u_1.S_r(k) = u_{s'}$ and  $u_{s'}.S_l(k) = u_1$ and $u_{s}.S_l(k) = u_{s'}$ and $u_{s'}.S_r(k) = u_s$. In round $r+k-1$, node $u_{s'}$ sends $\langle \text{FlyConst-R},u_1,k,u_{s'} \rangle$ and $\langle \text{FlyConst-L},u_{s},k,u_{s'} \rangle$ messages to node $u_s$ and $u_1$ respectively. Since the nodes do not set their $\mathtt{exit}$ variables to $1$ in rounds $r+k-1$ and $r+k$, node $u_1$ assigns $u_s$ to $u_1.S_r(k+1)$ and node $u_s$ assigns $u_1$ to $u_s.S_l(k+1)$ in round $r+k$ (if they have not already done so), which implies that $|u_1.R| \geq k+1$ and $|u_3.L| \geq k+1$ in round $r+k$. Recall that for any node $u$, for any $j \geq 1$, once $u.S_l(j)$ or $u.S_r(j)$ is assigned, then it can only removed (or updated) in $\mathbf{FlushFlyover}$ function (which is executed only after $u.\mathit{exit} = 1$). As no node in $B_s$ sets its $\mathtt{exit}$ variable to 1 until round $r+i'$ for any $i' \geq i$ and $i \in [k]$ where $k = \lfloor \log_2(s - 1) \rfloor$, these statements hold true in round $r+i'$ for any $i' \geq 1$.
\end{proof}

\begin{corollary} \label{corr:flyover-const}
Consider any backbone $B$ in any round $r$. If all nodes in $B$ have their $\mathtt{exit}$ variable set to $0$ in all rounds $r, r+1, \dots, r+i$ where $i \geq \lfloor \log_2(|B| - 1) \rfloor$, then in any round $r+i$, $B$ is a flyover.
\end{corollary}

\subsection{Fast Information Dissemination}

\begin{lemma} \label{lemma:inf-prop}
Consider any backbone $B$ in any round $r$. If there exists a node $u \in B$ such that $u.\mathit{exit} = 1$ at round $r$, then by round $r + 2\lfloor \log |B| \rfloor$, all nodes in $B$ would've exited from the dual-state (at least once).
\end{lemma}
\begin{proof}
Let $d_B(p, q)$ denote the distance between any two nodes $\{p, q\} \subseteq B$ over the backbone $B$ in round $r$. We recursively define it in the following way. (Note that this distance remains same, throughout the execution, as it is defined for backbone $B$ at round $r$.)
\begin{equation*}
  d_B(p, q) = \begin{cases}
              0 &\text{if $p = q$},\\
              (d_B(p.S_l(1), q)) + 1 &\text{if $(q.L = \emptyset \lor q.S_l(1) \notin B)$},\\
              d_B(p, q.S_l(1)) + 1 &\text{if $(p.L = \emptyset \lor p.S_l(1) \notin B)$},\\
              \lvert d_B(p, r) - (d_B(q, r) \rvert &\text{otherwise, where $r \in B \land (r.L = \emptyset \lor r.S_l(1) \notin B)$}.
              \end{cases}
\end{equation*}

Let $D_{B}(v) = \min_{w \in B \land w.\mathit{exit} = 1}d_B(w, v)$ denote the minimum distance, over all nodes in $B$ that have their $\mathtt{exit}$ variable set to 1, to the node $v \in B$.

If at round $r' \in \{ r, \dots, r+2\lfloor \log |B| \rfloor - 1\}$, if either $(v.L \neq \emptyset \land v.S_l(1) \in B \land v.S_l(1).\mathit{exit} = 1)$ or $(v.R \neq \emptyset \land v.S_r(1) \in B \land v.S_r(1).\mathit{exit} = 1)$, i.e., $D_{B}(v) < 2$, then node $v$ exits the dual-state in round $r'+1$. This is due to Remark \ref{rem:flushflyover}, every node, in $\mathbf{FlushFlyover}$ function, sends $\langle \text{RejFlyover} \rangle$ to all ids in its flyover-related address variables. Moreover, if $D_{B}(v) \geq 2$ in round $r$, then we show that it decreases exponentially over time (due to flyover construction, i.e., pointer doubling \cite{JaJa92}).

Consider a backbone $B_k = \{ v_{1, k}, v_{2, k}, \dots \}$ in round $r+k$, where $(v_{1, k}.L = \emptyset \lor v_{1, k}.S_l(1) \notin B_k)$ and $\forall i \in [|B_k|-1], v_{i, k}.S_r(1) = v_{i+1, k}$, and $B_0 = B$. In round $r+k$, consider a node $w$ that satisfies $d_{B}(w, v) = D_{B}(v) \geq 2$, and the backbone $B_k = \{ v_{1, k}, \dots, v_{|B_k|, k} \}$, part of the original backbone $B$, such that either $(v_{1, k} = w \land v_{|B_k|, k} = v)$ or $(v_{1, k} = v \land v_{|B_k|, k} = w)$ is true. For any $k \geq 1$, by Lemma \ref{lemma:fly-const}, we know that $\forall j \in [k], v_{1, k}.S_r(j) = v_{1+2^j, k}$ and $\forall j \in [k], v_{|B_k|, k}.S_l(j) = v_{|B_k|-2^j, k}$. Thus, by round $r+2\lfloor \log |B| \rfloor - 1$, either node $v$ exited the dual-state, or $D_B(v) < 2$.
\end{proof}

\subsection{Concurrent Advices}

\begin{claim} \label{claim:sync-advice}
Consider any set of nodes $Q \subset V$ that process advice messages in round $r$. For any node $v \in Q$, if the advice is effective, then for any node $u \notin Q$, then $u \notin v.S$ after $v$ enters dual-state.
\end{claim}
\begin{proof}
Once any node receives an advice, if the node id present in it, is also part of its memory (i.e., in \texttt{snap} variable), then the node proceeds to check that the distance (i.e., in \texttt{dist} variable) matches with its ``parent'' node id (as in the advice). See Algorithm \ref{alg:local-cert} for more details. This back-and-forth verification of advice between a node and its parent in a synchronous system, before executing the tree-to-path transformation (i.e., Algorithm \ref{alg:distrib-tree-to-path}), ensures that they must have processed the advice in the same round. If any node $u \notin Q$ is not part of this back-and-forth communication, then it is not part of $\mathbf{SetupLocalTransform}$ function in Algorithm \ref{alg:local-cert}, which in turn, implies that it is also not part of $\mathbf{LocalTransform}$ and $\mathbf{JoinPath}$ functions in Algorithm \ref{alg:distrib-tree-to-path}. Thus, $\forall v \in Q, u \notin v.S$.
\end{proof}

\begin{lemma} \label{lemma:sync-advice}
Let $D$ be the set of nodes in dual-state in any round $r > 5$. Let $D'$ be the set of nodes that transition to dual-state in round $r+1$. For any pair of nodes $u \in D$ and $v \in D'$, $u \notin v.S$ in round $r+1$.
\end{lemma}
\begin{proof}
Consider the following (exhaustive) cases for any node $u \in D$ in round $r$.
\begin{enumerate}
    \item Node $u$ was always in dual-state from round 1.
    \item Node $u$ processed an advice message in a round $r' \leq r-3$, as the distributed tree-to-path algorithms (cf. Section \ref{subsec:distrib-TtP}) takes 3 rounds, until node $u$ enters dual-state. 
\end{enumerate}
Each node in $D'$, as it enters dual-state in round $r+1$, processed an advice message at round $r-2$. Let $Q$ be the set of nodes that processed advice messages in round $r-2$. Clearly, $u \notin Q$ (as $u.t \neq 0$ or $u.S \neq \emptyset$ in that round). Thus, by Claim \ref{claim:sync-advice}, any pair of nodes $u \in D$ and $v \in D'$, $u \notin v.S$.
\end{proof}

\subsection{Fast Detection of Bad Advice} \label{anal:det-bad-advice}

\begin{lemma} \label{lemma:winged-backbone}
Let $B = \{ v_1, \dots, v_{|B|} \}$ be a maximal backbone in round $r > 5$, where either $(v_1.S_l(1) \notin B)$ or $(v_{|B|}.S_r(1) \notin B)$ is satisfied. Let $L$ be the set of lost nodes in round $r$. There is at least node in $B$ that sets its $\mathtt{exit}$ variable to 1 by round $r+2$. Moreover, every node in $L$ sets its $\mathtt{exit}$ variable to 1 by round $r+2$.
\end{lemma}
\begin{proof}
We provide the proof for the backbone $B$, but the same argument applies to the lost nodes in $L$. Wlog, let node $v_1.L \neq \emptyset$; similar argument also applies to $v_{|B|}$ if $v_{|B|}.R \neq \emptyset$.

Node $v_1$ sends $\langle \text{TestLine-L} \rangle$ to $v_1.S_l(1)$ in round $r$. But $v_1.S_l(1)$, in round $r+1$, by Lemma \ref{lemma:sync-advice}, as it holds that $(v_1.S_l(1)).S_r(1) \neq v_1$, sends $\langle \text{RejFlyover} \rangle$ to $v_1$ (regardless whether $v_1.S_l(1)$ in a dual-state or not). Node $v_1$ its $\mathtt{exit}$ variable set to 1 by round $r+2$.
\end{proof}


\begin{lemma} \label{lemma:maximal-ouroboros}
Each node in a maximal ouroboros $O$ in round $r > 5$, exits dual-state in $O(\log |O|)$ rounds.
\end{lemma}
\begin{proof}
By Lemma \ref{lemma:sync-advice}, the size of any maximal ouroboros (or maximal backbone) doesn't increase over time (i.e., for any round $r' > r$, if some node $u \notin O$, then $u \notin O$ in round $r'$). Let us consider the two cases of maximal ouroboros (as in Definition \ref{def:ouroboros}) in round $r$.
\begin{enumerate}
    \item If $O$ is a perfect ring, then at least one pair of consecutive nodes in $O$ set their $\mathtt{exit}$ variables to 1 by round $r+1$, as by design, every node $u$ in its dual-state, $u$ sends $u.\mathit{vID} + 1$ to $u.S_{r}(1)$ and $u.\mathit{vID} - 1$ to $u.S_{l}(1)$ in every round. Let $L \subseteq O$ be the set of nodes that are lost in round $r+2$; they exit from dual-state by round $r+4$ by Lemma \ref{lemma:winged-backbone}. The set of nodes $O' \subseteq (O \setminus L)$, that are in dual-state but not lost in round $r+2$, forms a partition of maximal winged backbones as every node $u \in O$ in round $r$, has $u.L \neq \emptyset \land u.R \neq \emptyset$. By Lemma \ref{lemma:winged-backbone}, combined with Lemma \ref{lemma:inf-prop}, by round $r+O(\log |O|)$, all nodes in $O'$ exit the dual-state. 
    \item If $O$ is a stylish ring, then there is at least one node $u \in O$ such that,
    \begin{equation*}
      \{u.S_l(1), u.S_r(1)\} \subseteq O \land ((v = u.S_l(1) \land v.S_r(1) \neq u) \lor (v = u.S_r(1) \land v.S_l(1) \neq u)).
    \end{equation*}
    Node $u$ sends $\langle \text{TestLine-R} \rangle
    $ and $\langle \text{TestLine-L} \rangle
    $ messages in round $r$. In round $r+1$, (at least) one of the recipients, sets the $\mathtt{exit}$ variable to 1, and sends back $\langle \text{RejFlyover} \rangle$ to node $u$. By round $r+2$, node $u$ exits the dual-state. Let $L \subseteq O$ be the set of nodes that are lost in round $r+2$; they exit from dual-state by round $r+4$ by Lemma \ref{lemma:winged-backbone}. The set of nodes $O' \subseteq O \setminus (M \cup L)$ that are in dual-state but not lost in round $r+2$, form a partition of maximal winged backbones, as every node $v \in O'$ in round $r$, has $v.L \neq \emptyset \land v.R \neq \emptyset$. By Lemma \ref{lemma:winged-backbone} and Lemma \ref{lemma:inf-prop}, by round $r+O(\log |O|)$, all nodes in $O'$ exit the dual-state.
\end{enumerate}
\end{proof}

\begin{lemma} \label{lemma:corr-config}
Consider any maximal backbone $B \subseteq V$ in any round $r > 5$, where $B$ is not winged. If all nodes in $B$ have their $\mathtt{exit}$ variable set to $0$ in all rounds $r, r+1, \dots, T$, then the following statements hold.
\begin{enumerate}
    \item If $T \geq r+ 2\lceil \log |B| \rceil$, then in any round $r' \in \{ r+ 2\lceil \log |B| \rceil, \dots, T\}$, $B$ is a flyover, and for all $u \in B, u.\mathit{flyID} = v.\mathit{id}$ where $v.L = \emptyset$.
    \item If $T \geq r+ 3\lceil \log |B| \rceil$, and for all $u \in B, u.S \subseteq B$, then in any round $r' \in \{ r+ 3\lceil \log |B| \rceil, \dots, T\}$, backbone $B$ is correctly configured.
\end{enumerate}
\end{lemma}
\begin{proof} First, by Lemma \ref{lemma:sync-advice}, the size of any maximal backbone doesn't increase over time (i.e., for any round $r' \geq r$, if some node $v \notin B$ in round $r$, then $v \notin B$ in round $r'$). Since $B$ is \emph{not} a winged maximal backbone, i.e., there exist $v_1$ and $v_{|B|}$ nodes in $B$ such that $(v_1.L = \emptyset)$ and $(v_{|B|}.R = \emptyset)$.

Due to our algorithm design, we recall the possible changes of flyover-related variables for any node $u$ whose $u.S \neq \emptyset$ and $u.\mathit{exit} = 0$.
\begin{enumerate}
    \item The variables $\texttt{vID}, \texttt{c-par}, \texttt{c-dist}$ are initially updated after receiving an advice message from the supervisor (see $\mathbf{GetAdvice}$ function). But they can only be updated again once the node executes $\mathbf{FlushFlyover}$ function (i.e., after $u.\mathit{exit} = 1$).
    \item If $u.L = \emptyset \land u.R \neq \emptyset$, then $u.\mathit{flyID} = u.\mathit{id}$; otherwise, the variable $\texttt{flyID}$ can be assigned some other id in $\mathbf{R\_TestFlyoverMetadata}$ function. Importantly, once $u.\mathit{flyID} \neq u.\mathit{id}$, it can only be updated again once the node executes $\mathbf{FlushFlyover}$ function (i.e., after $u.\mathit{exit} = 1$).
    \item The variables $\texttt{L}, \texttt{R}$ and $\texttt{c-ids}$ can be updated over time in that the existing ids in the variables cannot be removed but new ids can be added. Thus, this means that $u.S_l(1)$ and $u.S_r(1)$ can be updated again once the node executes $\mathbf{FlushFlyover}$ function (i.e., after $u.\mathit{exit} = 1$).
\end{enumerate}
As all nodes have their $\texttt{exit}$ variable set to 0, let us go over the required variable changes of a correctly configured flyover (as in Definitions \ref{def:flyover} and \ref{def:corr-config-backbone}). Wlog, let $B = \{ v_1, v_2, \dots, v_{|B|} \}$. If the following events do not happen, then $\exists v \in B, v.\mathtt{exit} = 1$ by round $r+3\lfloor \log |B| \rfloor$.
\begin{enumerate}
    \item (Virtual IDs) As $v_1.L = \emptyset$, $v_1.\mathit{vID}$ must be equal to 1. For all $i \in [|B|]$, in round $r+1$, due to the $\mathbf{TestFlyoverMetadata}$ function, $v_i$ must be equal to $i$.
    \item (Hypercube Edges) By Lemma \ref{lemma:fly-const} (and Corollary \ref{corr:flyover-const}), by round $r + \lfloor \log |B| \rfloor$, backbone $B$ is a flyover (as per Definition \ref{def:flyover}) due to the $\mathbf{TestFlyoverConstruction}$ function.
    \item (Flyover ID) By round $r + 2 \lfloor \log |B| \rfloor$, for all $i \in [|B|]$, $v_i.\mathit{flyID}$ must be equal to $v_1$ due to the $\mathbf{TestFlyoverMetadata}$ function (as $B$ must be a flyover by round $r + \lfloor \log |B| \rfloor$).
    \item (Connectivity Certificate) At round $r + 2 \lfloor \log |B| \rfloor$, the $\texttt{prop-flyID}$ variable is set to 1 for any node $u \in B$, as the clause $(u.\mathit{vID} = 1) \lor (u.\mathit{vID} > 1 \land u.\mathit{flyID} = v_1)$ is satisfied. For all $i \in [|B|] \setminus \{ 1 \}$, node $v_i$ routes the $v_i.\textit{c-dist} - 1$ value (over the flyover) to the node with the virtual ID $v_i.\textit{c-par}$ using the $\mathbf{NextStop}$ function. The message must be routed to some node \emph{within} the flyover, and that node \emph{verifies} the message with its $\texttt{c-dist}$ variable. Combined with the tests in the $\mathbf{BasicChecks}$ function, this results in adding the appropriate neighboring ids in the sorted-path graph over backbone $B$ (as per Definition \ref{def:corr-config-backbone}) by the next $\lceil \log |B| \rceil$ rounds.
\end{enumerate}
\end{proof}

\begin{corollary} \label{cor:conn-A0-1}
Let all nodes form a single backbone in round $r$. If all nodes have their $\mathtt{exit}$ variable set to 0 in all rounds $r, r+1, \dots, r+3\lceil \log n \rceil$, then the (sub)network over $\forall u \in V, \mathrm{M}_u(\mathcal{A}_0)$, includes the sorted-path topology in round $r+3\lceil \log n \rceil$.
\end{corollary}

\begin{lemma} \label{lemma:mult-backbone-exit}
Consider a maximal backbone $B = \{ v_1, \dots, v_{|B|} \}$ in any round $r > 5$, where $|B| < n$ and $(v_1.L = \emptyset \lor v_1.S_l(1) \notin B)$ and $\forall i \in [|B|-1], v_i.S_r(1) = v_{i+1}$. By round $r + O(\log n)$, at least one node in $B$ will set its $\mathtt{exit}$ variable to 1.
\end{lemma}
\begin{proof}
First, by Lemma \ref{lemma:sync-advice}, the size of any maximal backbone doesn't increase over time (i.e., for any round $r' > r$, if some node $u \notin B$, then $u \notin B$ in round $r'$).

If $B$ is a maximal winged backbone in round $r$, then by Lemma \ref{lemma:winged-backbone}, at least one node in $B$ sets its $\mathtt{exit}$ variable set to 1 by round $r+2$. Thus, for the rest of the proof, we consider the scenario where $B$ is \emph{not} a winged maximal backbone, i.e., $(v_1.L = \emptyset)$ and $(v_{|B|}.R = \emptyset)$.

Furthermore, for the rest of the proof, we consider the scenario where, no node in $B$ sets its $\texttt{exit}$ variable to 1 by round $r+ 2 \lceil \log n \rceil$. (Otherwise, that would only support the lemma statement.) By Lemma \ref{lemma:corr-config}, in any round $r' \geq r + 2 \lceil \log n \rceil$, $B$ is a flyover, and $\forall u \in B, u.\mathit{flyID} = v_1$.

Recall that the communication graph is weakly connected in any round (due to Lemma \ref{lemma:conn-global}). We split the proof into multiple cases (and subcases), analyzing how the backbone $B$ gets affected from the existence of a directed edge $(p, q)$ in round $r'$, where $(p \in B \land q \notin B)$ or $(p \notin B \land q \in B)$.

\emph{Case 1.} Node id $p$ is stored in $q$'s flyover-related variables, i.e., $p \in q.S \cup \{q.\mathit{flyID}\} \cup q.\textit{c-ids}$. In all the following subcases, some node in $B$ will set its $\texttt{exit}$ variable to 1 by round $r' + O(\log n)$.
\begin{enumerate}
    \item If $q.S = \emptyset$, then $q$ sends $\langle \text{RejFlyover} \rangle$ to all ids in $q.S \cup \{q.\mathit{flyID}\} \cup q.\textit{c-ids}$ (see, Remark \ref{rem:flushflyover}.).
    \item If $q.S \neq \emptyset$, and any one of the following scenarios holds: $q$ is a lost node, or $q$ is in a maximal ouroboros or maximal winged backbone. By Lemma \ref{lemma:winged-backbone} and Lemma \ref{lemma:maximal-ouroboros}, $q$ exits the dual-state by round $r'+O(\log n)$, and sends $\langle \text{RejFlyover} \rangle$ to all ids in $q.S \cup \{q.\mathit{flyID}\} \cup q.\textit{c-ids}$. 
    \item If $q.S \neq \emptyset$, and $q$ is in a maximal backbone $B_q$, where $B_q$ is \emph{not} winged. Let $q'$ be the node in $B_q$ where $q'.L = \emptyset$. Clearly, $B_q \cap B = \emptyset$ (as both are maximal), so $p.\mathit{flyID} \notin B_q$ and $q' \notin B$. If some node in $B_q$ exits dual-state by round $r' + 2 \lceil \log n \rceil$, then by Lemma \ref{lemma:inf-prop}, $q$ exits the dual-state by round $r' + 4 \lceil \log n \rceil$ and sends $\langle \text{RejFlyover} \rangle$ to all ids in $q.S \cup \{q.\mathit{flyID}\} \cup q.\textit{c-ids}$. On the other hand, if no node in $B_q$ sets its \texttt{exit} variable to 1 until round $r' + 2 \lceil \log n \rceil$, then by Lemma \ref{lemma:corr-config}, the \texttt{prop-flyID} variable is set to 1 in $q$, after which $q' = q.\mathit{flyID}$ is sent to $p$, due to $\mathbf{TestFlyoverMetadata}$ function, causing $p$ to set its \texttt{exit} variable to 1, as $p.\mathit{flyID} \neq q.\mathit{flyID}$.
\end{enumerate}

\emph{Case 2.} Node id $q$ is stored in (memory or channel of) $p$. As all nodes in $B$ have their \texttt{prop-flyID} variable set to 1, node $q$ receives $\langle \text{TestFlyID}, p.\mathit{flyID} \rangle$ in round $r'+1$. Based on $q$'s memory in round $r'+1$, we consider the following subcases.
\begin{enumerate}
    \item If $(q.S \neq \emptyset \land q.\mathit{flyID} \neq p.\mathit{flyID}) \lor (q.S = \emptyset)$, then due to $\mathbf{R\_TestFlyoverMetadata}$ function, $q$ sends $\langle \mathrm{RejFlyover} \rangle$ to $p.\mathit{flyID}$ in round $r'+1$, causing node $p.\mathit{flyID}$ to set its $\mathtt{exit}$ to 1.
    \item If $(q.S \neq \emptyset \land q.\mathit{flyID} = p.\mathit{flyID})$, then we resort to Case 1 (where $p$ is $p.\mathit{flyID}$).
\end{enumerate}

\emph{Case 3.} Node id $p$ is in $M_q({\mathcal{A}_0})$. If node $q$ uses the DR primitive (see, Definition \ref{def:bidirected}), i.e., sends node id $p$ to another node, within algorithm $\mathcal{A}_0$, by round $r' + O(\log n)$, then by Claim \ref{claim:dr-prim}, there would exist an edge $(p', q')$, where $p' \in B$ and $q' \in B$, in the next two rounds. In that scenario, we can resort to Case 2. Thus, we analyze the following subcases, where $p$ is always in $M_q({\mathcal{A}_0})$.
\begin{enumerate}
    \item If $q.S = \emptyset$, then in $O(1)$ rounds, $q.t = 0$ as it is decremented by 1 in every round (and it can be at most 5). (See, $\mathbf{BasicChecks2}$ function in Section \ref{subsec:distrib-TtP}.) When $q.t \leq 1$, $q.\mathit{vID}$ is set to 0, causing node $q$ to send $\langle \text{TestFlyID}, \bot \rangle$ to all ids in $\mathbf{TestFlyoverMetadata}$ function (before $q$ processes a new advice message; see Remark \ref{rem:attentive-bot}). Then, $p$ would set its \texttt{exit} variable to 1.
    \item If $q.\mathit{exit}$ is set to 1 in the next $O(\log n)$ rounds, then $q.S$ and $q.\mathit{vID}$ are set to $\emptyset$ and $0$ (resp.). In that round, node $q$ sends $\langle \text{TestFlyID}, \bot \rangle$ to all ids in $\mathbf{TestFlyoverMetadata}$ function (before $q$ processes a new advice message; see Remark \ref{rem:attentive-bot}), causing $p$ to set its \texttt{exit} variable to 1.
    \item If $q.S \neq \emptyset \land q.\mathit{exit} = 0$ for the next $O(\log n)$ rounds, then by Lemma \ref{lemma:winged-backbone} and \ref{lemma:maximal-ouroboros}, $q$ must belong to a maximal backbone $B_q$, where $B_q$ is \emph{not} winged. Let $q'$ be the node in $B_q$ where $q'.L = \emptyset$. Clearly, $B_q \cap B = \emptyset$ (as both are maximal), so $p.\mathit{flyID} \notin B_q$ and $q' \notin B$. By Lemma \ref{lemma:corr-config}, the \texttt{prop-flyID} variable is set to 1 in $q$, after which $q' = q.\mathit{flyID} \notin B$ is sent to $p$, due to $\mathbf{TestFlyoverMetadata}$ function, causing $p$ to set its \texttt{exit} variable to 1, as $p.\mathit{flyID} \neq q.\mathit{flyID}$.
\end{enumerate}

\emph{Case 4.} Node id $p$ is present in a flyover-related message of $q$. In each of the following subcases, we can resort to either Case 1 or Case 3.
\begin{enumerate}
    \item In $\mathbf{R\_TestFlyoverConstruction}$ function, node $q$ processes two types of messages. For any message $\langle \mathit{msg}, \mathit{sen} \rangle$, $\mathit{sen}$ is stored in, either $q.S$ or $\mathrm{M}_q(\mathcal{A}_0)$. For any message $\langle \mathit{msg}, w, i, \mathit{sen} \rangle$, $w$ and $\mathit{sen}$ are stored in either $q.S$ or $\mathrm{M}_q(\mathcal{A}_0)$.
    
    In $\mathbf{R\_TestFlyoverMetadata}$ function, for any message $\langle \text{TestFlyID}, \mathit{flyID} \rangle$, $\mathit{flyID}$ is stored in either $q.\mathit{flyID}$ or $\mathrm{M}_q(\mathcal{A}_0)$.
    \item In $\mathbf{R\_TestConnCertificate}$ function, for any message $\langle \text{IntroCert}, w \rangle$, $u$ stores $q$ in $u.\textit{c-ids}$. For any message $\langle \text{TestCert}, w, \mathit{vID}, \mathit{dist} \rangle$, there are three scenarios: (1) $w$ is sent to some node $q' \in q.S$, or (2) $q$ stores $w$ in $q.\textit{c-ids}$, or (3) $q$ stores $w$ in $\mathrm{M}_q(\mathcal{A}_0)$. In the first scenario, as the \texttt{prop-flyID} variable is set to 1, node $q$ would have sent $q.\mathit{flyID}$ to $p$. If $q.\mathit{flyID} \neq p.\mathit{flyID}$, then $p$ would set its \texttt{exit} variable to 1. Otherwise, if $q.\mathit{flyID} \neq p.\mathit{flyID}$, we resort to Case 1, where $p' = p.\mathit{flyID} \in B$ is stored in node $q$'s flyover-related variable.
\end{enumerate}

\emph{Case 5.} Finally, node id $p$ is present in a advice-related message of $q$ (see, algorithms in Section \ref{subsec:distrib-TtP}). In each of those algorithms, whenever a node $q$ processes any message, it simply flushes the node ids present in the message to $M_q(\mathcal{A}_0)$ at the end. Thus, we resort to Case 3.
\end{proof}

\subsection{Putting Everything Together}

\begin{claim} \label{claim:exit-dual-state}
If a node is in a precarious dual-state in round $r$, it exits dual-state by round $r+ O(\log n)$.
\end{claim}
\begin{proof}
For the first three cases of the precarious dual-state  property (see Definition \ref{def:prec-dual-state}), if all nodes do not form a single backbone, then by Lemma \ref{lemma:winged-backbone}, Lemma \ref{lemma:maximal-ouroboros} and Lemma \ref{lemma:mult-backbone-exit}, combined with fast information dissemination (i.e., Lemma \ref{lemma:inf-prop}), every node exits the dual-state in $O(\log n)$ rounds. However, if the network forms a single backbone, and if some node is sets its $\texttt{exit}$ variable to 1 by round $r+3\lceil \log n \rceil$, then again, combined with Lemma \ref{lemma:inf-prop}, every node exits the dual-state by round $r+ O(\log n)$.
\end{proof}

\begin{theorem} \label{theorem:s-robustness}
The network is s-robust.
\end{theorem}
\begin{proof} By design (see, $\mathbf{GetAdvice}$ function in Algorithm \ref{alg:local-cert}, and Definition \ref{def:well-formed-adv} for ``well-formed'' advice), for any node $u$, (1) there can be only \emph{one} node id present in any advice message, and (2) that id must \emph{already be present} in the $\mathtt{snap}$ variable when the advice is received at node $u$. Thus, the supervisor cannot add arbitrary connections to any node, ensuring the Sybil resistance property.

If any node is in a precarious dual-state (see Definition \ref{def:prec-dual-state}), then by Claim \ref{claim:exit-dual-state}, it exits the dual-state in $O(\log n)$ rounds. On the other hand, if no node is in a precarious dual-state, and the network forms a single backbone, and no node sets its \texttt{exit} variable to 1 for $\Theta(\log n)$ rounds, then by Lemma \ref{lemma:corr-config}, the backbone is a correctly configured flyover (i.e., the supervisor provided correct advice messages), i.e., by Corollary \ref{cor:conn-A0-1}, the network is in an almost-legal configuration. Moreover, this observation implies that even if the supervisor is malicious (e.g., providing ``bad'' advice messages), the network quickly rejects the ``supervised-induced'' subnetwork (from any arbitrary configuration). Thus, the network also achieves the robustness property, which would allow an honest supervisor to quickly (i.e., in $O(\log n)$ rounds) contact all nodes, provide correct advice messages to repair the network.
\end{proof}

\begin{claim} \label{claim:process-msg-ids}
Consider a node id $v$ in any flyover-related or advice-related message of a node $u$ in round $r$. By round $r+1$, one of the following events happen, where $w = u.\mathit{flyID}$ in round $r$.
\begin{enumerate}
    \item $v \in u.S \cup \{ u.\mathit{flyID} \} \cup u.\textit{c-ids}$.
    \item $v \in \mathrm{M}_u(\mathcal{A}_0)$.
    \item $v.\mathit{flyID} = w$.
    \item $w \in \mathrm{M}_v(\mathcal{A}_0)$.
\end{enumerate}
\end{claim}
\begin{proof}
Similar to the proof of Lemma \ref{lemma:conn-global}, we provide a case-by-case analysis depending on the message that the node id $v$ is part of.
\begin{enumerate}
    \item If $u.S = \emptyset$, in functions, $\mathbf{R\_TestFlyoverConstruction}, \mathbf{R\_TestFlyoverMetadata}$ and $\mathbf{R\_TestConnCertificate}$, node id $v$ straightaway gets added to $\mathrm{M}_u(\mathcal{A}_0)$ in round $r$.
    \item If $v$ belongs to any function in the distributed tree-to-path algorithms (cf. Section \ref{subsec:distrib-TtP}), again, node id $v$ straightaway gets added to $\mathrm{M}_u(\mathcal{A}_0)$ in round $r$.
    \item In $\mathbf{R\_TestFlyoverConstruction}$ function, node $u$ processes two types of messages. For any message $\langle \mathit{msg}, \mathit{sen} \rangle$, $\mathit{sen}$ is stored in, either $u.S$ or $\mathrm{M}_u(\mathcal{A}_0)$. For any message $\langle \mathit{msg}, w, i, \mathit{sen} \rangle$, $w$ and $\mathit{sen}$ are stored in either $u.S$ or $\mathrm{M}_u(\mathcal{A}_0)$.
    
    In $\mathbf{R\_TestFlyoverMetadata}$ function, for any message $\langle \text{TestFlyID}, \mathit{flyID} \rangle$, there are two cases: $\mathit{flyID}$ is stored in either $u.\mathit{flyID}$ or $\mathrm{M}_u(\mathcal{A}_0)$.
    \item In $\mathbf{R\_TestConnCertificate}$ function, for any message $\langle \text{TestCert}, v, \mathit{vID}, \mathit{dist} \rangle$, there are three cases: (1) $v$ is sent to node $u' \in u.S$, or (2) $u$ stores $v$ in $u.\textit{c-ids}$, or (3) $u$ stores $v$ in $\mathrm{M}_u(\mathcal{A}_0)$; in the first case, before $v$ is sent to an id $u' \in u.S$, $u$ must've sent $\langle \mathrm{TestFlyID}, u.\mathit{flyID} \rangle$ to $v$ in round $r$ (as the \texttt{prop-flyID} variable is set to 1). In round $r+1$, node $v$ stores $u.\mathit{flyID}$ in either $v.\mathit{flyID}$, or $\mathrm{M}_v(\mathcal{A}_0)$. And, for any message $\langle \text{IntroCert}, v \rangle$, $u$ stores $v$ in $u.\textit{c-ids}$.
\end{enumerate}
\end{proof}

\begin{theorem} \label{theorem:final}
If supervisor is honest, then the network reaches an almost-legal configuration in $O(\log n)$ rounds; otherwise, it reaches an almost-legal configuration in $2\mathrm{R}(\mathcal{A}_0) + O(\log n)$ rounds.
\end{theorem}
\begin{proof} We denote the subnetwork formed by edges in $\forall u \in V, \mathrm{M}_u(\mathcal{A}_0)$ as $\mathrm{N}(\mathcal{A}_0)$. Recall that the communication graph is always weakly connected (in any round) due to Lemma \ref{lemma:conn-global}.

First, we introduce some notation for showing that $N(\mathcal{A}_0)$ would be (weakly) connected. Let the (directed) edge $(u, v)$ in the communication graph in any round $r \geq 1$; we define the color of an edge $(u, v)$ in the following way.
\begin{enumerate}
    \item If $(v \notin \mathrm{M}_u(\mathcal{A}_0)) \land (v \notin u.S \cup \{ u.\mathit{flyID} \} \cup u.\textit{c-ids})$, but node id $v$ is present in a flyover-related or advice-related message in $u$, then the color of $(u, v)$ is red.
    \item If $(v \notin \mathrm{M}_u(\mathcal{A}_0)) \land (v \in u.S \cup \{ u.\mathit{flyID} \} \cup u.\textit{c-ids})$, then the color of $(u, v)$ is blue.
    \item If $v \in \mathrm{M}_u(\mathcal{A}_0)$, then the color of the edge $(u, v)$ is green. Similarly, we say that a path between any pair of nodes (ignoring edge directions) is a \emph{green path} if all edges in it are green.
\end{enumerate}

We proceed to show that $\mathrm{N}(\mathcal{A}_0)$ forms a single (weakly) connected component, i.e., there exists a green path between any pair of nodes, in $O(\log n)$ rounds, regardless of the initial configuration and any node's interactions with the supervisor. If $\mathrm{N}(\mathcal{A}_0)$ is weakly connected, then it always remains a single (weakly) connected component in any subsequent round, as Algorithm $\mathcal{A}_0$ is a self-stabilizing overlay algorithm (i.e., it preserves connectivity from any initial configuration), and importantly, nodes do not delete connections from the memory (incl. channel) dedicated to Algorithm $\mathcal{A}_0$ while interacting with the supervisor (cf. Theorem \ref{theorem:s-robustness}).

Towards that end, consider any pair of nodes $u$ and $v$ that do \emph{not} have any green path between them in any round $r \geq 1$. As the communication graph is weakly connected, there exists some (undirected, i.e., ignoring edge directions) path $P = (u, v_1, \dots, v_t)$, where $v_t = v$ for $t \geq 1$, in that round. Consider any edge $e = (p, q)$ belonging to $P$; the following invariants hold in round $r+1$.
\begin{enumerate}
    \item If $e$ is red, then $e$ turns either blue or green, due to Claim \ref{claim:process-msg-ids}.
    \item If $e$ is blue, then $e$ either remains blue or turns green, due to Remark \ref{rem:flushflyover}. 
    \item If $e$ is green, then due to the guarantee that Algorithm $\mathcal{A}_0$ preserves connectivity, there exists some green path $P_e = (p, q_1, \dots, q_{t'})$, where $q_{t'} = q$ for $t' \geq 1$.
\end{enumerate}

We proceed to show that any blue edge turns green in $O(\log n)$ rounds. Consider any blue edge $e = (p, q)$, where the node id $q$ belongs to a flyover-related address variable of node $p$. In the following (exhaustive) cases, nodes $p$ and $q$ end up with a green path between them.
\begin{enumerate}
    \item If $p.S = \emptyset$, then all flyover-related address variables of $p$ are straightaway flushed to $\mathrm{M}_p(\mathcal{A}_0)$.
    \item If $p$ is in a precarious dual-state, then by Claim \ref{claim:exit-dual-state}, $p$ exits dual-state in $O(\log n)$ rounds.
    \item If the network forms a single backbone, and if all nodes have their $\mathtt{exit}$ variables set to 0 in all rounds $r, r+1, \dots, r+3\lceil \log n \rceil$. By Corollary \ref{cor:conn-A0-1}, $\mathrm{N}(\mathcal{A}_0)$ becomes (weakly) connected.
\end{enumerate}

If $\mathrm{N}(\mathcal{A}_0)$ is (weakly) connected (regardless of any supervisor interaction), the network reaches an almost-legal configuration in $2\mathrm{R}(\mathcal{A}_0)$ rounds. This is because each node $u$ executes Algorithm $\mathcal{A}_0$ (in the background), and it is guaranteed that Algorithm $\mathcal{A}_0$, when executed independently, converges to the target topology $G^*$ in $\mathrm{R}(\mathcal{A}_0)$ rounds. There is an extra factor of 2, in our upper bound, because ``Delegation'' (of a node id by any node via a message), that takes a single round, is implemented in two rounds as ``Delegate-after-Reversal'' (see Definition \ref{def:bidirected}). Thus, if supervisor is not honest, the network reaches an almost-legal configuration in $2\mathrm{R}(\mathcal{A}_0) + O(\log n)$ rounds.

If the supervisor is honest, we can leverage the (correct) advice for showing fast stabilization. First, we show that within $O(\log n)$ rounds, there exists at least one attentive node, in which case, the supervisor is contacted for help. Once the supervisor is contacted by a node, the supervisor waits for $O(\log n)$ rounds until all nodes become attentive. Then, the supervisor sends the correct advice messages, leading to all nodes forming a correctly-configured flyover in $O(\log n)$ rounds. Finally, each node finds its predecessor and successor (in the sorted-path topology $G^*$) in the \texttt{c-id} variable, due to Corollary \ref{cor:conn-A0-1}, as no node sets its $\mathtt{exit}$ variable to 1.

Towards that end, we prove the following arguments.
\begin{enumerate}
    \item If $\mathrm{N}(\mathcal{A}_0)$ does not include $G^*$ in $O(\log n)$ rounds, then there exists at least one attentive node in $O(\log n)$ rounds. (If $G^*$ is already included, then we are done.)
    \item If there is one attentive node, then all nodes become attentive in $O(\log n)$ rounds.
\end{enumerate}

Consider the following (exhaustive) cases in any round $r > 5$.
\begin{enumerate}
    \item By Claim \ref{claim:exit-dual-state}, any node in a precarious dual-state exits dual-state in $O(\log n)$ rounds.
    \item If the network forms a single backbone, and if all nodes have their $\mathtt{exit}$ variables set to 0 in all rounds $r, r+1, \dots, r+3\lceil \log n \rceil$. By Corollary \ref{cor:conn-A0-1}, $\mathrm{N}(\mathcal{A}_0)$ becomes (weakly) connected.
\end{enumerate}
For any node $u$, $u.t$ decreases by one in every round (and $u.t$ can be at most $5$). (See $\mathbf{BasicChecks2}$ function in Section \ref{subsec:distrib-TtP}). By Claim \ref{claim:process-msg-ids}, if there is any node id $v$ in any flyover-related or advice-related message of any node $u$ (where $u.S = \emptyset$), then it is flushed to $\mathrm{M}_u(\mathcal{A}_0)$. Once node $u$ exits dual-state, it becomes attentive (i.e., $u.S = \empty \land u.\mathit{exit} = 0 \land u.t = 0$) in $O(1)$ rounds. Thus, either all nodes become attentive in $O(\log n)$ rounds, or $N(\mathcal{A}_0)$ includes the target topology $G^*$.
\end{proof}

\subsection{Lower Bound for Convergence Time}

\begin{theorem}
For an s-robust network, there exists some initial configuration such that the convergence time for any pair of supervisor and overlay algorithms is $\Omega(\log n)$ rounds.
\end{theorem}

\begin{proof}

Consider an initial configuration where two nodes $u$ and $v$ are at a distance $D = \Omega(n)$ in the communication graph. However, these nodes are meant to be adjacent in target topology $G^*$. If there is a pair of supervisor and self-stabilizing overlay algorithm for which this network reaches an almost-legal configuration in $o(\log D)$ rounds, then we can show a contradiction that the nodes $u$ and $v$ were at a distance $d < D$ in the initial state.

Recall that the nodes of an s-robust network do not arbitrarily create connections of any node id present in the advice message. In particular, a node can only add (directed) edges between any two ids present in its memory or channel. Thus, even if the supervisor is aware of memory states (incl. channels) of all nodes, and is allowed to interact with any subsets of nodes simultaneously, it is still limited in its power to create new connections.

Let $t$ be the minimum number of rounds (from that initial configuration at round $1$) taken for node $u$ to acquire the id of node $v$ (or vice versa) in its channel. Let $S = (a_1, a_2, \dots, a_{t} = u)$ be the sequence of nodes from round $1$ to $t$ holding the id of $v$ (until it reached $u$). For ease of exposition, let us consider the class of distributed overlay algorithms that only add new (directed) edges, i.e, do not remove existing edges. If we arrive at a contradiction for this class of algorithms, then it also applies to all other algorithms.

To that end, observe that for this class of algorithms, the ids in $S$ must initially belong to some path $P$ (regardless of directions) from $v$ to $u$. The key observation is that for $i \in [t-1]$, between $a_{i}$ and $a_{i+1}$ in the path $P$, there can be at most $2^{i-1} - 1$ nodes (and $2^{i-1}$ edges), i.e., node $a_{i+1}$ cannot be too far from node $a_{i}$ in the initial communication graph. This is because even if all nodes introduce all ids stored in their respective memories (incl. channels) to each other via messages in all rounds, each node can only learn all ids in its $2^j$-hop neighborhood in $j$ rounds. Thus, the total number of edges (i.e., distance between $v$ and $u$) in this path $P$ is at most $2^t$.

Therefore, we arrive at a contradiction that if there exists a pair of supervisor and self-stabilizing overlay algorithm for which the network reached an almost-legal configuration in $t = o(\log D)$ rounds, then $u$ and $v$ must have been at a distance of at most $2^t < D$ in the initial configuration.
\end{proof}
\newpage
\section{Helpful Figures and Tables} \label{sec:helpful-figs}

\begin{figure}[tbh]
    \centering \includegraphics[width=0.5\textwidth]{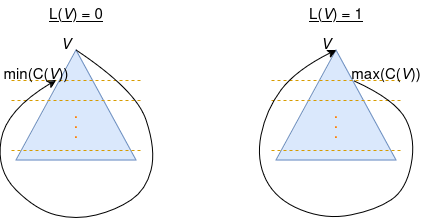}
    \caption{Intuition for the Tree-to-Path algorithm.}
    \label{fig:TtP-intuition}
\end{figure}

\begin{table}[tbh]
\parbox{.45\linewidth}{
\centering
    \begin{tabular}{|l|c|c|}
    \hline
     Variable & Store node ids? \\
    \hline
    $u.L$ & \cmark \\
    $u.R$ & \cmark \\
    $u.\mathit{vID}$ & \xmark \\
    $u.\mathit{flyID}$ & \cmark \\
    \hdashline
    $u.\text{\emph{c-par}}$ & \xmark \\
    $u.\text{\emph{c-dist}}$ & \xmark \\
    $u.\text{\emph{c-ids}}$ & \cmark \\
    \hdashline
    $u.\mathit{exit}$ & \xmark \\
    \hline
    \end{tabular}
    \caption{Flyover variables of a node $u$.}
    \label{tab:fly-vars}
}
\hfill
\parbox{.45\linewidth}{
\centering
    \begin{tabular}{|l|c|c|}
    \hline
     Variable & Store node ids? \\
    \hline
    $u.t$ & \xmark \\
    $u.\mathit{dist}$ & \xmark \\
    \hline
    \end{tabular}
    \caption{Advice variable of a node $u$.}
    \label{tab:adv-vars}
}
\end{table}


\end{document}